\documentclass[11pt,a4paper]{article}
\usepackage{authblk,url}
\usepackage{xcolor}
\usepackage{algorithm}
\usepackage{algpseudocode}

\usepackage{silence}
\usepackage[T1]{fontenc}
\usepackage{lmodern}
\usepackage[utf8]{inputenc}

\usepackage[margin=2.4cm]{geometry}

\usepackage{amsmath,amssymb,amsthm,amsfonts,mathtools}
\usepackage{bm,braket,xspace}
\usepackage{booktabs,array,tabularx}
\usepackage{graphicx,xcolor,xurl,float}

\usepackage[sort&compress]{natbib}
\usepackage[colorlinks=true,linkcolor=blue!35!black,citecolor=blue!35!black,urlcolor=blue!35!black]{hyperref}
\usepackage[nopatch=footnote]{microtype}
\usepackage{setspace}
\allowdisplaybreaks[2]

\usepackage{quantikz}

\newtheorem{theorem}{Theorem}[section]
\newtheorem{proposition}[theorem]{Proposition}
\newtheorem{lemma}[theorem]{Lemma}
\newtheorem{corollary}[theorem]{Corollary}
\theoremstyle{definition}
\newtheorem{definition}[theorem]{Definition}
\newtheorem{example}[theorem]{Example}
\theoremstyle{remark}
\newtheorem{remark}[theorem]{Remark}

\newcommand{\C}{\mathbb C}

\newcommand{\R}{\mathbb R}
\newcommand{\F}{\mathbb F}
\newcommand{\Z}{\mathbb Z}
\newcommand{\cA}{\mathcal A}
\newcommand{\cB}{\mathcal B}

\newcommand{\cF}{\mathcal F}
\newcommand{\cH}{\mathcal H}

\newcommand{\cP}{\mathcal P}
\newcommand{\cQ}{\mathcal Q}
\newcommand{\cR}{\mathcal R}

\newcommand{\ba}{\mathbf a}
\newcommand{\bb}{\mathbf b}

\newcommand{\bv}{\mathbf v}
\newcommand{\bu}{\mathbf u}
\newcommand{\bx}{\mathbf x}
\newcommand{\by}{\mathbf y}
\newcommand{\bz}{\mathbf z}

\newcommand{\supp}{\operatorname{supp}}

\newcommand{\rw}{\operatorname{rw}}
\newcommand{\lrw}{\operatorname{lrw}}
\newcommand{\tw}{\operatorname{tw}}

\newcommand{\fhtw}{\operatorname{fhtw}}

\newcommand{\lline}{L}
\newcommand{\poly}{\operatorname{poly}}
\newcommand{\SOP}{\textnormal{\textsc{SOP}}\xspace}
\newcommand{\Leaves}{\operatorname{Leaves}}
\newcommand{\CZ}{\mathrm{CZ}}
\newcommand{\gX}{X}
\newcommand{\gY}{Y}
\newcommand{\gT}{T}
\newcommand{\gZ}{Z}
\newcommand{\gH}{H}
\newcommand{\gS}{S}
\newcommand{\Toffoli}{\mathrm{Toffoli}}

\newcommand{\CNOT}{\mathrm{CNOT}}
\newcommand{\InsideOut}{\textnormal{\textsc{InsideOut}}\xspace}

\newcommand{\BucketEval}{\textnormal{\textsc{BagEval}}\xspace}

\newcommand{\phases}[1]{[#1]}
\newcommand{\ostar}{\mathbin{\circledast}}
\makeatletter
\newcommand{\bigostar}{\mathop{\mathpalette\bigostar@\relax}\displaylimits}
\newcommand{\bigostar@}[2]{\vcenter{\hbox{%
  \sbox\z@{$\m@th#1\bigotimes$}%
  \resizebox*{!}{\dimexpr\ht\z@+\dp\z@\relax}{$\m@th#1\circledast$}%
}}}
\makeatother

\newcommand{\iv}[1]{\mathbf1\!\left[#1\right]}
\newcommand{\me}[3]{\langle #1\vert #2\vert #3\rangle}

\newcommand{\Hblock}[1]{\mathsf H_{#1}}
\newcommand{\Tblock}[1]{\mathsf T_{#1}}
\newcommand{\CZblock}[1]{\mathsf Z_{#1}}
\newcommand{\tlsw}{\mathsf{tlsw}}
\newcommand{\blockgroup}[1]{
  \gategroup[
    wires=4,
    steps=8,
    style={
      dashed,
      rounded corners,
      fill=blue!4,
      inner xsep=2pt
    },
    background,
    label style={
      label position=below,
      anchor=north,
      yshift=-0.15cm
    }
  ]{#1}
}

\title{Unifying and Extending Strong Simulation of Quantum Circuits}

\author[1]{Floris Geerts}

\author[2]{Rihan Hai}

\author[3]{Matthias Lanzinger}

\author[3]{Reinhard Pichler}

\author[3]{Emanuel Sallinger}

\author[3]{Daniel Unterberger}

\affil[1]{University of Antwerp, Belgium}
\affil[2]{Technical University Delft, The Netherlands}
\affil[3]{TU Wien, Austria}

\date{}

\begin{document}
\maketitle
\begingroup
\renewcommand{\thefootnote}{}
\footnotetext{\noindent Author contact addresses ---
Floris Geerts: \href{mailto:floris.geerts@uantwerp.be}{floris.geerts@uantwerp.be},
Rihan Hai: \href{mailto:r.hai@tudelft.nl}{R.Hai@tudelft.nl},
Matthias Lanzinger: \href{mailto:matthias.lanzinger@tuwien.ac.at}{matthias.lanzinger@tuwien.ac.at},
Reinhard Pichler: \href{mailto:reinhard.pichler@tuwien.ac.at}{reinhard.pichler@tuwien.ac.at},
Emanuel Sallinger: \href{mailto:emanuel.sallinger@tuwien.ac.at}{emanuel.sallinger@tuwien.ac.at},
Daniel Unterberger: \href{mailto:daniel.unterberger@tuwien.ac.at}{daniel.unterberger@tuwien.ac.at}.}
\endgroup

\begin{abstract}
We establish functional aggregate queries (FAQs) as a unifying language
for exact classical simulation of quantum circuits. A circuit defines
a sum--product query: factors encode gates, internal wire variables are
aggregated, and free boundary variables index transition amplitudes.
The central insight is that seemingly different sources of simulation
tractability can be exploited by the same \InsideOut{} evaluation
scheme. The query specifies what is computed; the evaluation plan,
semiring, and representation of intermediate factors determine how
efficiently it can be computed.

This perspective unifies structural and algebraic simulation guarantees.
With explicit factor representations, FAQ evaluation recovers the
treewidth bound for tensor-network contraction and yields finer,
sparsity-sensitive bounds based on fractional covers. Over a formal
phase semiring, compressed intermediate factors recover rank-width-based
simulation for circuits with compatible quadratic phase representations.
For Clifford circuits, affine--quadratic factors are closed under
multiplication and marginalization, and their representations remain
polynomial in size throughout evaluation. Consequently, the same
evaluation scheme computes exact computational-basis amplitudes in
polynomial time, without any bounded-width assumption.

Beyond recovering these guarantees, the FAQ framework yields a new
tractability criterion: \emph{tensor layout symmetry width}.
This parameter combines local cut-rank with separator symmetry through
exact tree-tensor representations. We give a constructive evaluation
bound and exhibit a circuit family with bounded tensor layout symmetry
width but unbounded phase-graph rank-width and circuit line-graph
treewidth. Together, these results establish representation-aware FAQ
evaluation as a common algorithmic foundation for classical simulation
and a systematic route to new tractable regimes.
\end{abstract}

\begingroup
\begin{spacing}{0.7}
\tableofcontents
\end{spacing}
\endgroup
\clearpage

\section{Introduction}
\label{sec:intro}

Classical quantum circuit simulation, that is, simulating quantum circuits through classical computation, remains an important problem in quantum computing. 
They are also used for calibrating, benchmarking, and validating quantum hardware \citep{villalonga2019flexible,%
ZulehnerWille2018,%
DBLP:conf/ijcar/MeiCBL24,%
DBLP:journals/pacmpl/AbdullaCHHLLT26,%
DBLP:journals/pacmpl/AbdullaCCHLLLT25,%
DBLP:journals/cacm/ChenCLLTY25}. 
Second, in the quest for showing ``quantum supremacy'' (i.e., identifying computational tasks
that a quantum computer can perform beyond the practical capabilities of classical computers), 
the efficient simulation of quantum circuits highlights the limits of quantum supremacy~\citep{Boixto2018,Arute2019,%
Singh2025,%
Shchesnovic2021,%
HarrowMontalee2017}. 
Improvements in classical simulation therefore continually sharpen, and sometimes substantially shift, the boundary at which such a quantum advantage can be claimed~\citep{PanChenZhang22}.

In this paper, we show that several major sources of tractability in classical quantum circuit simulation admit a common algorithmic explanation through functional aggregate queries (FAQs) \citep{AboKhamisNgoRudra2016}. Our central observation is that a quantum circuit can be represented as a sum--product query whose factors encode gates and whose aggregated variables encode intermediate wire states. Several seemingly different simulation techniques can then be understood as instances of the same evaluation paradigm, differing in the semiring and representation used for intermediate factors. Before developing this framework in detail, we first review several prominent sources of classical tractability, including stabilizer structure, bounded treewidth, rank-width, and compact tensor or decision-diagram representations. We then show how these viewpoints fit into the FAQ framework, recover and sharpen existing simulation bounds, and use the same perspective to identify new tractable regimes.

\smallskip
\noindent
\textbf{Overview of tractable regimes.}
Several important classes of quantum circuits admit efficient classical simulation.
The seminal work of \citet{Gottesman1998}, culminating in the Gottesman--Knill theorem, and its extension \citep{AaronsonGottesman2004}, are famous examples.
These methods track the stabilizer structure of Clifford circuits, employing tableau representations, and yield polynomial-time simulation.

For general quantum circuits, however, classical simulation is computationally intractable in the worst case; in particular, computing exact output probabilities is $\#P$-hard \citep{VandenNest10}. 
A popular approach of dealing with intractability comes from the field of parameterized algorithms and 
complexity~\citep{DBLP:books/sp/CyganFKLMPPS15}. 
Here, one aims to identify structural properties of the problem instances that determine the hardness of the problem. 
These structural properties are captured by a structural ``parameter''  (typically a numerical measure of the input structure), with the goal of establishing an upper bound on the complexity that may be exponential in the parameter but not in the size of the 
given problem instance. Such a bound immediately gives polynomial-time algorithms on every class of
instances for which 
the parameter is bounded by a constant. 

This perspective was brought to quantum circuit simulation in the pioneering work of~\citet{MarkovShi2008}. 
They showed that tensor-network contraction yields a simulation algorithm whose running time is 
 exponential in the {\em treewidth}\footnote{All notions used here will be formally defined in subsequent sections.} 
of an associated graph representation of the circuit, while remaining polynomial in the circuit size itself. 
Consequently, any class of quantum circuits whose associated graph representations have bounded treewidth admits polynomial-time classical simulation.

A different structural route emerged from the \emph{sum-of-powers} (\SOP) representation of circuit amplitudes.
\citet{WangEtAl2025} introduced \textsc{FeynmanDD},
which represents the Feynman path sum of a circuit by a classical decision diagram. Subsequent work by
\citet{Cheng+2025} showed that the size of the resulting
decision diagram can be bounded in terms of the \emph{linear
rank-width} of the corresponding interaction graph.
For the important special case of \emph{quadratic} \SOP
representations,  \citet{deColnet+2026} strengthened this dependence from linear rank-width to \emph{rank-width}.
The result applies, in particular, to circuits over gates
such as 
single-qubit rotations (e.g., $\gS$ and  $\gT$ gates) as well as Hadamard, $\CZ$, etc. These approaches are applicable to various choices of universal gate sets such as, for instance, $\{\gH, \gT, \CZ\}$.  Thus, as with the treewidth bound above, bounded rank-width yields polynomial-time classical simulation.

Beyond treewidth- and rank-width-based methods, classical simulation has been approached through a variety of representations. Matrix product states and tree tensor networks exploit low bond dimension~\citep{Vidal2003,shi2006classical,seitz2023simulating}; projected entangled-pair states extend tensor-network methods to more general geometries~\citep{guo2019general,lee2025scalable}; and decision diagrams exploit repeated structure in quantum states and circuit path sums~\citep{VinkhuijzenEtAl2023,WangEtAl2025}. Recent work has also related structural graph parameters to ZX-diagram contraction and stabilizer decompositions~\citep{KuyanovKissinger2026,CodsiLaakkonen2026}.

Despite their different formulations, these approaches address a common algorithmic question: \emph{which intermediate objects arising during simulation can be represented and manipulated compactly?}
We study this question
through the formalism of \emph{functional aggregate queries} (FAQs). A quantum circuit gives rise to a
sum-product query; an evaluation plan determines the intermediate factors that produces
messages, while the representation chosen for these factors determines the cost of evaluation. This perspective provides a common language for several existing simulation methods, recovers known structural bounds, and yields our
new tensor-layout  symmetry-width bound, including circuit families for which the competing treewidth and rank-width parameters are unbounded.

\smallskip
\noindent
\textbf{A unified approach}.
The functional aggregate query (FAQ) formalism \citep{AboKhamisNgoRudra2016} generalizes ``sum of product'' expressions through queries of the form: 
\begin{equation}
Q(X) := \bigoplus_{Y}\bigotimes_i F_i(X_i,Y_i), \label{eq-faq1}
\end{equation}
where $X$ and $Y$ are disjoint sets of free and aggregated variables, respectively, all ranging over finite domains. Here,
 $X_i \subseteq X$, $Y_i \subseteq Y$, $\bigcup_i  X =  X$, and
$\bigcup_i Y_i = Y$. 
The unifying power of FAQs comes from the fact that,when evaluated,  the sum $\bigoplus$ and product $\bigotimes$ are operations in 
an arbitrary commutative semiring. By considering different semirings,
different problems can be solved. FAQs have been used for bringing together results from a variety of areas 
such as databases, constraint satisfaction, graphical models, and matrix and tensor computation, see e.g., \citet{AboKhamisNgoRudra2016} and references therein.

We show that the behaviour of quantum circuits can be expressed directly as FAQs. In particular, the {\em factors} $F_i$ can be used to encode quantum gates and the variables used as 
arguments in a factor represent the input and output wires of that gate. The sequence of quantum gates in a circuit is 
captured by the product $\bigotimes$ of the factors corresponding to these gates.
Summation $\bigoplus$ allows us to aggregate away variables whose concrete value is not 
of interest. In particular, this applies to variables representing intermediate states of qubits 
that result from the 
application of one quantum gate and serve as input to the next gate. 

For a quantum circuit $C$ implementing a unitary matrix $U_C$, this construction yields an FAQ $Q_C$ whose free variables
correspond to the circuit's input and output boundaries.
Evaluated over the field $\C$ with the gate factors described
above, it satisfies 
\[
    Q_C(\ba,\bb)=\braket{\bb | U_C | \ba}
\]
for every computational-basis input $\ba$ and output $\bb$.
Computing this amplitude for specified boundary assignments
is a standard amplitude-based strong-simulation task. Most previous works concentrate on this simulation task.

The FAQ formulation is not restricted to computing individual
amplitudes. FAQs are very flexible and they allow us to model many different scenarios 
simply by considering certain variables of the FAQ 
as free and aggregating away all the others.
For instance, we can express the relationship between all possible initial states 
and all possible end states of a quantum circuit $C$ 
by 
considering the variables corresponding to initial and end states as ``free'' 
(i.e., $X$ in Equation~(\ref{eq-faq1})) and 
aggregating with $\bigoplus$ over all other variables. The support of the resulting FAQ gives all 
feasible combinations of initial and end states; 
the annotation of each tuple from the support by a value
from the semiring (in our case, by a complex number) gives the corresponding probability amplitude.

Key to the usefulness of FAQs is that an  algorithm called \InsideOut \citep{AboKhamisNgoRudra2016}
exists for evaluating them efficiently over any commutative semiring. The algorithm is a variation of traditional variable elimination algorithms using dynamic programming~\citep{Dechter1999}. 
Also \InsideOut works by variable elimination -- in this case, eliminating those variables of the FAQ
which are aggregated away by the summation. More precisely, when eliminating a particular variable $z$, 
the algorithm replaces the factors containing $z$ 
by a ``message'' $M$ which takes the product of these factors and then sums over the possible values of 
$z$ (when modelling qubits, the possible values of $z$ are 0 and 1).

The structure of an FAQ is nicely captured by a hypergraph, whose vertices correspond to the variables of the FAQ
and hyperedges correspond to the sets of variables jointly occurring as arguments of a factor. 
\cite{DBLP:journals/talg/GroheM14} introduced 
the {\em fractional hypertree width} ($\fhtw$) as an 
important parameter of hypergraphs, which is used to upper bound various tasks in query evaluation and constraint solving. 
Also the complexity of the \InsideOut algorithm is governed by the $\fhtw$ of the 
(hypergraph underlying an) FAQ.
More specifically, evaluating
an FAQ requires time exponential in the fractional hypertree width but not in the size of the FAQ. 

The pay-off of modelling quantum circuits by FAQs 
comes from the fact that three central connections can be made. In all three cases, 
we use the same FAQ to model the quantum circuit and we evaluate the FAQ with the 
\InsideOut algorithm. The difference comes from 
considering different semirings.

\begin{itemize}
    \item {\em Recovering the treewidth bound of} \cite{MarkovShi2008}. In fact, we get this result ``for free'' 
    by considering the FAQ-representation of quantum circuits over the complex numbers and observing 
    that the fractional hypertree-width of (the hypergraph underlying) an FAQ is 
    bounded from above by the treewidth of the circuit graph plus 1.
    And, actually, $\fhtw$ can be significantly smaller than $\tw$.
    
    \item {\em Recovering the rank-width bound akin to} \cite{WangEtAl2025,Cheng+2025,deColnet+2026}. 
    Towards this result, we introduce the {\em phase semiring} $K_r\coloneqq\mathbb R[\Z_r]$ for some 
    natural number $r$. We then exploit the flexibility of the FAQ framework and consider FAQs over $K_r$.
    This allows us to obtain the upper bound of     \cite{WangEtAl2025,Cheng+2025,deColnet+2026}
    in terms of the (linear) rank-width 
    via FAQ evaluation.
    As an additional benefit, the reconstruction of the rank-width bound is technically much simpler and more compact 
    than the original construction. 
    
\item {\em Recovering the exact-amplitude consequence of the Gottesman--Knill theorem}:
For Clifford circuits, tractability does not rely on a bounded structural width of the circuit. Instead, the intermediate messages produced by \InsideOut{} admit compact representations. More precisely, we represent the gate factors and all intermediate messages as \emph{affine--quadratic} (AQ) factors with values in the ring $\mathbb Z[i,2^{-1/2}]$. Crucially, this representation is preserved under the multiplication and aggregation operations performed by \InsideOut{}, and remains polynomial in size throughout the evaluation. Consequently, \InsideOut{} evaluates the FAQ of a Clifford circuit in polynomial time, thereby recovering the exact-amplitude consequence of the Gottesman--Knill theorem within the FAQ framework.
\end{itemize}

\noindent
That is, many existing important results on strong simulation in quantum circuits can be seen as instances of the same algorithm simply evaluated over different semirings encoding the gate behaviour. 
By unifying and conceptually simplifying these results in this form, we are also able to extend on them significantly. As already mentioned, in the setting of \cite{MarkovShi2008}, we obtain stronger bounds on the exponents of simulation algorithms. However, we also obtain an entirely new structural condition for tractable simulation, roughly speaking, by combining circuit parts with low rank-width, with a global measure that controls their interaction. 
\begin{itemize}
\item
We introduce a new parameter we call
{\em tensor layout symmetry width} ($\tlsw$) and show that it identifies tractable fragments for strong simulation and is strictly more general than both the treewidth approach of \cite{MarkovShi2008}, and the later rank-width methods of  \cite{WangEtAl2025,Cheng+2025,deColnet+2026}. 
\end{itemize}
Again the algorithm is simply \InsideOut, but over a more complex semiring on what we call \emph{phase coefficient tenors}.  Moreover, we exhibit a class of quantum circuits, such that the treewidth of the 
graph representation
as well as the rank-width
are unbounded but the $\tlsw$ is at most 2. This family shares the alternating interaction-and-mixing structure of QAOA and collective-spin simulation circuits. It thus also serves to illustrate  how $\tlsw$ structurally captures efficient strong simulation based on intricate interplay of structure and symmetry, beyond the guarantees of treewidth and rank width.

\smallskip
\noindent
\textbf{Structure of the paper}. In Section~\ref{sec:circuit-faq} we recall basic notions of quantum circuits and 
we formally introduce FAQs to model quantum circuits.
In Section~\ref{sec:FAQs}, we introduce the \InsideOut algorithm in technical detail. 
In Sections~\ref{sec:explicit-tables} and \ref{sec:compressed-boundaries},
we first recover the treewidth bound of \cite{MarkovShi2008} and then 
the rank-width bound of  \cite{WangEtAl2025,Cheng+2025,deColnet+2026}. 
Our new parameter tlsw and an enhanced quantum circuit simulation algorithm 
based on this parameter are presented in Section~\ref{sec:sopw}.
The  Gottesman-Knill theorem is recovered in 
Section~\ref{sec:gk-faq}. 
We provide an overview of related work in Section~\ref{sec:related-work}
and we conclude with Section \ref{sec:conclusion}.

\section{Quantum circuits as functional aggregate queries}
\label{sec:circuit-faq}
Classical simulation of quantum circuits can be formulated as the evaluation
of \emph{functional aggregate queries} (FAQs)
\citep{Dudek+2020,Geerts+2025}.  The FAQ formalism \citep{AboKhamisNgoRudra2016} represents sum--product
expressions over a commutative semiring, which captures the common structure of
factor graphs and tensor networks \citep{Kschischang+2001,Robeva+2019}.
The semiring specifies how factor values are combined, while the factor
scopes specify which variables interact.  This separates the structure of
the query from the algebra of its values and the representation of its
factors.  We first introduce FAQs, encode quantum circuits, and show how
support constraints can simplify the encoding. 

\subsection{Functional aggregate queries}
\label{subsec:faq}

A \emph{commutative semiring} is a tuple
$\mathbb K=(K,\oplus,\otimes,0,1)$ in which $(K,\oplus,0)$ and
$(K,\otimes,1)$ are commutative monoids, multiplication distributes over
addition, and $0\otimes a=0$ for every $a\in K$.  Commutative rings and fields,
including $\C$, are special cases.

Let $X$ be a finite set of variables, each variable $x\in X$ equiped with a finite nonempty domain
$D_x$.  For $S\subseteq X$, we write $D_S\coloneqq\prod_{x\in S}D_x$.
An element $\mathbf a\in D_S$ can also be interpreted as an \emph{assignment}: it associates a value
$a_x\in D_x$ with each $x\in S$.  For $T\subseteq S$, we write
$\mathbf a_T$ for its restriction to $T$.  By convention, $D_\emptyset$
contains the unique empty assignment.

A \emph{factor} with \emph{scope} $S\subseteq X$ is a function
$F:D_S\to K$, and equivalently a $K$-valued tensor whose indices are the
variables in $S$.  For a factor $F$, we denote its scope by $S_F$ and its \emph{support} by $\supp(F)\coloneqq\{\mathbf a\in D_{S_F} \mid F(\mathbf a)\neq0\}$.
Given a finite indexed collection $\cF$ of factors, allowing repetitions,
and a set $Y\subseteq X$ of \emph{free variables}, the associated
\emph{functional aggregate query} is the factor with scope $Y$ defined by
\begin{equation}
\label{eq:faq}
    Q(\mathbf x_Y)
    \coloneqq
    \bigoplus_{\mathbf x_{X\setminus Y}\in D_{X\setminus Y}}
        \bigotimes_{F\in\cF}F(\mathbf x_{S_F}).
\end{equation}
For each fixed assignment $\mathbf x_Y$, the expression ranges over all assignments $\mathbf x_{X\setminus Y}$ to $X\setminus Y$.  Together, these assignments determine
$\mathbf x\in D_X$, which is restricted to each factor's scope when that
factor is evaluated.  The variables in $X\setminus Y$ are called
\emph{bound variables}: they are aggregated out using $\oplus$.  The free
variables index the resulting factor.  When $Y=\emptyset$, the answer is a single element of $K$.

The \emph{structural hypergraph} of the FAQ is
$\cH_Q\coloneqq(X,\{S_F\mid F\in\cF\})$.  It records which variables occur together in
factors, independently of their values and representations.  We distinguish
its \emph{structural size},
$|Q|\coloneqq|X|+|\cF|+\sum_{F\in\cF}|S_F|$, from the total input encoding length 
$\|Q\|_{\mathrm{in}}$.  The latter includes the variable domains, scopes,
and factor representations.  It therefore depends on whether factors are
stored as dense tables, sparse listings, or compact symbolic objects,
whereas $|Q|$ and $\cH_Q$ do not.  The \emph{primal graph} of $\cH_Q$ has
vertex set $X$, with two distinct variables adjacent precisely when they
occur together in some factor scope.

\subsection{Quantum circuits and simulation}
\label{subsec:circuit-encoding}

We use Dirac notation: a vector $v\in\C^d$ is written as a \emph{ket}
$\ket v$, and its conjugate transpose $v^\dagger$ as the corresponding
\emph{bra} $\bra v$.  Their pairing
$\braket{w|v}=w^\dagger v=\sum_i\overline{w_i}v_i$ is the standard
Hermitian inner product on $\C^d$.
Let $\cQ$ be a finite set of \emph{qubits}.  For each $q\in \cQ$, let
$\mathcal H_q\cong\C^2$ be the Hilbert space of qubit $q$, with
orthonormal computational basis $\{\ket{0}_q,\ket{1}_q\}$.  For
$S\subseteq \cQ$, define
$\mathcal H_S\coloneqq\bigotimes_{q\in S}\mathcal H_q$, with
$\mathcal H_\emptyset\coloneqq\C$.  Its computational basis is indexed by
assignments $\bb\in \{0,1\}^S$, with
$\ket{\bb_S}\coloneqq\bigotimes_{q\in S}\ket{b_q}$.
The tensor factors are labelled by their qubits and are understood up to
canonical reordering.  If $S=\{q_1<\cdots<q_k\}$ is ordered, we may display
$\bb$ as the bit string $b_{q_1}\cdots b_{q_k}$.  For example, $\ket{0001}$ on
$S=\{1,2,3,4\}$ denotes the basis vector indexed by the assignment $b$ with
$b_1=b_2=b_3=0$ and $b_4=1$.  When $S=\cQ$, we usually omit the subscript
from $\ket{\bb  }_S$.  At times, we take $\cQ=[n]\coloneqq\{1,2,\ldots,n\}$. In this case, we use the notations $\{0,1\}^{[n]}$ and $\{0,1\}^n$ interchangeably.

A \emph{quantum gate} $G$ supported on $S\subseteq \cQ$ is specified by a
unitary operator $U_G:\mathcal H_S\to\mathcal H_S$.  Its extension to the
full state space is $\widehat U_G\coloneqq U_G\otimes I_{\cQ\setminus S}$, under the
canonical identification
$\mathcal H_{\cQ}\cong\mathcal H_S\otimes\mathcal H_{\cQ\setminus S}$.
Here, $I_{\cQ\setminus S}$ is the identity on
$\mathcal H_{\cQ\setminus S}$.  Thus, $\widehat U_G$ acts as $U_G$ on the
qubits in $S$ and as the identity on all other qubits.  We say that $G$
\emph{touches} the qubits in $S$.

The example gates used below are specified by their action on computational-basis
states; linearity determines their action on every state.  In these
formulas, we write the gate symbol for its local unitary.  The variables
$u,v,w$ are bits, and $\oplus$ denotes addition modulo two. 

\begin{center}
\small
\begin{tabular}{@{}ll@{}}
\toprule
Gate & Action on basis states\\
\midrule
Pauli-$X$ $\gX$
    & $\gX\ket u=\ket{u\oplus 1}$\\
Pauli-$Y$ $\gY$
    & $\gY\ket u=i(-1)^u\ket{u\oplus1}$\\ 
Pauli-$Z$ $\gZ$
    & $\gZ\ket u=(-1)^u\ket u$\\
Hadamard $\gH$
    & $\gH\ket u=2^{-1/2}(\ket0+(-1)^u\ket1)$\\
Phase $\gS$
    & $\gS\ket u=i^u\ket u$\\
Phase $\gT$
    & $\gT\ket u=e^{i\pi u/4}\ket u$\\
Controlled-$Z$ $\CZ$
    & $\CZ\ket{u,v}=(-1)^{uv}\ket{u,v}$\\
Controlled-NOT $\CNOT$
    & $\CNOT\ket{u,v}=\ket{u,u\oplus v}$\\
Toffoli
    & $\Toffoli\ket{u,v,w}=\ket{u,v,w\oplus uv}$\\
\bottomrule
\end{tabular}
\end{center}

The first input of $\CNOT$ is the control and the second is the target;
Toffoli has two controls followed by its target.  Subscripts specify the
qubits in these roles: for example, $\gH_q$ acts on $q$, and
$\CNOT_{p,q}$ has control $p$ and target $q$.
The gate $\CZ_{p,q}$ acts on qubits $p$ and $q$ and multiplies the basis
state by $-1$ exactly when both qubits are in state $\ket1$; since
$\CZ$ is symmetric in its two inputs, neither qubit has a distinguished
control or target role.  The one-qubit identity gate
$\mathsf{I}$ satisfies $\mathsf{I}\ket u=\ket u$.

A \emph{quantum circuit} on $\cQ$ is a sequence $C=(G_1,\ldots,G_T)$, where
$G_t$ is supported on $S_t\subseteq \cQ$.  The gates are applied in the listed
order, starting with $G_1$, so the circuit implements
\[
    U_C=\widehat U_{G_T}\circ\cdots\circ\widehat U_{G_1}
    \colon\mathcal H_{\cQ}\to\mathcal H_{\cQ}.
\]
We measure circuit size by
$|C|\coloneqq|\cQ|+\sum_{t=1}^T|S_t|$, so the qubits and
all gate incidences are included.
For $q\in \cQ$, let $t_{q,1}<\cdots<t_{q,T_q}$ be the indices of the gates
that touch $q$.  The \emph{wire} of $q$ is the ordered chain from its input
boundary, through these gate occurrences, to its output boundary.  The
connections between consecutive elements are its $T_q+1$ \emph{wire
segments}.  An untouched qubit has a single segment connecting its input
and output boundaries.

\paragraph{Circuit blocks.}
rather than as a single gate.  We allow blocks within circuit expressions,
with the convention that they are expanded in place, from left to right.
Thus, for circuits $C$ and $D$ on the same qubit set, $(C,D)$ denotes
the concatenation of their gate sequences, executing $C$ first and $D$
second, so that $U_{(C,D)}=U_DU_C$.
Write $C^{[m]}$ for $m$ consecutive copies of $C$, with $C^{[0]}$
the empty sequence, which implements the identity.
Each use of a block introduces fresh occurrences of its constituent
gates.  Circuit size, wire segments, and the associated tensor networks
and hypergraphs are always defined using the expanded gate sequence. On the qubit set $\cQ=[n]$, we use the blocks
\begin{equation}
\begin{aligned}
    \Hblock{n}&:=(H_1,\ldots,H_n),&
    \Tblock{n}&:=(T_1,\ldots,T_n),\\
    \CZblock{n}&:=(\CZ_{i,j})_{1\leq i<j\leq n},
\end{aligned}
\label{eq:standard-circuit-blocks}
\end{equation}
where the pairs $(i,j)$ in $\CZblock{n}$ are listed in lexicographic
order.  For example,
$C=(\Hblock{n},\Tblock{n},\CZblock{n})$
denotes the circuit consisting of the $n$ Hadamard gates, followed by
the $n$ $T$ gates, followed by the $\binom{n}{2}$ controlled-$Z$ gates
in the specified order.  This order is retained when defining wire
segments and the raw tensor network, even when gates commute.

\paragraph{Strong-simulation.}
We take the basic \emph{strong-simulation} task to be the exact computation
of a \emph{transition amplitude} $\me{\by}{U_C}{\bx}\in\C$, where
$\bx,\by\in\{0,1\}^{\cQ}$ are the input and output \emph{boundary assignments}.
This amplitude is the coefficient of $\ket \by$ in $U_C\ket \bx$, and
$\lvert\me{\by}{U_C}{\bx}\rvert^2$ is the probability of observing $\by$ when that
state is measured in the computational basis.  A simple special case fixes
the input to the all-zero assignment $\mathbf{0}_{\cQ}$ and the output to a one-hot assignment $\mathbf{e}^{q_0}$, where $q_0\in \cQ$,
$e_{q_0}^{q_0}\coloneqq 1$, and $e_q^{q_0}\coloneqq 0$ for $q\neq q_0$.
We allow arbitrary input and output assignments.

\subsection{Quantum circuits as FAQs}
\label{subsec:circuits-as-faqs}

Let $C=(G_1,\ldots,G_T)$ be a circuit on $\cQ$, with gate $G_t$ supported on
$S_t\subseteq\cQ$.  For the encoding as FAQs, we may add an identity gate to each otherwise
untouched qubit.  This leaves $U_C$ unchanged and ensures that the input
and output boundaries correspond to distinct segment variables.

Using the wire notation above, associate the Boolean variables
$z^q_0,\ldots,z^q_{T_q}$ with the successive segments of qubit $q$.
The variable $z^q_0$ is the input index, $z^q_{T_q}$ is the output index,
and $z^q_{j-1},z^q_j$ are the indices immediately before and after
$G_{t_{q,j}}$.  These variables index computational-basis states in the
expansion; they do not assert that a qubit has a definite classical value
along its wire.  Every segment variable has domain $\{0,1\}$.
Let $X$ be the set of all segment variables, and let
$B_{\mathrm{in}}\coloneqq\{z^q_0 \mid q\in \cQ\}$ and
$B_{\mathrm{out}}\coloneqq\{z^q_{T_q}\mid q\in \cQ\}$ denote the input and output boundary variables.

Each gate occurrence contributes one factor.  For each $q\in S_t$, let
$j_q(t)$ be the unique index satisfying $t=t_{q,j_q(t)}$, and set
\[
X_t
\coloneqq
\bigl\{
    z^q_{j_q(t)-1},z^q_{j_q(t)}
    \bigm|
    q\in S_t
\bigr\}.
\]
Thus, $X_t$ consists of the input and output segment variables of the qubits
on which $G_t$ acts, and $|X_t|=2|S_t|$.  Every assignment
$\pmb{\alpha}\in \{0,1\}^{X_t}$ induces input and output assignments
$\pmb\alpha_t^{\mathrm{in}}$ and $\pmb\alpha_t^{\mathrm{out}}$ in $\{0,1\}^{S_t}$ given by
\[
\alpha_t^{\mathrm{in}}(q)
    \coloneqq \alpha\bigl(z^q_{j_q(t)-1}\bigr),
\qquad
\alpha_t^{\mathrm{out}}(q)
    \coloneqq \alpha\bigl(z^q_{j_q(t)}\bigr).
\]
The gate factor $F_t:\{0,1\}^{X_t}\to\C$ is then
\[
F_t(\pmb\alpha)
\coloneqq
\me{\pmb\alpha_t^{\mathrm{out}}}{U_{G_t}}
   {\pmb\alpha_t^{\mathrm{in}}}.
\]
Figure~\ref{fig:gate-factors} gives the resulting factors for several
standard gates. We thus obtain factors with values in $\C$ and hence $\oplus$ and $\otimes$ represent the usual sum and product of the complex numbers.

\begin{figure}[t]
\centering
\small
\setlength{\arraycolsep}{3pt}
\renewcommand{\arraystretch}{1}

\[
\begin{array}{cc|c}
  \multicolumn{3}{c}{H}\\
  i_1&o_1&F\\ \hline
  0&0&2^{-1/2}\\
  0&1&2^{-1/2}\\
  1&0&2^{-1/2}\\
  1&1&-2^{-1/2}
\end{array}
\qquad
\begin{array}{cc|c}
  \multicolumn{3}{c}{T}\\
  i_1&o_1&F\\ \hline
  0&0&1\\
  1&1&e^{i\pi/4}
\end{array}
\qquad
\begin{array}{cc|cc|c}
  \multicolumn{5}{c}{\CZ}\\
  i_1&i_2&o_1&o_2&F\\ \hline
  0&0&0&0&1\\
  0&1&0&1&1\\
  1&0&1&0&1\\
  1&1&1&1&-1
\end{array}
\]

\vspace{-0.6em}

\[
\begin{array}{cc|cc|c}
  \multicolumn{5}{c}{\CNOT}\\
  i_1&i_2&o_1&o_2&F\\ \hline
  0&0&0&0&1\\
  0&1&0&1&1\\
  1&0&1&1&1\\
  1&1&1&0&1
\end{array}
\qquad\qquad
\begin{array}{ccc|ccc|c}
 \multicolumn{7}{c}{\text{Toffoli}}\\
  i_1&i_2&i_3&o_1&o_2&o_3&F\\ \hline
  0&0&0&0&0&0&1\\
  0&0&1&0&0&1&1\\
  0&1&0&0&1&0&1\\
  0&1&1&0&1&1&1\\
  1&0&0&1&0&0&1\\
  1&0&1&1&0&1&1\\
  1&1&0&1&1&1&1\\
  1&1&1&1&1&0&1
\end{array}
\]

\vspace{-0.8em}
\caption{Gate encodings as factors over $\C$. Here $i_k$ and $o_k$
denote the input and output variables of the $k$th qubit, respectively.
Only assignments with nonzero factor value are shown.}
\label{fig:gate-factors}
\end{figure}

Fixed boundary assignments are imposed by unary indicator factors
$\iv{z^q_0=x(q)}$ and $\iv{z^q_{T_q}=y(q)}$, one for each boundary
variable.  Here, $\iv{P}$ equals $1$ if the condition $P$ holds and $0$
otherwise.  Let $\cF_C$ consist of the gate factors and these
boundary factors.  We write $Q_C$ for the scalar circuit FAQ, with the
chosen computational-basis input and output assignments understood.
With all segment variables bound, it is
\begin{equation}
\label{eq:amp-faq}
    Q_C
    \coloneqq
    \bigoplus_{\pmb\alpha\in\{0,1\}^{X}}
        \bigotimes_{F\in\cF_C}F(\pmb\alpha_{S_F})
    =
    \me{\by}{U_C}{\bx}.
\end{equation}
The equality follows by repeated tensor multiplication: summing over the
internal indices contracts adjacent gate factors, while the boundary
factors select $\bx$ and $\by$.  This is an FAQ over
$(\C,+,\times,0,1)$ with no free variables.  Its structural hypergraph
contains the gate scopes $X_t$ and the singleton scopes of the boundary
factors.  Such encodings also underlie factor-graph and graphical-model
approaches to tensor contraction
\citep{Kschischang+2001,Robeva+2019,BoixoEtAl2017}.

A key benefit of the FAQ viewpoint is that the same gate factors support
several simulation tasks, as shown in the table below:
\begin{center}
\small
\begin{tabular}{@{}lll@{}}
\toprule
Boundary factors fix & Free boundary variables & Query value\\
\midrule
Input $\bx$ and output $\by$
    & $\emptyset$ & $Q_C=\me{\by}{U_C}{\bx}$\\
Input $x$
    & $B_{\mathrm{out}}$ & $Q(\by)=\me{\by}{U_C}{\bx}$\\
Output $y$
    & $B_{\mathrm{in}}$ & $Q(\bx)=\me{\by}{U_C}{\bx}$\\
Neither
    & $B_{\mathrm{in}}\cup B_{\mathrm{out}}$
    & $Q(\bx,\by)=\me{\by}{U_C}{\bx}$\\
\bottomrule
\end{tabular}
\end{center}
In the last three rows, $Q$ denotes the query with the indicated free
variables.  These queries give, respectively, the coefficients of
$U_C\ket \bx$, the matrix row $\bra \by U_C$, and the full matrix of
$U_C$ in the computational basis.

\begin{example}
\label{ex:qctofaq}
Consider the circuit on $\cQ=\{1,2,3\}$ with gate sequence
$C=(H_1,H_2,H_3,\allowbreak \CNOT_{1,2},T_2,\allowbreak\CZ_{2,3},H_2)$.
The initial Hadamard gates act on disjoint qubits and may be applied in
parallel.  We fix the input $\bx=(x_1,x_2,x_3)$ and leave the outputs
$y_1,y_2,y_3$ free.  The diagram is read from left to right, and its
segment labels are the variables used below.
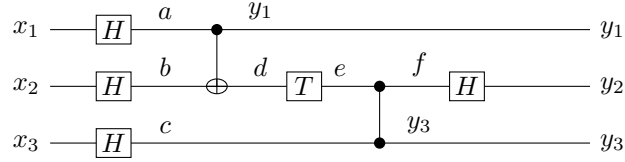
\begin{center}
\begin{tikzpicture}[
    x=1.05cm,
    y=0.75cm,
    baseline=(current bounding box.center),
    gate/.style={draw,minimum width=0.46cm,minimum height=0.38cm,inner sep=1pt},
    dot/.style={circle,fill,inner sep=1.5pt},
    every node/.style={font=\small}
]
    \draw (0,0) -- (6.8,0);
    \draw (0,-1) -- (6.8,-1);
    \draw (0,-2) -- (6.8,-2);

    \node[left] at (0,0) {$x_1$};
    \node[left] at (0,-1) {$x_2$};
    \node[left] at (0,-2) {$x_3$};

    \node[gate,fill=white] at (0.8,0) {$H$};
    \node[gate,fill=white] at (0.8,-1) {$H$};
    \node[gate,fill=white] at (0.8,-2) {$H$};

    \node[above] at (1.45,0) {$a$};
    \node[above] at (1.45,-1) {$b$};
    \node[above] at (1.45,-2) {$c$};

    \node[dot] at (2.1,0) {};
    \draw (2.1,0) -- (2.1,-1);
    \draw (2.1,-1) circle (0.13);
    \draw (1.97,-1) -- (2.23,-1);
    \draw (2.1,-1.13) -- (2.1,-0.87);

    \node[above] at (2.65,0) {$y_1$};
    \node[above] at (2.65,-1) {$d$};

    \node[gate,fill=white] at (3.2,-1) {$T$};
    \node[above] at (3.65,-1) {$e$};

    \node[dot] at (4.15,-1) {};
    \node[dot] at (4.15,-2) {};
    \draw (4.15,-1) -- (4.15,-2);

    \node[above] at (4.65,-1) {$f$};
    \node[above] at (4.65,-2) {$y_3$};

    \node[gate,fill=white] at (5.25,-1) {$H$};

    \node[right] at (6.8,0) {$y_1$};
    \node[right] at (6.8,-1) {$y_2$};
    \node[right] at (6.8,-2) {$y_3$};
\end{tikzpicture}
\end{center}
The variables $a,b,c$ occur after the initial Hadamard layer.  The
$\CNOT$ has inputs $(a,b)$ and outputs $(y_1,d)$, with $y_1$ already
the final output of qubit~$1$.  The $T$ gate has input $d$ and output $e$.
The $\CZ$ has inputs $(e,c)$ and outputs $(f,y_3)$, with $y_3$ already
the final output of qubit~$3$.  The final Hadamard has input $f$ and
output $y_2$.

For compactness, we write $G(\bv;\bu)\coloneqq\me{\bv}{U_G}{\bu}$ for a gate factor,
listing output bits before the semicolon and input bits after it.
Substituting the fixed input values into the initial Hadamard factors gives
\[
\begin{aligned}
    Q(y_1,y_2,y_3)
    =
    \sum_{a,b,c,d,e,f\in\{0,1\}}
      &H(a;x_1)H(b;x_2)H(c;x_3)\\
      &\cdot\CNOT(y_1,d;a,b)\,
       T(e;d)\,\CZ(f,y_3;e,c)\,H(y_2;f).
\end{aligned}
\]
Each gate contributes one factor, the internal variables are summed out,
and the output variables remain free.  The resulting coefficient is
$Q(y_1,y_2,y_3)=\me{y_1y_2y_3}{U_C}{x_1x_2x_3}$.
\end{example}

Partial boundary choices work in the same way.  Fix the input $\bx$ and an
output assignment $\ba\in\{0,1\}^S$ for some $S\subseteq \cQ$, leaving the
outputs on $\cQ\setminus S$ free.  The result represents the unnormalized
postselected state on $\cQ\setminus S$: its coefficient at
$\bb:\cQ\setminus S\to\{0,1\}$ is $\me{\ba,\bb}{U_C}{\bx}$, where $(\ba,\bb)$ denotes
the assignment on $\cQ$ extending both $\ba$ and $\bb$.
Its squared norm is the probability of observing $\ba$ on $S$.

Probability marginals require products of amplitudes with their complex
conjugates.  On input $\bx$, the probability of observing $\ba$ on $S$ is
\[
    p_S(\ba)
    =
    \sum_{\bb\in\{0,1\}^{\cQ\setminus S}}
        \me{\ba,\bb}{U_C}{\bx}\,
        \overline{\me{\ba,\bb}{U_C}{\bx}}.
\]
The above is an FAQ over $\C$ formed from two copies of the gate factors,
one conjugated.  Both inputs are fixed to $x$, and corresponding output
variables are identified.  The output variables on $S$ remain free, those
on $\cQ\setminus S$ are bound and summed out, and the internal variables of
the two copies remain distinct and bound.

\subsection{Rewriting FAQs with functional dependencies}
\label{subsec:circuit-fds}

The raw quantum-circuit encoding uses distinct basis-index variables
immediately before and after every gate.  For many gates, the gate factor is nonzero
only when some variables are determined by others.  If a determined
variable is bound, then, for each assignment to the other variables, at
most one of its values contributes to the sum.  We can eliminate its
summation by substituting that value throughout the FAQ.  As we will see,
reducing the number of variables can lower the evaluation cost, making
this a useful preprocessing step.

\begin{definition}[Functional dependency]
\label{def:fd}
A factor $F:D_{S_F}\to K$ satisfies the \emph{functional dependency} (FD)
$A\to b$, where $b\in S_F$ and $A\subseteq S_F\setminus\{b\}$, if there
exists a function $g:D_A\to D_b$ such that $z_b=g(\bz_A)$ for every
$\bz\in\supp(F)$.  Equivalently, any two assignments in $\supp(F)$ that
agree on $A$ also agree on $b$. 
\end{definition}

A functional dependency constrains the support of a factor.
Suppose that a factor $F$ satisfies $A\to b$,
witnessed by $g$, and define its residual factor by
\[
    F_{\setminus b}(\bz_{S_F\setminus\{b\}})
    \coloneqq
    \left.F(\bz_{S_F})\right|_{z_b=g(\bz_A)}.
\]
Then $F(\bz_{S_F})=\iv{z_b=g(\bz_A)}
F_{\setminus b}(\bz_{S_F\setminus\{b\}})$.
If $b$ is bound in an FAQ with factor set $\cF$ containing
a chosen occurrence of $F$,
the following holds for every fixed assignment to the other variables:
\begin{equation}
\label{eq:fd-substitution}
\begin{aligned}
    &\bigoplus_{\bz_b\in D_b}
        F(\bz_{S_F})
        \bigotimes_{F'\in\cF\setminus\{F\}}F'(\bz_{S_{F'}})\\
    &\quad=
    \bigoplus_{\bz_b\in D_b}
        \iv{z_b=g(\bz_A)}\,
        F_{\setminus b}(\bz_{S_F\setminus\{b\}})
        \bigotimes_{F'\in\cF\setminus\{F\}}F'(\bz_{S_{F'}})\\
    &\quad=
    F_{\setminus b}(\bz_{S_F\setminus\{b\}})
    \left.
        \bigotimes_{F'\in\cF\setminus\{F\}}F'(\bz_{S_{F'}})
    \right|_{z_b=g(\bz_A)}.
\end{aligned}
\end{equation}
Here, $\cF\setminus\{F\}$ removes the chosen occurrence of $F$.
The indicator selects the only value of $z_b$ that can contribute.
Eliminating $b$ therefore requires substituting $g(z_A)$ into every factor
containing $b$, while $F_{\setminus b}$ retains the weight contributed by
$F$.

The dependency is witnessed by one factor, but the substitution affects
every factor sharing the determined variable.  A factor $F'$ containing
$b$ acquires the declared scope $(S_{F'}\setminus\{b\})\cup A$.
The determining factor itself loses only $b$, since $A$ was already
contained in its scope.  Other factors may grow: $F'(b,c)$ becomes
$F'(g(a_1,a_2),c)$, with scope $\{a_1,a_2,c\}$.
Identity substitutions replace one variable by another and do not
increase arity.

In the quantum setting, for diagonal gates such as $T$ and $\CZ$, the support constraints equate
corresponding input and output variables. 
Classical reversible gates such as $\CNOT$ and Toffoli
determine their target outputs from their inputs and also permit inverse
substitutions.  Table~\ref{tab:gate-fd} lists representative dependencies
and the nonzero values that must be retained. 

\begin{table}[t]
\centering
\small
\setlength{\tabcolsep}{4pt}
\begin{tabular}{@{}llll@{}}
\toprule
Gate & Support relation & FDs & Nonzero value\\
\midrule
$X$
    & $y=x\oplus 1$
    & $x\to y,\quad y\to x$
    & $1$\\[2pt]
$Y$
    & $y=x\oplus 1$
    & $x\to y,\quad y\to x$
    & $i(-1)^x$\\[2pt]
$Z$
    & $y=x$
    & $x\to y,\quad y\to x$
    & $(-1)^x$\\[2pt]
$H$
    & full support
    & none
    & $2^{-1/2}(-1)^{xy}$\\[2pt]
$S$
    & $y=x$
    & $x\to y,\quad y\to x$
    & $i^x$\\[2pt]
$T$
    & $y=x$
    & $x\to y,\quad y\to x$
    & $e^{i\pi x/4}$\\[2pt]
$\CZ$
    & $\mathbf y=\mathbf x$
    & $\begin{gathered}[t]
         x_1\to y_1,\quad y_1\to x_1\\
         x_2\to y_2,\quad y_2\to x_2
       \end{gathered}$
    & $(-1)^{x_1x_2}$\\[2pt]
$\CNOT$
    & $\mathbf y=(x_1,x_1\oplus x_2)$
    & $\begin{gathered}[t]
         x_1\to y_1,\quad y_1\to x_1\\
         \{x_1,x_2\}\to y_2\\
         \{y_1,y_2\}\to x_2
       \end{gathered}$
    & $1$\\[2pt]
Toffoli
    & $\begin{gathered}[t]
         y_i=x_i\ (i=1,2)\\
         y_3=x_3\oplus x_1x_2
       \end{gathered}$
    & $\begin{gathered}[t]
         x_1\to y_1,\quad y_1\to x_1\\
         x_2\to y_2,\quad y_2\to x_2\\
         \{x_1,x_2,x_3\}\to y_3\\
         \{y_1,y_2,y_3\}\to x_3
       \end{gathered}$
    & $1$\\
\bottomrule
\end{tabular}
\caption{Support relations, representative functional dependencies, and
nonzero gate-factor values.  Here, $\mathbf x$ and $\mathbf y$ are the full
input and output tuples.}
\label{tab:gate-fd}
\end{table}

The substitution rule eliminates bound variables.  If an identity
dependency relates a bound variable to a free one, we retain the free
variable and eliminate the bound one.  If both are free, we retain an
explicit equality constraint so that both output coordinates remain
present.  More generally, determined free coordinates are retained unless
an output reconstruction map is explicitly part of the representation.
Boundary factors are substituted like all other factors.

We write $\widehat Q$ for the result of a specified sequence of these
substitutions; $\widehat Q$ and $Q$
are pointwise equal on the same free-variable domain. A substitution may
nevertheless enlarge a factor scope, so the structural bounds below refer
to the particular rewritten FAQ $\widehat Q$ under consideration.  The
cost of constructing its chosen factor representations is included in
preprocessing.  By contrast, changing only the representation of a fixed
factor leaves both its function and its declared scope unchanged.

\begin{example}[Rewriting the FAQ of Example~\ref{ex:qctofaq}]
\label{ex:qctofaq-fd}
For the raw output-state FAQ of Example~\ref{ex:qctofaq}, the $\CNOT$
factor imposes $y_1=a$ and $d=a\oplus b$.  The $T$ factor imposes $e=d$
and contributes the phase $t(d)\coloneqq e^{i\pi d/4}$.  The $\CZ$ factor imposes
$f=e$ and $y_3=c$ and contributes the phase
$\zeta(e,c)\coloneqq(-1)^{ec}$.

First use the identity dependencies to substitute
$a\leftarrow y_1$, $c\leftarrow y_3$, $f\leftarrow e$, and then
$e\leftarrow d$.  The free output variables are retained, and the query
becomes
\[
\begin{aligned}
    Q(y_1,y_2,y_3)
    =
    \sum_{b,d\in\{0,1\}}
      &H(y_1;x_1)H(b;x_2)H(y_3;x_3)\\
      &{}\cdot\iv{d=y_1\oplus b}\,
       t(d)\,\zeta(d,y_3)\,H(y_2;d).
\end{aligned}
\]
The remaining indicator determines $d$ from $y_1,b$.  Substituting
$d\leftarrow y_1\oplus b$ into all factors containing $d$ eliminates its
sum and gives
\[
\begin{aligned}
    \widehat Q(y_1,y_2,y_3)
    =
    \sum_{b\in\{0,1\}}
      &H(y_1;x_1)H(b;x_2)H(y_3;x_3)\\
      &{}\cdot t(y_1\oplus b)\,
       \zeta(y_1\oplus b,y_3)\,
       H(y_2;y_1\oplus b).
\end{aligned}
\]
Equation~\eqref{eq:fd-substitution} shows that the two FAQs agree for every
output assignment.  Of the six original bound variables, only $b$ remains.
The final substitution also illustrates why a factor
scope may grow: $H(y_2;d)$ has scope $\{y_2,d\}$, whereas
$H(y_2;y_1\oplus b)$ has scope $\{y_1,y_2,b\}$.

\end{example}

\noindent
{\bf Note}. The factors in an FAQ corresponding to the various gates in a quantum circuit 
are evaluated over specific relations. For instance, in Figure~\ref{fig:gate-factors}, 
the relations representing possible combinations of input and output values are shown for 
factors corresponding to $H$, $T$, $CZ$, $\CNOT$, and Toffoli gates. The functional dependencies 
mentioned in this section are properties of these relations. That is why we are allowed to 
use these FDs if this helps to simplify the FAQ. But we may also decide not to use a particular
FD. In the sequel, we will only use FDs that allow us to 
replace a variable by another variable. In Example~\ref{ex:qctofaq-fd} above, we would thus 
carry out the substitutions $a\leftarrow y_1$, $c\leftarrow y_3$, $f\leftarrow e$, and 
$e\leftarrow d$, but not the substitution $d\leftarrow y_1\oplus b$, which 
would replace a variable by a functional term.
Moreover, we will never explicitly introduce indicator expressions.
This convention is motivated by the evaluation algorithm for FAQs to be presented in the next section. 
This algorithm assumes factors to only have constants and variables as arguments and no indicator factors.
In Example~\ref{ex:qctofaq-fd}, we would thus transform the original FAQ from 
Example~\ref{ex:qctofaq} to 
   $$Q(y_1,y_2,y_3)
    =
    \sum_{b,d\in\{0,1\}}
      H(y_1;x_1)H(b;x_2)H(y_3;x_3) t(d)\,\zeta(d,y_3)\,H(y_2;d)$$
and not any further.

\section{FAQ evaluation and InsideOut}
\label{sec:FAQs}
We next describe \InsideOut, the FAQ evaluation method of~\citet{AboKhamisNgoRudra2016}.  At a high level, \InsideOut is a semiring-valued form of classical variable elimination, or bucket elimination~\citep{Dechter1999}: one chooses a variable, collects the factors that depend on it, multiplies them, aggregates out the variable, and replaces the bucket by the resulting factor.  The semiring laws, and in particular distributivity, justify pushing the aggregation through the product in this way.  The distinctive feature of the FAQ framework is that the resulting bucket need not be evaluated by naively materializing its full table; its computation can exploit the structure and representation of the participating factors. 
We therefore expose \BucketEval as the operation whose implementation varies between representations, while \InsideOut specifies the surrounding variable-elimination scheme.  

\subsection{The \InsideOut algorithm}
\label{subsec:insideout}

Let $(K,\oplus,\otimes,\mathbf0,\mathbf1)$ be a commutative semiring,
let $X$ be a finite set of variables with finite domains $D_x$, and let
$Y\subseteq X$ be the free variables.  Recall that $D_S\coloneqq\prod_{x\in S}D_x$ for $S\subseteq X$. Given a finite indexed collection $\cF$ of factors
$F:D_{S_F}\to K$ with scopes $S_F\subseteq X$, consider the FAQ
\[
    Q(\mathbf x_Y)
    \coloneqq
    \bigoplus_{\mathbf x_{X\setminus Y}\in D_{X\setminus Y}}
       \bigotimes_{F\in\cF}F(\mathbf x_{S_F}).
\]

\noindent
Algorithm \InsideOut evaluates this query by eliminating the bound variables
one at a time.  Choose any order $z_1,\ldots,z_m$ of $X\setminus Y$.
The free variables are never eliminated and remain arguments of the result.
The order affects the evaluation cost, but not the query value. We summarize \InsideOut in Algorithm~\ref{alg:insideout}.

\paragraph{Bags and messages}
Let $\cA_1\coloneqq\cF$ be the initial indexed collection of active factors.
At step $i$, the \emph{bag} of $z_i$ contains all active factors whose
scopes include $z_i$.  Define
\begin{equation}
 \cB_i\coloneqq\{F\in\cA_i\mid z_i\in S_F\},\qquad
 U_i\coloneqq\{z_i\}\cup\bigcup_{F\in\cB_i}S_F,\qquad
 R_i\coloneqq U_i\setminus\{z_i\}.
 \label{eq:bucket-bag-boundary}
\end{equation}
The set $U_i$ is the \emph{elimination bag}, and
$R_i$ contains the variables that remain after aggregating $z_i$
from the bucket.  The resulting \emph{message} is
\begin{equation}
 M_i'(\mathbf x_{R_i})
 \coloneqq
 \bigoplus_{x_{z_i}\in D_{z_i}}
    \bigotimes_{F\in\cB_i}F(\mathbf x_{S_F}).
 \label{eq:unfiltered-message}
\end{equation}
Thus a message summarizes the consumed factors as a function of their
remaining variables.  
\paragraph{Indicator projections.}
Sparse evaluators may also use the supports of factors outside the bag 
to restrict the assignments considered.  For a factor $G$ and
$T\subseteq X$, its \emph{indicator projection} onto $T$ is the factor on
$S_G\cap T$ defined by
\begin{equation}
 G_{/T}(\mathbf a)\coloneqq
 \mathbf1\!\left[
   \exists\mathbf b\in D_{S_G\setminus T}\mid 
   G(\mathbf a,\mathbf b)\ne\mathbf0
 \right],
 \qquad \mathbf a\in D_{S_G\cap T}.
 \label{eq:indicator-projection}
\end{equation}
Here the indicator takes the value $\mathbf1$ when the condition holds and
$\mathbf0$ otherwise.  It records whether
$\mathbf a$ extends to a supported tuple of $G$.

Let $\cP_i$ be any selected set of projections $G_{/U_i}$ with
$G\in\cA_i\setminus\cB_i$.  Since these outside factors do not contain
$z_i$, every projection has scope contained in $R_i$.  When such projections are used, the basic message in~\eqref{eq:unfiltered-message} is replaced by its filtered form below;
equivalently, $\BucketEval$ should be understood as depending also on the
selected set $\cP_i$.
The bag evaluator
returns
\begin{equation}
 M_i(\mathbf x_{R_i})
 \coloneqq
 M_i'(\mathbf x_{R_i})\otimes
    \bigotimes_{P\in\cP_i}P(\mathbf x_{S_P}),
 \qquad
 \cA_{i+1}\coloneqq
    (\cA_i\setminus\cB_i)\cup\{M_i\}.
 \label{eq:insideout-message}
\end{equation}
The evaluator may apply the projections during computation, without first
materializing $M_i'$.  Choosing $\cP_i=\varnothing$ gives ordinary variable
elimination. A projection may change the message on assignments that cannot contribute
to the remaining query.  

\paragraph{Factorized representation.} \InsideOut computes a  \emph{factorized representation} of an FAQ answer, that is,
$Q:D_Y\to K$ is a collection of factors $\mathcal A$,
with $S_F\subseteq Y$ for every $F\in\mathcal A$, such that
\[
    Q(\mathbf x_Y)
    =
    \bigotimes_{F\in\mathcal A}
    F(\mathbf x_{S_F})
\]
for every $\mathbf x_Y\in D_Y$.

\begin{algorithm}
\caption{\InsideOut}
\label{alg:insideout}
\begin{algorithmic}[1]
\Require Factors $\mathcal F$, free variables $Y$, and an order
         $z_1,\ldots,z_m$ of the bound variables
\Ensure FAQ $Q$ in factorized representation.
\State $\mathcal A \gets \mathcal F$
\For{$i=1,\ldots,m$}
    \State $\mathcal B_i \gets
        \{F\in\mathcal A \mid z_i\in S_F\}$
    \State $U_i \gets
        \{z_i\}\cup\bigcup_{F\in\mathcal B_i} S_F$
    \State $R_i \gets U_i\setminus\{z_i\}$
    \State Choose
        $\mathcal P_i \subseteq
        \{F_{/U_i} \mid F\in\mathcal A\setminus\mathcal B_i,\,
        S_F\cap U_i\neq\emptyset\}$
        \Comment{set $\mathcal P_i=\emptyset$ for non-sparse evaluation}
    \State Define the message FAQ
        $\displaystyle
        M_i(\mathbf x_{R_i})
        :=
        \bigoplus_{x_{z_i}\in D_{z_i}}
        \left(
          \bigotimes_{F\in\mathcal B_i}
          F(\mathbf x_{S_F})
        \right)
        \otimes
        \left(
          \bigotimes_{P\in\mathcal P_i}
          P(\mathbf x_{S_P})
        \right)$
    \State $G_i \gets \Call{BagEval}{\mathcal M_i}$
    \State $\mathcal A \gets
        (\mathcal A\setminus\mathcal B_i)\uplus\{G_i\}$
\EndFor
\State \Return
    $\displaystyle
    Q(\mathbf x_Y)
    :=
    \bigotimes_{F\in\mathcal A}
    F(\mathbf x_{S_F})$ in factorized representation.
\end{algorithmic}
\end{algorithm}
    
\paragraph{Batch elimination.}  Let $\cA$ be the current indexed collection of active factors and let $Z$ be a
nonempty set of remaining bound variables. We can evaluate FAQs also by eliminating variables in batch.  Define
$\cB_Z\coloneqq
       \{F\in\cA\mid S_F\cap Z\ne\emptyset\}$,
    $U_Z\coloneqq Z\cup\bigcup_{F\in\cB_Z}S_F$,
$R_Z\coloneqq U_Z\setminus Z$. 
The block message is then 
$M_Z'(\mathbf x_{R_Z})
    \coloneqq
    \bigoplus_{\mathbf z\in D_Z}
       \bigotimes_{F\in\cB_Z}
          F\bigl((\mathbf z,\mathbf x_{R_Z})|_{S_F}\bigr)$.
Including $Z$ in $U_Z$ ensures that variables occurring in no active factor
are still aggregated over their domains.
As in~\eqref{eq:insideout-message}, the evaluator may multiply this message
by selected projections $G_{/U_Z}$ from factors
$G\in\cA\setminus\cB_Z$.  It then replaces $\cB_Z$ by the resulting message.
Every active factor containing an eliminated variable must participate.
All outside factors are independent of $Y$, so the same distributivity
and indicator-projection argument proves correctness.

\section{Explicit evaluation and fractional hypertree width}
\label{sec:explicit-tables}
We first consider FAQs whose factors are represented by indexed listings of
their nonzero entries. Dense tables can also be used, with the cost of
scanning their entries charged separately. Bounds count arithmetic on
explicit semiring values at unit cost and exclude the cost of finding a
decomposition.

As before, let $Q(\mathbf x_Y)
    =
    \bigoplus_{\mathbf x_{X\setminus Y}}
       \bigotimes_{F\in\cF}F(\mathbf x_{S_F})$,
and let $\cH_Q=(X,\{S_F:F\in\cF\})$ be its hypergraph. Its
\emph{primal graph} has vertex set $X$, with two variables adjacent
when they occur in a common factor scope.

A \emph{tree decomposition} of $\cH_Q$ is a tree $T$ with a bag
$B_t\subseteq X$ at each node $t$, such that every factor scope is
contained in some bag and, for every variable $x$, the bags containing
$x$ form a connected subtree. The latter is the
\emph{running-intersection property}. The width is
$\max_t |B_t|-1$, and minimizing this quantity gives the
\emph{treewidth} $\tw(\cH_Q)$, equivalently the treewidth of
$\operatorname{primal}(\cH_Q)$~\citep{DBLP:journals/jal/RobertsonS86}.

For free variables $Y$, a decomposition is 
\emph{$Y$-connex}~\citep{Bagan+2007}
if
there is a connected subtree $T_Y$ whose bags have union exactly
$Y$. We call $T_Y$ the \emph{core}. For a scalar query
($Y=\emptyset$), every decomposition is $Y$-connex and no core is
needed.

Put
\[
N=\max\bigl(\{1\}\cup\{|\supp(F)|:F\in\cF\}\bigr).
\]
For a bag $B$, its \emph{fractional cover number}
$\rho^*_{\cH_Q}(B)$ is the minimum of
$\sum_{F\in\cF}\lambda_F$ over $\lambda_F\geq0$ such that
\[
\sum_{F:x\in S_F}\lambda_F\geq1
\qquad\text{for every }x\in B.
\]
The maximum of these numbers over all bags is the decomposition's
\emph{fractional width}~\citep{DBLP:journals/talg/GroheM14}. 
Minimizing over decompositions gives the
\emph{fractional hypertree width} $\fhtw(\cH_Q)$, or
$\fhtw_Y(\cH_Q)$ when the decomposition is required to be
$Y$-connex.
Since every variable is contained in some factor scope,
\[
    \rho^*_{\cH_Q}(B)\leq |B|,
    \qquad\text{and hence}\qquad
    \fhtw(\cH_Q)\leq\tw(\cH_Q)+1.
\]
Note that without loss of generality, we can assume that on a hypergraph with $n$ vertices, any fractional hypertree decomposition has at most $n$ nodes.
\begin{proposition}[FAQ width bound \citep{AboKhamisNgoRudra2016}]
\label{prop:sparse-decomposition}
Let $Q$ be an FAQ with support-listed input factors of size at most
$N$. Given a $Y$-connex decomposition with 
fractional width $w$, it is possible to compute a factorized representation of $Q$
in
\[
    \widetilde O\!\left(N^w\right)\poly(|Q|)
\]
time. In particular, when $Y=\emptyset$, this computes the scalar value of $Q$. 
\end{proposition}
Here, and elsewhere, $\widetilde O(\cdot)$ suppresses factors
polylogarithmic in $N$.
 If the semiring has no zero divisors, this factorized representation can
 be expanded into a support listing of $Q$ in additional
 $\widetilde O\!\left(|\supp(Q)|\right)\poly(|Q|)$
 time. 
Consequently, given an optimal decomposition, scalar evaluation
takes $
\widetilde O\!\left(N^{\fhtw(\cH_Q)}\right)\poly(|Q|)$,
while for free variables, over a semiring without zero divisors, the
corresponding output-sensitive bound is
\begin{equation}
 \widetilde O\!\left(
 N^{\fhtw_Y(\cH_Q)}+|\supp(Q)|
 \right)\poly(|Q|).
 \label{eq:insideout-fhtwF-bound}
\end{equation}

\paragraph{Recovering the Markov--Shi bound.}
Recall the \emph{circuit graph} $G_C$: its vertices are the gates of $C$
and one boundary node for each input and output, and its edges are the
wire segments. We retain parallel edges and, as in
Section~\ref{subsec:circuits-as-faqs}, pad untouched wires with identity
gates. This is the underlying multigraph of the circuit tensor
network, also denoted by $N_C$ below.

A landmark result of~\citet{MarkovShi2008} states that circuits of bounded
gate arity can be simulated in time
$2^{O(\tw(G_C))}\poly(|C|)$. We recover this bound directly from
Proposition~\ref{prop:sparse-decomposition} by comparing fractional
hypertree width with the treewidth of $G_C$. We first consider amplitudes,
then extend the encoding to measurement probabilities below.

Let $\cH$ be the structural hypergraph of the scalar FAQ
$Q_C$: its vertices are the variables and its hyperedges are
the factor scopes. In graph theoretic terminology, the circuit graph $G_C$, with parallel edges identified, is the \emph{line graph} of $\cH$. The following kind of comparison is well-known and easy to verify (see e.g.,~\citet[Lemma~4.6]{DBLP:conf/pods/Lanzinger22}). Note that this holds generally for hypergraphs and their line graphs and is not specific to our application to quantum circuits.

\begin{proposition}
\label{lem:raw-faq-width-comparison}
Let $C$ be a quantum circuit with circuit graph $G_C$, and let
$\bx,\by$ be computational-basis input and output strings. Let $\cH$
be the structural hypergraph of the scalar FAQ $Q_C$. Then
\begin{equation}
 \fhtw(\cH)\leq\tw(G_C)+1.
 \label{eq:raw-faq-circuit-width}
\end{equation}
\end{proposition}

The standard argument constructively turns a width-$k$ tree decomposition of $G_C$ into a tree decomposition of $\cH$ and
fractional width at most $k+1$. Applying
Proposition~\ref{prop:sparse-decomposition} to the latter thus computes
$\me{\by}{U_C}{\bx}$ in
\[
 \widetilde O\!\left(N^{k+1}\right)\poly(|C|)
 =2^{O(k)}\poly(|C|).
\]
Taking $k=O(\tw(G_C)+1)$ clearly recovers the
Markov--Shi bound immediately. More interestingly, fractional hypertree width can be strictly lower even in simple circuits (see example below). Thus by simply viewing the problem as a CQ we already obtain polynomial speedup over Markov--Shi.

The same argument computes measurement probabilities, including
marginals. Pair the amplitude and conjugate-amplitude indices on each
wire into a single variable $(z,z')\in\{0,1\}^2$, and replace each gate
factor by $G[\mathbf c;\mathbf b]\,
\overline{G[\mathbf c';\mathbf b']}$. An input bit $x$ contributes
$\iv{z=z'=x}$; an output fixed to $y$ contributes $\iv{z=z'=y}$,
whereas an unobserved output contributes $\iv{z=z'}$, tracing out that
qubit. Summing all paired variables gives the probability of the
specified outcomes. More generally, a single-qubit measurement effect
$E$ contributes $\langle z'|E|z\rangle$, and a quantum channel
$\mathcal E$ contributes the gate factor
$\langle\mathbf c|\mathcal E(|\mathbf b\rangle\langle\mathbf b'|)
|\mathbf c'\rangle$. This is in fact the density-matrix encoding used
by~\citet{MarkovShi2008} and it works effectively unchanged in our setting. The factor scopes, and hence $\cH$ and $G_C$,
are unchanged, only the wire domain grows from two to four values so the respective widths are also unchanged.

We see that in the simplest setting, where gate behaviour is represented explicitly and directly, viewing circuits as FAQs provides an immediate theoretical benefit. Existing bounds turn out to be special cases of more general theorems, and we can obtain real improvements in runtime by considering functional dependencies and fractional edge covers. 
The following example illustrates these graphs and hypergraphs and the corresponding width measures.

\begin{example}\label{ex:tw-vs-fractional}
Take three data qubits and ancillas $a_{12},a_{13},a_{23}$, initially
$\ket0$. Apply a Hadamard to each data qubit, then a Toffoli with data
controls $i,j$ and target $a_{ij}$ for each pair $(i,j)$. Postselect
all ancillas to $\ket0$, leaving the data outputs free:
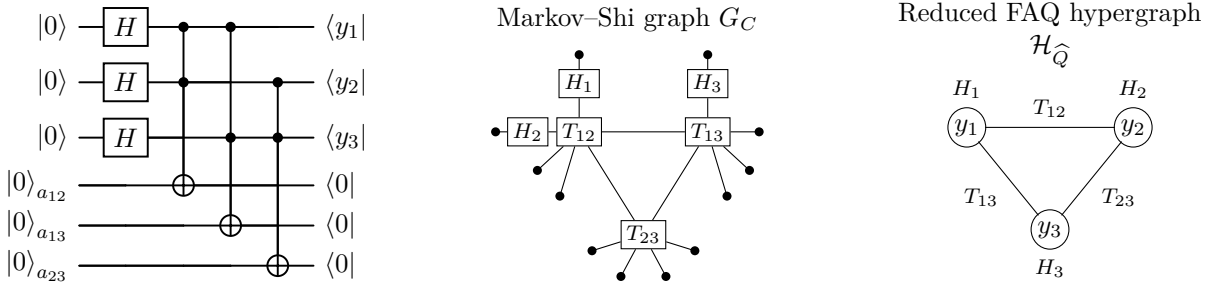
\begin{figure}[H]
\centering
\begin{minipage}[c]{0.31\linewidth}
\centering
\resizebox{\linewidth}{!}{%
\begin{quantikz}[column sep=0.35cm, row sep=0.25cm]
\lstick{$\ket 0$}         & \gate{H} & \ctrl{3} & \ctrl{4} & \qw      & \rstick{$\bra{y_1}$} \\
\lstick{$\ket 0$}         & \gate{H} & \ctrl{2} & \qw      & \ctrl{4} & \rstick{$\bra{y_2}$} \\
\lstick{$\ket 0$}         & \gate{H} & \qw      & \ctrl{2} & \ctrl{3} & \rstick{$\bra{y_3}$} \\
\lstick{$\ket 0_{a_{12}}$} & \qw      & \targ{}  & \qw      & \qw      & \rstick{$\bra 0$} \\
\lstick{$\ket 0_{a_{13}}$} & \qw      & \qw      & \targ{}  & \qw      & \rstick{$\bra 0$} \\
\lstick{$\ket 0_{a_{23}}$} & \qw      & \qw      & \qw      & \targ{}  & \rstick{$\bra 0$}
\end{quantikz}}
\end{minipage}\hfill
\begin{minipage}[c]{0.29\linewidth}
\centering
{\small Markov--Shi graph $G_C$\par}
\medskip
\begin{tikzpicture}[x=0.85cm,y=0.85cm,
  gate/.style={draw,fill=white,inner sep=2pt,font=\scriptsize},
  boundary/.style={circle,fill,inner sep=1.2pt}]
\node[gate] (t12) at (0,0) {$T_{12}$};
\node[gate] (t13) at (2,0) {$T_{13}$};
\node[gate] (t23) at (1,-1.6) {$T_{23}$};
\node[gate] (h1) at (0,0.75) {$H_1$};
\node[gate] (h2) at (-0.8,0) {$H_2$};
\node[gate] (h3) at (2,0.75) {$H_3$};
\draw (t12)--(t13)--(t23)--(t12)
      (h1)--(t12) (h2)--(t12) (h3)--(t13);
\foreach \name/\x/\y in {h1/0/1.2,h2/-1.3/0,h3/2/1.2,
  t12/-0.65/-0.6,t12/-0.3/-1.0,
  t13/2.8/0,t13/2.65/-0.6,t13/2.3/-1.0,
  t23/0.2/-1.85,t23/0.65/-2.25,t23/1.35/-2.25,t23/1.8/-1.85}
  \draw (\name)--(\x,\y) node[boundary] {};
\end{tikzpicture}
\end{minipage}\hfill
\begin{minipage}[c]{0.28\linewidth}
\centering
{\small Reduced FAQ hypergraph $\cH_{\widehat Q}$\par}
\begin{tikzpicture}[x=1cm,y=1cm,
  faqvertex/.style={inner sep=2pt,font=\small}]
\node[faqvertex] (y1) at (0,0) {$y_1$};
\node[faqvertex] (y2) at (2.2,0) {$y_2$};
\node[faqvertex] (y3) at (1.1,-1.35) {$y_3$};
\draw (y1)--node[above,font=\scriptsize] {$T_{12}$} (y2)
      (y1)--node[below left,font=\scriptsize] {$T_{13}$} (y3)
      (y2)--node[below right,font=\scriptsize] {$T_{23}$} (y3);
\draw (y1) ellipse [x radius=0.25cm,y radius=0.26cm];
\draw (y2) ellipse [x radius=0.25cm,y radius=0.26cm];
\draw (y3) ellipse [x radius=0.25cm,y radius=0.26cm];
\node[font=\scriptsize,above] at (0,0.26) {$H_1$};
\node[font=\scriptsize,above] at (2.2,0.26) {$H_2$};
\node[font=\scriptsize,below] at (1.1,-1.61) {$H_3$};
\end{tikzpicture}
\end{minipage}
\caption{In $G_C$, $T_{ij}$ denotes the Toffoli with control bits $i,j$. Black leaves denote input/output boundary nodes.
In $\cH_{\widehat Q}$, vertices are the free variables; circles are unary Hadamard scopes,
and lines are binary scopes of the $T_{ij}$ factors; recall that constants are irrelevant.}
\label{fig:toffoli-faq-graphs}
\end{figure}

We first inspect the circuit graph $G_C$: since the graph contains a clique (i.e., a triangle) of 
the vertices $y_1,y_2,y_3$, any tree decomposition must have a bag that contains 
all three vertices of the clique. Hence, the treewidth is $2$.

Now consider the FAQ representation of this circuit. 
Let the variables of the intermediate states at the first qubit 
be $a_1,a_2$, at the second qubit $b_1,b_2$, and at the third qubit $c_1,c_2$.
Then the raw FAQ has the form 
\begin{eqnarray*}
 Q(y_1,y_2,y_3)  &=& 
 H(0,a_1) H(0,b_1) H(0,c_1)\\
 && T(a_1,b_1,0,a_2,b_2,0)
 T(a_2,c_1,0,y_1,c_2,0) T(b_2,c_2,0,y_2,y_3,0)
\end{eqnarray*}

\noindent
We observe that the Toffoli factors satisfy the following FDs: 
$a_1 = a_2 = y_1$, $b_1 = b_2 = y_2$, and $c_1 = c_2 = y_3$.
We can thus simplify the FAQ as follows:

\begin{eqnarray*}
 \widehat Q(y_1,y_2,y_3) &=& 
 H(0,y_1) H(0,y_2) H(0,y_3)\\
 && T(y_1,y_2,0,y_1,y_2,0)
 T(y_1,y_3,0,y_1,y_3,0) T(y_2,y_3,0,y_2,y_3,0)
\end{eqnarray*}

\noindent
A fractional hypertree decomposition of the simplified FAQ has a single bag and 
the fractional cover number of the variables $y_1,y_2,y_3$ is $\frac{3}{2}$, i.e., 
we may put weight $\tfrac12$ on each of the 3 $T_{ij}$ hyperedges.
Hence $\fhtw_Y(\cH_{\widehat Q})=\frac{3}{2}$.

The zeros at the last input and last output position act as selections, 
which reduce the support of each $T$-factor to 3 tuples. 
Hence, the support of $\widehat Q(y_1,y_2,y_3)$ has at most $3^{3/2}$ tuples and it
can also be computed in this order of magnitude by the OutsideIn algorithm 
of \cite{AboKhamisNgoRudra2016}

This extends to $n\geq3$ data qubits: prepare $\ket0^{\otimes n}$,
use a fresh zero-initialized ancilla and a Toffoli for each pair $(i,j)$,
and postselect all ancillas to $\ket0$. The reduced FAQ is
\[
 \widehat Q_n(\mathbf y)
 =\prod_{i=1}^n H(0,y_i) \prod_{1\leq i<j\leq n}T(y_i,y_j,0,y_i,y_j,0).
\]
The circuit graph is the clique $K_n$. Again, 
every tree decomposition must have a bag containing all vertices of the clique.
Hence, $\tw = n-1$.
On the other hand, every vertex in the hypergraph of the simplified FAQ 
is contained in  $n-1$ $T$-hyperedges. Hence, by putting weight $\frac{1}{n-1}$ on every $T$-edge
we get a fractional edge cover of weight $\frac{n}{2}$, which is optimal in this case. 
Hence $\fhtw_Y(\cH_{\widehat Q_n})=\frac{n}{2}$. \end{example}

\section{Phase representations and rank-width}
\label{sec:compressed-boundaries}
We now return to the circuit FAQ $Q_C$.  For a circuit $C$ over a fixed
quadratically $r$-compatible gate set, we evaluate its FD-reduced form
$\widehat Q_C$ over a semiring that retains phases formally.
The atoms have values in $K_r$.  We eliminate branches of a
rank-decomposition by batch \InsideOut{}, maintaining small
$K_r$-expressions for the resulting factors.
We show that, for fixed computational-basis boundaries and given a
rank-decomposition of the resulting phase graph of width $k$, this
yields an \InsideOut{} evaluation of $Q_C$ in time
$O(4^k\poly(|C|))$.

\subsection{Quantum circuits over compatible gate sets}

\paragraph{The phase semiring.}
For an integer $r\geq2$, write $\omega_r\coloneqq e^{2\pi i/r}$ and
consider the \emph{phase semiring}
$K_r\coloneqq\mathbb R[\Z_r]$, whose elements are formal sums
$\sum_{a\in\Z_r}n_a\phases{a}$ with $n_a\in\mathbb R$.
Equality is coefficientwise, so an element can be represented by its
coefficient vector $(n_0,\ldots,n_{r-1})$.
Addition is componentwise:
\[
    \Bigl(\sum_a n_a\phases{a}\Bigr)
    \oplus
    \Bigl(\sum_a m_a\phases{a}\Bigr)
    \coloneqq
    \sum_a(n_a+m_a)\phases{a},
\]
and multiplication is convolution modulo $r$:
\[
    \Bigl(\sum_a n_a\phases{a}\Bigr)
    \otimes
    \Bigl(\sum_b m_b\phases{b}\Bigr)
    \coloneqq
    \sum_{c\in\Z_r}
       \Bigl(
          \sum_{\substack{a,b\in\Z_r\\a+b=c}}n_am_b
       \Bigr)\phases{c}.
\]
In particular,
$\phases{a}\otimes\phases{b}=\phases{a+b}$, with exponents added
in $\Z_r$.
The additive identity $\mathbf0$ is the zero coefficient vector,
whereas the multiplicative identity is
$\mathbf1\coloneqq\phases{0}$, whose coefficient vector is
$(1,0,\ldots,0)$.  Thus $\phases{0}$ is not zero.
These operations make $K_r$ a commutative semiring, and indeed a ring.
We identify a real scalar $n$ with $n\phases{0}$.

The \emph{evaluation map}
\begin{equation}
    \nu_r:K_r\longrightarrow\C,
    \qquad
    \nu_r\!\left(\sum_a n_a\phases{a}\right)
    \coloneqq
    \sum_a n_a\omega_r^a
    \label{eq:phase-evaluation}
\end{equation}
is a semiring homomorphism: it preserves addition, multiplication,
and their identities.
Formal equality is stronger than equality after evaluation.
For example, $\tfrac12(\phases{0}\oplus\phases{4})\in K_8$ has the
nonzero coefficient vector
$(\tfrac12,0,0,0,\tfrac12,0,0,0)$, yet its image under $\nu_8$ is zero.

\paragraph{Circuits over $r$-compatible gates.}
Let $X$ be a finite set of Boolean variables.
For $S\subseteq X$, write
$x_S\coloneqq\prod_{x\in S}x$, with $x_\varnothing=1$.
A \emph{$K_r$-expression} over $X$ is built from real constants and
phase terms $\phases{a x_S}$ using $\oplus$ and $\otimes$.
Its \emph{size} $\lVert E\rVert$ is the size of its expression tree,
counting variable occurrences within phase terms.  More precisely,
\[
\begin{aligned}
    \lVert c\rVert&=1 &&(c\in\mathbb R),\\
    \lVert\phases{a x_S}\rVert&=1+|S|,\\
    \lVert E_1\circ E_2\rVert
       &=1+\lVert E_1\rVert+\lVert E_2\rVert
       &&(\circ\in\{\oplus,\otimes\}).
\end{aligned}
\]
Repeated subexpressions are counted at each occurrence; no expansion
by distributivity is implicit in this definition.
A constant in $K_r$ has an expression of size $O(r)$ using its
$r$ real coefficients.  Substituting Boolean constants does not
increase expression size, and removing constant identities can only
decrease it.
For a phase polynomial $p=\sum_j a_jx_{S_j}$ over $\Z_r$, we use
$\phases{p}$ as shorthand for $\bigotimes_j\phases{a_jx_{S_j}}$;
its size includes this monomial encoding.  Thus
$\phases{p}\otimes\phases{q}=\phases{p+q}$, with coefficients
collected modulo $r$, and equal phase terms simplify by
$\alpha\phases{p}\oplus\beta\phases{p}=(\alpha+\beta)\phases{p}$.

For gate representatives, we use multiplicative expressions of the form
\[
    E(X)=\lambda\bigotimes_{j=1}^m\phases{a_jx_{S_j}},
    \qquad
    \lambda>0,\quad a_j\in\Z_r,\quad S_j\subseteq X.
\]
For an expression of this form, each assignment $\bx\in\{0,1\}^X$
gives a positive multiple of a single formal basis element in $K_r$.
For even $r$, we call $E$ \emph{quadratic} if $|S_j|\leq2$ for every
$j$ and $a_j=r/2$ whenever $|S_j|=2$.
Thus constant and linear terms may have arbitrary coefficients in
$\Z_r$, whereas every quadratic term has coefficient $r/2$.

For a gate $G$, let
\[
    \widehat F_G:\{0,1\}^{S_G}\longrightarrow\C
\]
be its reduced gate factor, obtained by the FD-based variable
identifications of Section~\ref{sec:circuit-faq}.
Here $S_G$ is the \emph{reduced scope}: input and output legs known
to carry the same bit are represented by a single variable.
These identifications are propagated throughout the circuit FAQ.

We call $G$ \emph{$r$-compatible} if there exists a single
multiplicative $K_r$-expression $E_G$ of the above form over $S_G$ such that
\[
    \widehat F_G(\bx)=\nu_r(E_G(\bx))
    \qquad
    \text{for every }\bx\in\{0,1\}^{S_G}.
\]
For even $r$, we call $G$ \emph{quadratically $r$-compatible} if
$E_G$ can be chosen quadratic.
Representability is therefore required on the entire reduced Boolean
domain, rather than on the original input--output factor.
In particular, the reduced factor has full support, even when the
original gate factor does not.

A finite gate set $\mathcal G$ is \emph{$r$-compatible}, respectively
\emph{quadratically $r$-compatible}, if every gate in $\mathcal G$
has the corresponding property.
The same terminology applies to a circuit
$C=(G_1,\ldots,G_T)$ on $\cQ$.
We fix the reduced gate factors and their witnessing expressions
$E_G$ with the gate set.

A chosen expression $E_G$ determines the lift
\[
    \widetilde F_G:\{0,1\}^{S_G}\longrightarrow K_r,
    \qquad
    \widetilde F_G(\bx)\coloneqq E_G(\bx).
\]
By construction,
\[
    \nu_r\circ\widetilde F_G=\widehat F_G.
\]
The lift changes only the annotations of the reduced gate relation.
Its scope and full support remain unchanged.

\begin{example}
The gate set $\{H,T,\CZ\}$ is quadratically $8$-compatible.
For $T$, the FD reduction identifies its input and output variables.
For $\CZ$, it identifies each output variable with the corresponding
input variable.  No identification is made for $H$.
The reduced gate factors are therefore
\[
\begin{aligned}
    \widehat F_H(b,c)
        &=2^{-1/2}\omega_8^{4bc},\\
    \widehat F_T(b)
        &=\omega_8^b,\\
    \widehat F_{\CZ}(b_1,b_2)
        &=\omega_8^{4b_1b_2}.
\end{aligned}
\]
They are witnessed by the quadratic $K_8$ expressions
\[
\begin{aligned}
    E_H(b,c)
        &=2^{-1/2}\phases{4bc},\\
    E_T(b)
        &=\phases{b},\\
    E_{\CZ}(b_1,b_2)
        &=\phases{4b_1b_2}.
\end{aligned}
\]
Thus $H$ and $\CZ$ have binary reduced factors, while $T$ has a unary
reduced factor.  Each has full support on its reduced scope.
Figure~\ref{fig:phase-gate-lifts} lists their values and lifts. We refer to \citet{WangEtAl2025,Cheng+2025} for more examples. 
\end{example}

\begin{figure}[t]
\centering
\small
\setlength{\arraycolsep}{3pt}
\renewcommand{\arraystretch}{1}
\[
\begin{array}[t]{c|c|c|c}
  \multicolumn{4}{c}{H}\\
  b & c & \widehat F_H & \widetilde F_H\\
  \hline
  0 & 0 & 2^{-1/2}  & 2^{-1/2}\phases{0}\\
  0 & 1 & 2^{-1/2}  & 2^{-1/2}\phases{0}\\
  1 & 0 & 2^{-1/2}  & 2^{-1/2}\phases{0}\\
  1 & 1 & -2^{-1/2} & 2^{-1/2}\phases{4}
\end{array}
\qquad
\begin{array}[t]{c|c|c}
  \multicolumn{3}{c}{T}\\
  b & \widehat F_T & \widetilde F_T\\
  \hline
  0 & 1        & \phases{0}\\
  1 & \omega_8 & \phases{1}
\end{array}
\]
\[
\begin{array}{c|c|c|c}
  \multicolumn{4}{c}{\CZ}\\
  b_1 & b_2 & \widehat F_{\CZ} & \widetilde F_{\CZ}\\
  \hline
  0 & 0 & 1  & \phases{0}\\
  0 & 1 & 1  & \phases{0}\\
  1 & 0 & 1  & \phases{0}\\
  1 & 1 & -1 & \phases{4}
\end{array}
\]
\caption{FD-reduced gate factors over $\C$ and their lifts to $K_8$.
All Boolean assignments to the displayed reduced scopes are supported.}
\label{fig:phase-gate-lifts}
\end{figure}

\paragraph{Evaluation over the phase semiring.}
Let $C$ be a circuit over an $r$-compatible gate set, and let
$\widehat Q_C$ be its FD-reduced scalar query.  To evaluate this query
over $K_r$, we replace each reduced gate factor by its chosen phase
representation and interpret aggregation and multiplication as the
operations of $K_r$.  Any remaining $0$--$1$ boundary relations are
represented by replacing $0$ and $1$ with $\mathbf0$ and $\mathbf1$.
The query expression, factor scopes, input supports, and bound variables
remain unchanged.  Only the factor annotations and their semiring
interpretation change.

Since $\nu_r$ commutes with finite sums and products, if this evaluation
returns $a\in K_r$, then
\begin{equation}
    \nu_r(a)=Q_C=\me{\by}{U_C}{\bx}.
    \label{eq:circuit-faq-phase-readout}
\end{equation}
Thus we apply $\nu_r$ only after aggregation.

For fixed computational-basis boundaries, the FD reduction
propagates the prescribed bits through the identified variables.
Inconsistent prescriptions give zero immediately.  Otherwise,
$\widehat Q_C$ has the same scalar value as $Q_C$ over $\C$.
For a query $Q$ with free boundary variables, we retain these variables
and any required constraints in its reduced form $\widehat Q$.
Evaluating $\widehat Q$ over $K_r$ and applying $\nu_r$ to each resulting
value recovers $Q$ pointwise on the original free-variable domain.

Recall that \InsideOut{} is independent of the factor representation:
it multiplies factors and aggregates variables in the chosen semiring.
Here the lifted gate factors are stored as $K_r$-expressions, and
each \InsideOut{} update constructs an expression for the $K_r$ values
of the new atom.
Although the input supports are preserved, intermediate supports
may differ from those obtained over $\C$, because a nonzero formal
value can evaluate to zero.

\subsection{The phase graph and rank-width}
\paragraph{The phase graph.}
Let $C$ be a quadratically $r$-compatible circuit and let $\eta=r/2$.
Let $\mathcal F$ be the indexed collection of lifted gate atoms after
FD reduction and substitution of consistent fixed computational-basis
boundaries, with variable set $X$.
Write the chosen representative of each atom $F(\bx_{S_F})$ as
\begin{equation}
    E_F(\bx_{S_F})
       =\lambda_F\phases{c_F}
          \otimes\bigotimes_{v\in S_F}\phases{b_{F,v}x_v}
          \otimes\bigotimes_{\{u,v\}\subseteq S_F}
                       \phases{\eta a_{F,uv}x_ux_v},
    \label{eq:phase-atom-representation}
\end{equation}
where $\lambda_F>0$, $c_F,b_{F,v}\in\Z_r$, and
$a_{F,uv}\in\{0,1\}$.  Extend $b_{F,v}$ and $a_{F,uv}$ by zero
outside the indicated scope.
The use of phase expressions over Boolean path variables follows the
sum-over-paths approach~\citep{DawsonEtAl2005,Amy2019}.
Define a symmetric matrix $A$ with zero diagonal by
\begin{equation}
    A[u,v]\coloneqq\sum_{F\in\mathcal F}a_{F,uv}\quad\text{in }\F_2
       \quad(u\ne v).
    \label{eq:phase-coefficient-data}
\end{equation}
The \emph{phase graph} is $P_C=(X,E)$ with adjacency matrix $A$.
Repeated quadratic terms cancel because $2\eta=0$ in $\Z_r$.
Isolated variables remain vertices.  The matrix $A$ records the parity
of the quadratic interactions.  The atoms retain their original scopes
and representative expressions.

\paragraph{Rank-width.}
\label{sec:cut-rank}
Let $G=(V,E)$ be a simple undirected graph and write $A_G$ for its
adjacency matrix over $\F_2$.
For $X\subseteq V$, let $\overline X\coloneqq V\setminus X$ and
define the \emph{cut-rank}
\[
    \gamma_G(X)\coloneqq
       \operatorname{rank}_{\F_2}A_G[X,\overline X].
\]
Then $\gamma_G(X)=\gamma_G(\overline X)$.

For $|V|\leq1$, we set $\rw(G)=\lrw(G)=0$.
Otherwise, a \emph{rank-decomposition} of $G$ is a pair $(T,L)$,
where $T$ is a subcubic tree, meaning that every vertex has degree
at most three, and $L:V\to\Leaves(T)$ is a bijection.
We require the normalization that $T$ has no degree-two vertices.
Removing an edge $e\in E(T)$ partitions the leaves and, through $L$,
induces a bipartition $(X_e,V\setminus X_e)$ of $V$.
The \emph{width} of $(T,L)$ is $\max_{e\in E(T)}\gamma_G(X_e)$,
and the \emph{rank-width} $\rw(G)$ is the minimum of this quantity
over all rank-decompositions of $G$~\citep{OumSeymour2006}.
Suppressing degree-two vertices preserves all induced cuts and leaves $O(|V|)$ nodes.
The \emph{linear rank-width} $\lrw(G)$ is the minimum, over vertex
orderings, of the maximum cut-rank of a prefix of the ordering.
Every such ordering gives a rank-decomposition of at most the same
width, so $\rw(G)\leq\lrw(G)$.

\subsection{Evaluation and rank-width bound}
The rank-decomposition specifies the branches to eliminate.  We show
that each branch produces a factor with a small $K_r$-expression.

\begin{theorem}
\label{thm:rw-sop-evaluation}
Fix an even integer $r\geq2$ and a finite quadratically
$r$-compatible gate set $\mathcal G$, together with its chosen
witnessing expressions.
Let $C$ be a circuit over $\mathcal G$, and let $(T,L)$ be a
rank-decomposition of its phase graph $P_C$ of width $k$.
Then $\widehat Q_C$ can be evaluated over $K_r$ using \InsideOut{}
in $O\bigl(4^k\poly(|C|)\bigr)$ time.
\end{theorem}
\begin{proof}
Let $X=V(P_C)$ and $A=A_{P_C}$.
Root the supplied rank-decomposition by subdividing an edge, and let
$X_t$ be the variables in the branch below $t$.
We eliminate these branches recursively by batch \InsideOut{}.
For each branch $t$, let $\mathcal B_t$ be the collection of original
atom occurrences touching $X_t$, let $E_t$ be the product of their
original $K_r$-expressions, and let
$R_t=(\bigcup_{F\in\mathcal B_t}S_F)\setminus X_t$.
Write $M_t(R_t)$ for the factor obtained by aggregating $E_t$ over
$X_t$.  In each product, every original atom occurrence is included
once.  Let $\lambda_t=\prod_{F\in\mathcal B_t}\lambda_F$.

We prove by induction that each branch update produces a factor
$M_t(R_t)$ represented by a $K_r$-expression of size
$2^k\poly(|C|)$ with at most $r2^k$ distinct bracketed
terms.  We simplify by multiplying phase terms and
combining equal phase terms using $K_r$ arithmetic.  The resulting
$K_r$-expression consists of terms $\alpha\phases{p}$ with distinct $p$.

For each distinct bracketed term we also retain one generating
assignment to $X_t$ and its multiplicity, the number of assignments
producing that term.  A term with multiplicity $m$ has coefficient
$\lambda_t m$ in $M_t$.  These data accompany the $K_r$-expression
when the message is passed to its parent.

At a leaf $t$ with variable $x$, the update multiplies the
$K_r$-expressions of the atoms touching $x$, substitutes $x=0$ and
$x=1$, and adds:
\[
    M_t(R_t)=E_t(0,R_t)\oplus E_t(1,R_t).
\]
Since the gate expressions have fixed size, multiplying the atoms
touching $x$ gives an expression $E_t$ of size $O(|C|)$.
Substituting a constant and simplifying cannot increase this size.
The two substitutions followed by addition therefore give an
expression of size at most $2\lVert E_t\rVert+1=O(|C|)$.  This proves the leaf case for every $k\geq0$.

At the leaf, each Boolean assignment initially has multiplicity one;
if the resulting bracketed terms agree, we add their multiplicities
and retain either assignment.  There are at most $2\leq r2^k$ entries.

Now let $t$ have children $u,v$, with messages $M_u(R_u)$ and
$M_v(R_v)$ represented by $K_r$-expressions
and the retained assignments and multiplicities above.
Since $X_t=X_u\mathbin{\dot\cup}X_v$ and
$\mathcal B_t=\mathcal B_u\cup\mathcal B_v$, the parent update is
\begin{equation}
    M_t(R_t)
       =\bigoplus_{\ba\in\{0,1\}^{X_u}}
          \bigoplus_{\bb\in\{0,1\}^{X_v}}
             E_t(\ba,\bb,R_t).
    \label{eq:rw-parent-update}
\end{equation}
We compute this update from the terms of $M_u$ and $M_v$, together
with their retained assignments and multiplicities.  For each pair
of terms, one from each child, take their representatives $\ba,\bb$ and multiplicities $m_u,m_v$,
substitute $(\ba,\bb)$ into $E_t$, and add $m_um_v$ to the
multiplicity of the resulting bracketed term.  Retain $(\ba,\bb)$
as its representative if the term is new.  Multiplying the final
multiplicities by $\lambda_t$ gives the coefficients of $M_t$.
Thus the child messages determine the representative pairs and their
weights; substitution into $E_t$ determines the parent term for each
pair.  The resulting $M_t$, with its retained assignments and
multiplicities, is then available for the next parent update.

To justify this computation, suppose two assignments to $X_u$
produce the same bracketed term in the child expression.  Their
substitutions into $E_u$ therefore agree.  Every atom in
$\mathcal B_t\setminus\mathcal B_u$ is independent of $X_u$, so
their substitutions into $E_t$ also agree, for every assignment to
$X_v\cup R_t$.  The same argument applies to $v$.  Consequently,
each pair of child entries represents $m_um_v$ assignments producing
the same parent term.  This proves that the computation evaluates
\eqref{eq:rw-parent-update}, including when an original atom touches
both children.

For the size bound, let $f_t$ be the sum modulo $r$ of the arguments
of all bracketed phase terms in the original multiplicative input
$K_r$-expressions for the atoms in $\mathcal B_t$, before aggregation.
For every assignment $\ba$ to $X_t$, quadraticity gives
\[
\begin{aligned}
    f_t(\ba,\by)
       ={}&f_t(\mathbf0,\by)
           +f_t(\ba,\mathbf0)-f_t(\mathbf0,\mathbf0)\\
          &+\frac r2\sum_{v\in R_t}y_v
             \left(\sum_{u\in X_t}A[u,v]a_u\bmod2\right)
             \pmod r.
\end{aligned}
\]
Indeed, the first line accounts for all monomials contained in
$R_t$ or $X_t$, counting the constant monomial once.  The second
line accounts for every monomial crossing from $X_t$ to $R_t$,
whose coefficient is $r/2$ times the corresponding entry of $A$.

The expression $f_t(\mathbf0,R_t)$ is fixed for this branch.
The residue $f_t(\ba,\mathbf0)-f_t(\mathbf0,\mathbf0)$ has at most
$r$ values.  The vector of parities in the second line is
$\ba^{\mathsf T}A[X_t,R_t]$ over $\F_2$ and has at most $2^k$
values, since $A[X_t,R_t]$ is a submatrix of a cut matrix of rank
at most $k$.
Consequently, after collecting coefficients modulo $r$, at most
$r2^k$ different expressions $\phases{f_t(\ba,R_t)}$ occur.
Simplifying equal phase terms using $K_r$ arithmetic leaves at most
this many terms $\alpha\phases{f_t(\ba,R_t)}$.
Substitution and collection of phase terms do not increase the
size of $E_t$, so each such term has size at most
$\lVert E_t\rVert+O(1)$.  Including the addition nodes, the entire
$K_r$-expression for $M_t(R_t)$ has size
\[
    O\!\left(r2^k(\lVert E_t\rVert+1)\right)
       =2^k\poly(|C|).
\]
This proves the claimed bound on the size after simplification.

Each child has at most $r2^k$ entries by the same bound (and the
leaf bound above), so a parent update examines at most $r^2 4^k$
pairs.  Substituting a representative assignment into $E_t$,
collecting coefficients modulo $r$, and looking up the resulting
bracketed term in a dictionary take $\poly(|C|)$ time per pair.
The stored assignments have length at most $|X|$, and the
multiplicities require only the arithmetic operations described
above.  Thus each branch update takes $O(4^k\poly(|C|))$ time.

The $O(|C|)$ branches therefore give
the claimed total running time.  At the root, multiplying the
resulting scalar by the original empty-scope atoms gives
$\widehat Q_C$.  For $|X|\leq1$, direct aggregation suffices.
\end{proof}

\begin{corollary}
\label{cor:rw-amplitude}
Under the hypotheses of Theorem~\ref{thm:rw-sop-evaluation},
the amplitude $\me{y}{U_C}{x}$ can be computed in
$O\bigl(4^k\poly(|C|)\bigr)$ time, given the
rank-decomposition $(T,L)$.
\end{corollary}

\begin{remark}
Corollary~\ref{cor:rw-amplitude} concerns a single amplitude with
fixed computational-basis boundaries.
The circuit FAQ also allows some boundary variables to remain free.
For example, for a four-qubit circuit $C$, fixing the output to
$0000$ and leaving the first and last input bits free gives the
amplitude function
\[
    A_C(x,y)\coloneqq\braket{0000|U_C|x00y},
    \qquad x,y\in\{0,1\}.
\]
Eliminating only the bound variables from the reduced query $\widehat Q$ over $K_r$, while retaining any equality condition on $x,y$,
yields a representation of a factor
$M_C:\{0,1\}^2\to K_r$ satisfying
\[
    A_C(x,y)=\nu_r\bigl(M_C(x,y)\bigr).
\]
The four amplitudes can therefore be recovered by substituting the
four assignments to $x,y$ and applying $\nu_r$.
The scalar theorem does not, by itself, bound the size of this
symbolic output representation.
Materializing a full output tensor on $f$ free Boolean variables
requires writing $2^f$ entries.
\end{remark}

\begin{example}
\label{ex:rw-four-clique}
Consider the four-qubit circuit $C$ consisting of a layer of
Hadamard gates, a $T$ gate on each qubit, all six $\CZ$ gates,
and a final layer of Hadamard gates, with input and output fixed
at $0000$.
After reduction, let $v_1,v_2,v_3,v_4$ be the remaining variables.
The input atoms have representatives
\[
\begin{aligned}
    \widetilde H_i^{-}(v_i)
       &=2^{-1/2}\phases{0},
    &\widetilde H_i^{+}(v_i)
       &=2^{-1/2}\phases{0},\\
    \widetilde T_i(v_i)
       &=\phases{v_i},
    &\widetilde{\CZ}_{ij}(v_i,v_j)
       &=\phases{4v_iv_j},
\end{aligned}
\]
for $1\leq i\leq4$ and $1\leq i<j\leq4$, respectively.
The two Hadamard atoms retain their declared scope $\{v_i\}$ even
though their representatives are constant after boundary substitution.
The lifted query is
\[
    \widehat Q_C()
       =\bigoplus_{v_1,v_2,v_3,v_4\in\{0,1\}}
          \left(
             \bigotimes_{i=1}^4
                \bigl(\widetilde H_i^{-}(v_i)
                   \otimes\widetilde H_i^{+}(v_i)
                   \otimes\widetilde T_i(v_i)\bigr)
             \otimes
             \bigotimes_{1\leq i<j\leq4}
                \widetilde{\CZ}_{ij}(v_i,v_j)
          \right).
\]
The phase graph is $P_C=K_4$ and has rank-width one.
Apply \InsideOut{} in the order $v_1,v_2,v_3,v_4$.
The first update consumes the two Hadamard atoms and the $T$ atom
on $v_1$, together with the three incident $\CZ$ atoms, and produces
\[
\begin{aligned}
    M_1(v_2,v_3,v_4)
       &\coloneqq\bigoplus_{v_1\in\{0,1\}}
          \widetilde H_1^{-}(v_1)\otimes\widetilde H_1^{+}(v_1)
          \otimes\widetilde T_1(v_1)
          \otimes\bigotimes_{j=2}^4
                    \widetilde{\CZ}_{1j}(v_1,v_j)\\
       &=\tfrac12\left(
             \phases{0}\oplus\phases{1}
                \otimes\bigotimes_{j=2}^4\phases{4v_j}
          \right).
\end{aligned}
\]
Its representative includes the normalization from the two consumed
Hadamard atoms.  The other input atoms remain active.
The next update is
\[
\begin{aligned}
    M_2(v_3,v_4)
       &\coloneqq\bigoplus_{v_2\in\{0,1\}}
          M_1(v_2,v_3,v_4)
          \otimes\widetilde H_2^{-}(v_2)\otimes\widetilde H_2^{+}(v_2)
          \otimes\widetilde T_2(v_2)
          \otimes\bigotimes_{j=3}^4
                    \widetilde{\CZ}_{2j}(v_2,v_j)\\
       &=\tfrac14\left(
             \mu_0\oplus\mu_1\otimes\phases{4v_3}\otimes\phases{4v_4}
          \right),
\end{aligned}
\]
where $\mu_0=\phases{0}\oplus\phases{6}$ and $\mu_1=2\phases{1}$.
Continuing with the atoms containing $v_3$ gives
\[
\begin{aligned}
    M_3(v_4)
       &\coloneqq\bigoplus_{v_3\in\{0,1\}}
          M_2(v_3,v_4)
          \otimes\widetilde H_3^{-}(v_3)\otimes\widetilde H_3^{+}(v_3)
          \otimes\widetilde T_3(v_3)
          \otimes\widetilde{\CZ}_{34}(v_3,v_4)\\
       &=\tfrac18\left(\theta_0\oplus\theta_1\otimes\phases{4v_4}\right),
\end{aligned}
\]
with
\[
    \theta_0=\mu_0\oplus\phases{5}\otimes\mu_1
              =\phases{0}\oplus3\phases{6},
    \qquad
    \theta_1=\mu_1\oplus\phases{1}\otimes\mu_0
              =3\phases{1}\oplus\phases{7}.
\]
Finally, eliminating $v_4$ consumes the remaining atoms and yields
\[
\begin{aligned}
    a=M_4()
       &\coloneqq\bigoplus_{v_4\in\{0,1\}}
          M_3(v_4)\otimes\widetilde H_4^{-}(v_4)
          \otimes\widetilde H_4^{+}(v_4)\otimes\widetilde T_4(v_4)\\
       &=\tfrac1{16}\bigl(
             \phases{0}\oplus4\phases{1}\oplus\phases{4}
             \oplus6\phases{6}\oplus4\phases{7}
          \bigr).
\end{aligned}
\]
Applying $\nu_8$ gives the amplitude
\[
    \braket{0000|U_C|0000}
       =\frac{1+4\omega_8+\omega_8^4+6\omega_8^6+4\omega_8^7}{16}.
\]
In this example the representatives simplify to two terms
at every nonfinal step.  The same updates can be batched: eliminating
$\{v_1,v_2\}$ produces $M_2$ directly, and then eliminating
$\{v_3,v_4\}$ produces $M_4$.
\end{example}

A rank-width dynamic program was also obtained
by~\citet{deColnet+2026} for \emph{quadratic sum-of-powers} (SOP)
representations.
Such a representation expresses a circuit amplitude, up to a
scalar factor, as a sum of terms $\omega_r^{p(\bx)}$ over Boolean
assignments $\bx$, where $p$ is a quadratic polynomial of the phase form considered here.
In our setting, this form arises simply by collecting the phases
in the product of the reduced gate factors of $\widehat Q_C$:
multiplication of their phase values becomes addition in the
exponent, and the graph recording the quadratic interactions
of $p$ is precisely $P_C$.

Our contribution is to obtain the rank-width bound within the
existing FAQ evaluation framework.
\InsideOut{} evaluates the original lifted atoms of $\widehat Q_C$
over $K_r$.  The rank-decomposition guides the branch batches, whose factor
expressions remain small, and $\nu_r$ recovers
the complex amplitude after the scalar evaluation.
Thus the gate-factor representation and the sum-of-powers
representation are two views of the same query.
The choice of semiring and message representation exposes the phase
structure within the sum-product evaluation scheme.
Earlier, \citet{Cheng+2025} established a linear-rank-width bound
for \textsc{FeynmanDD}, a quantum-circuit simulator based on
decision diagrams.

\section{Tensor compression and separator symmetries}
\label{sec:sopw}
We continue with quadratically $r$-compatible circuits and their lifted
reduced queries over $K_r$.  In addition to considering simply a circuit where the whole interaction graph has low rank-width, we now consider queries where we have many parts that individually have low rank-width, but the whole circuit does not. Instead we show how to combine compact evaluation of these local parts with structural observations on how to combine the individual representations into a compact global representation. Intuitively, we consider a tree decomposition where each bag has low rank-width (in the phase graph of the respective induced subcircuit for the bag). The complexity of classical simulation then depends on both the phase interactions within bags and the symmetries of the messages passed between them.

\begin{example}
\label{ex:compr-two-bags}
We will use the following circuit as a running example in this section. For $k\geq2$, use the circuit blocks of the form~\eqref{eq:standard-circuit-blocks} to
define the circuit on $[k]$
\begin{equation}
    C_k:=\Hblock{k};\Tblock{k};\CZblock{k};
         \Hblock{k};\CZblock{k};\Hblock{k}.
    \label{eq:compr-example-circuit}
\end{equation}
shown below for $k=3$:
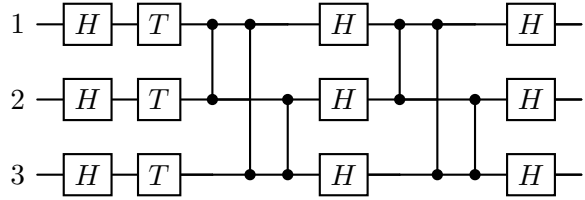
\begin{center}
\begin{quantikz}[row sep=0.45cm, column sep=0.35cm]
\lstick{$1$}
    & \gate{H}
    & \gate{T}
    & \ctrl{1}
    & \ctrl{2}
    & \qw
    & \gate{H}
    & \ctrl{1}
    & \ctrl{2}
    & \qw
    & \gate{H}
    & \qw
\\
\lstick{$2$}
    & \gate{H}
    & \gate{T}
    & \control{}
    & \qw
    & \ctrl{1}
    & \gate{H}
    & \control{}
    & \qw
    & \ctrl{1}
    & \gate{H}
    & \qw
\\
\lstick{$3$}
    & \gate{H}
    & \gate{T}
    & \qw
    & \control{}
    & \control{}
    & \gate{H}
    & \qw
    & \control{}
    & \control{}
    & \gate{H}
    & \qw
\end{quantikz}
\end{center}
We compute $\me{0^k}{U_{C_k}}{0^k}$.  On qubit $i$, the identity
dependencies of the diagonal gates identify all wire-segment variables
through $\Tblock{k};\CZblock{k}$ into one variable $a_i$, and all
wire-segment variables through the second $\CZblock{k}$ into one
variable $b_i$.  Thus the three Hadamard occurrences on that qubit have
input--output arguments $(0,a_i)$, $(a_i,b_i)$, and $(b_i,0)$,
respectively.  Set $A=\{a_i:i\in[k]\}$ and $B=\{b_i:i\in[k]\}$.

Consider a tree decomposition of the resulting phase graph with root
bag $B$ and child bag $A\cup B$.  The child bag subcircuit contains the
first two $H$ layers, the $T$ layer, and the first $\CZ$ layer; its
separator is $B$.  With all bag variables free, its $K_8$-valued factor is
\[
    W(\mathbf a,\mathbf b)
       =2^{-k}\phases{
          \sum_i a_i+4\sum_i a_ib_i+4\sum_{i<j}a_ia_j}.
\]
The scalar $2^{-k}$ accounts for its $2k$ Hadamard occurrences.
Summing the $A$ variables gives the outgoing message.
For $S\subseteq B$, put
$\chi_S(\mathbf b):=\bigotimes_{b_i\in S}\phases{4b_i}$.
Expanding the message into phase monomials gives
\[
    M(\mathbf b)
       =\bigoplus_{\mathbf a\in\{0,1\}^k}
          \mu_{\mathbf a}\otimes\chi_{S_{\mathbf a}}(\mathbf b),
    \qquad S_{\mathbf a}:=\{b_i:a_i=1\},\qquad
    \mu_{\mathbf a}
      =2^{-k}\phases{|\mathbf a|+4\binom{|\mathbf a|}{2}}.
\]
Here $|\mathbf a|:=\sum_i a_i$.  The $2^k$ indices
$S_{\mathbf a}$ are distinct, and every $\mu_{\mathbf a}$ is nonzero.
Thus this phase-monomial representation has $2^k$ terms.
\end{example}

The message in the example has $2^k$ terms. We will introduce a tensor representation below and show that the corresponding tensor has an exact tree tensor network with $O(k)$ arrays and bounded
bond dimension.

\subsection{Phase coefficient tensors}
\label{sec:tree-tensor-compression}

We introduce our new notion of a phase coefficient tensor, which stores the coordinates of an atom's $K_r$-values. Multiplication of atoms becomes convolution in the residue coordinate, while aggregation sums over Boolean coordinates.  These are the two operations required by \InsideOut{}, which will again form the core of our algorithm.

Fix an even integer $r\geq2$.  The \emph{phase coefficient tensor} of a factor
$F:\{0,1\}^{J}\to K_r$ is the real array
$A_F:\{0,1\}^{J}\times\Z_r\to\R$ determined by
\begin{equation}
    F(\bx)=\sum_{a\in\Z_r}A_F(\bx,a)\phases a.
    \label{eq:coefficient-tensor-interpretation}
\end{equation}
The coefficients are unique because equality in $K_r$ is coefficientwise.
The residue coordinate records the formal phase basis element.

Addition of factors on a fixed scope adds their coefficient tensors
componentwise.  For coefficient tensors $A,B$ of factors on scopes $J,L$,
the product factor has coefficient tensor $A\ostar B$ on $J\cup L$, where
\begin{equation}
    (A\ostar B)(\bx,a)
    :=\sum_{\substack{b,c\in\Z_r\\b+c=a}}
         A(\bx|_J,b)B(\bx|_L,c),
    \label{eq:coefficient-tensor-product}
\end{equation}
where residue addition is modulo $r$.  Summing a Boolean variable $v\in J$
gives the coefficient tensor $\Sigma_v A$ on $J\setminus\{v\}$, where
\begin{equation}
    (\Sigma_v A)(\bx,a)
    :=\sum_{b\in\{0,1\}}A(\bx\cup\{v\mapsto b\},a).
    \label{eq:coefficient-tensor-marginal}
\end{equation}
For $U\subseteq J$, write $\Sigma_U A$ for successive summation over the
variables in $U$; these sums commute.  Both formulas follow directly from
the operations of $K_r$.

\subsection{Exact TTN representations}

We represent coefficient tensors by exact tree tensor networks.
A coefficient layout specifies how their variable coordinates are
arranged, and the bond dimensions determine the storage and update
costs.  We first establish exact compression and the operations needed
by \InsideOut{}.

We follow the leaf-indexed dimension-tree TTN format of
\citet[Section~2.2]{CerutiKressnerSulz25}.  They use arrays over $\C$
with leaf indices in sets $[n_i]$; below we write the same contraction
over a field $\mathbb F$ with leaf indices in finite sets $D_i$.

\begin{definition}[Tree tensor network]
\label{def:tree-tensor-network}
A \emph{tree tensor network} (TTN) over a field $\mathbb F$ consists of a
rooted tree $\mathcal T$ whose internal vertices have at least two children
and whose leaves are labelled bijectively by a finite nonempty set $I$.
For each $i\in I$, let $D_i$ be a finite nonempty set.  For each
edge $e$ of $\mathcal T$, choose a finite nonempty bond-index set $H_e$.
Identify each leaf with its label in $I$.  At a leaf $i$ with incident edge $e_i$,
choose an array $U_i:D_i\times H_{e_i}\to\mathbb F$.  At each internal
vertex $u$, including the root, choose an array
$C_u:\prod_{e\ni u}H_e\to\mathbb F$.  If the tree has just one
vertex $i$, choose instead an array $U_i:D_i\to\mathbb F$.
These data represent $\Theta:\prod_{i\in I}D_i\to\mathbb F$ by
\begin{equation}
    \Theta(\mathbf x)
    :=\sum_{\mathbf h\in\prod_{e\in E(\mathcal T)}H_e}
       \left(\prod_{i\in I}U_i(x_i,h_{e_i})\right)
       \left(\prod_{u\in V(\mathcal T)\setminus I}
          C_u\bigl((h_e)_{e\ni u}\bigr)\right),
    \label{eq:ttn-contraction}
\end{equation}
For a one-vertex tree, the formula means $\Theta(x_i)=U_i(x_i)$.
Each $h_e$ occurs at the two endpoints of $e$ and is summed over
$H_e$; $d_e:=|H_e|$ is the \emph{bond dimension} of $e$.  The representation is
\emph{exact} for an array $A$ if $\Theta(\mathbf x)=A(\mathbf x)$ for
every $\mathbf x$.
\end{definition}

For a TTN $\mathcal N$ with underlying tree $\mathcal T$ and edge
dimensions $(d_e)_{e\in E(\mathcal T)}$, define its \emph{maximum bond
dimension} by
\[
    \mathsf{bd}(\mathcal N)
    :=\max\bigl(\{1\}\cup\{d_e:e\in E(\mathcal T)\}\bigr).
\]
Thus $\mathsf{bd}(\mathcal N)=1$ when $\mathcal T$ has one vertex.

Let $V$ be the full variable set.  A \emph{coefficient layout} $\mathcal R$
is a rooted tree in which every vertex has at most two children and whose
leaves are labelled bijectively by $V$.
For $J\subseteq V$, let $\mathcal R[J]$ be the minimal subtree containing
the root and the leaves in $J$;
for $J=\emptyset$ it consists only of the root.

For each scope $J\subseteq V$, the layout determines a tree $\mathcal T_J$
for representing coefficient tensors on $J$.  Suppress the unlabelled vertices of
$\mathcal R[J]$ with one child.  If $J\ne\emptyset$, add a new root
whose children are the remaining tree and a new leaf labelled $*$.
If $J=\emptyset$, use the one-vertex tree labelled $*$.  Denote this
rooted tree by $\mathcal T_J$ and observe that its leaves are
$J\mathbin{\dot\cup}\{*\}$.  A TTN $\mathcal N$ on this tree has
scalar field $\R$, leaf domains
$D_v=\{0,1\}$ for $v\in J$ and $D_*=\Z_r$.
The leaf $*$ exposes the residue coordinate.  For a coefficient tensor
$A$ on scope $J$, an exact representation on $\mathcal T_J$ means choosing
bond-index sets and local arrays as in
Definition~\ref{def:tree-tensor-network} whose contraction equals
$A(\bx,a)$.  The tree and leaf domains depend only on the scope and layout;
the bond dimensions and arrays depend on $A$.

Let $A$ be a coefficient tensor on scope $J$.
For each $e\in E(\mathcal R[J])$, let $J_e\subseteq J$ be the labels of
the leaves in the component of $\mathcal R[J]-e$ that does not contain
the root.  Define $\operatorname{Mat}_e(A)$ to be the matrix with rows
indexed by $\bx\in\{0,1\}^{J_e}$, columns indexed by
$(\by,a)\in\{0,1\}^{J\setminus J_e}\times\Z_r$, and entry
$A(\bx\cup\by,a)$ in row $\bx$ and column $(\by,a)$.  Set
\begin{equation}
    D_{\mathcal R}(A)
    :=\max\bigl(\{1\}\cup
       \{\operatorname{rank}_{\R}\operatorname{Mat}_e(A) \mid
                                      e\in E(\mathcal R[J])\}\bigr).
    \label{eq:separator-tree-rank}
\end{equation}
That is, $D_{\mathcal R}(A)$ is the largest matrix rank across a variable edge of the coefficient layout. The residue index is always on the root side of the cut.

We show an analogue of the edge-cut rank characterization for tree tensor
networks of \citet[Theorem~8.8]{YeLim18} in the present leaf-indexed setting,
including exact reduction from a supplied real TTN.

\begin{lemma}
\label{lem:exact-tree-reduction}
Given an exact TTN $\mathcal N$ for $A$ on $\mathcal T_J$, one can
compute either that $A=0$ or an exact TTN $\mathcal N'$ for $A$ on
$\mathcal T_J$ with
$\mathsf{bd}(\mathcal N')\leq\max\{r,D_{\mathcal R}(A)\}$ in
$\poly(|\mathcal R|,r,\mathsf{bd}(\mathcal N))$ time.
\end{lemma}
\begin{proof}
We use two sweeps.  The upward sweep removes redundant coordinates in
each subtree; the downward sweep removes those invisible to the rest of
the tree.  Together they reduce every bond to the rank of its flattening.
Write $B=\mathsf{bd}(\mathcal N)$.  If $J=\emptyset$, the TTN is the
one-vertex array $A(a)$; testing its $r$ entries proves the claim.  Otherwise,
for an edge $e$ directed from a vertex $u$ to its parent, let $L_e$ be the
matrix whose rows are assignments to the leaves below $e$, whose columns are
bond indices at $e$, and whose entries contract the part of the TTN below
$e$.  Let $R_e$ have rows indexed by assignments to the remaining leaves,
columns indexed by the same bond indices, and entries obtained by
contracting the part above $e$ while leaving its bond index open.  The full
flattening across $e$ is $L_eR_e^{\mathsf T}$.  These matrices describe
the invariant; we never form them explicitly.  All computations use the
local arrays.  Each update below replaces the two arrays incident to an
edge while preserving their contraction.  Thus both component maps at
every other edge remain unchanged.

Process the vertices other than the root from the leaves upward.  At $u$,
flatten its local array into a matrix $M_u$ with the parent bond as columns
and its other indices as rows (the Boolean or residue index if $u$ is a leaf,
and the child bonds otherwise).  The previously processed child maps are
injective, so $L_e=(\bigotimes_f L_f)M_u$, with $f$ ranging over the child
edges of $u$, has rank
$s=\operatorname{rank}_{\R}M_u$; at a leaf the tensor product is the
identity on its physical index.  If $s=0$, then $A=0$.  Otherwise compute a
rank factorization $M_u=P_uQ_u$ with $P_u$ having $s$ independent columns.
Use $P_u$ as the new array at $u$, give $e$ bond dimension $s$, and multiply
$Q_u$ into the parent array along its old $e$ index.  This preserves the
contraction and makes the new $L_e$ injective.  After the sweep, the tensor
product of the child maps at the root is injective; hence $A=0$ exactly when
the root array is zero, which we test directly.

Suppose $A\ne0$.  Process the edges from the root downward.  At a parent
vertex $u$ and one of its child edges $e$, flatten the current array at $u$
into a matrix $M_e$ with the $e$ index as columns and all other indices as
rows.  The map from these other indices to the leaves outside $e$ is
injective: it is the tensor product of the already processed complementary
map at the parent edge (or the scalar identity at the root) and the injective
subtree maps at the other child edges.  Thus $R_e$ has the same rank $s$ as
$M_e$.  Factor $M_e=P_eQ_e$, where $P_e$ has $s$ independent columns and
$Q_e$ has $s$ independent rows.  Replace the array at $u$ by $P_e$, give
$e$ dimension $s$, and multiply $Q_e$ into the array at the child of $e$.
The new complementary map is injective, and the new subtree map
$L_eQ_e^{\mathsf T}$ is injective because both factors are.  Since maps
at other edges are unchanged, this maintains all subtree maps injective
and makes the complementary map injective at each processed edge.
At the end both maps at every edge are injective, so its bond dimension is
$\operatorname{rank}_{\R}(L_eR_e^{\mathsf T})$, the rank of the
unchanged full flattening of $A$ across $e$.

Each edge of $\mathcal T_J$ within the variable subtree separates a set
$J_e$ from the remaining variables and the residue index.  Its flattening
has rank at most $D_{\mathcal R}(A)$: the corresponding cut is an edge of
$\mathcal R[J]$ before unary vertices are suppressed.  The two edges at the
new root separate the residue index from all Boolean indices and have rank
at most $r$.  Every local matrix used above has at most $\max\{r,B^2\}$
rows and at most $B$ columns; there are $O(|\mathcal R|)$ such matrices.
Exact rank factorizations and the array updates therefore take
$\poly(|\mathcal R|,r,B)$ time.
\end{proof}

\begin{lemma}
\label{lem:coefficient-tree-operations}
Fix a coefficient layout $\mathcal R$.  Let $A,B$ be coefficient tensors
of factors on scopes $J,L$, and let $v\in J$.  Given exact TTNs
$\mathcal N_A$ for $A$ on $\mathcal T_J$ and
$\mathcal N_B$ for $B$ on $\mathcal T_L$, one can
construct a TTN $\mathcal N_{\ostar}$ for $A\ostar B$ on
$\mathcal T_{J\cup L}$ and a TTN $\mathcal N_\Sigma$ for $\Sigma_vA$
on $\mathcal T_{J\setminus\{v\}}$
such that
\[
    \mathsf{bd}(\mathcal N_{\ostar})
       \leq\mathsf{bd}(\mathcal N_A)\mathsf{bd}(\mathcal N_B),
    \qquad
    \mathsf{bd}(\mathcal N_\Sigma)\leq\mathsf{bd}(\mathcal N_A),
\]
and the construction takes
$\poly(|\mathcal R|,r,
   \mathsf{bd}(\mathcal N_A)\mathsf{bd}(\mathcal N_B))$
time.
\end{lemma}
\begin{proof}
Put $H=J\cup L$.  For either scope $S\subseteq H$, the tree $\mathcal T_S$
is obtained from $\mathcal T_H$ by retaining the leaves $S\cup\{*\}$,
deleting the other branches, and suppressing vertices with one child.
Consequently a TTN on $\mathcal T_S$ extends to $\mathcal T_H$ without
increasing a bond dimension.  Each edge $e$ of $\mathcal T_S$ corresponds
to a path in $\mathcal T_H$: give every edge of that path dimension $d_e$
and put the array $\mathbf 1[p=q]$ on its two path indices at each
restored vertex.  Give all edges in deleted branches dimension one and put
the constant-one array at their vertices.  At vertices already present in
$\mathcal T_S$, retain the original array and ignore any added
dimension-one indices.  If
$S=\emptyset$, put its one-vertex array at $*$, assign dimension one to
every edge, and use all-one arrays elsewhere.  Denote the resulting TTNs on
$\mathcal T_H$ by $\widetilde{\mathcal N}_A$ and
$\widetilde{\mathcal N}_B$; their contractions are respectively
$A(\bx|_J,a)$ and $B(\bx|_L,a)$.

For each edge of $\mathcal T_H$, use the pair $(h_e^A,h_e^B)$ as its new
bond index.  At an internal vertex, multiply the two entries with these
paired indices.  At a Boolean leaf $w$, set
\[
    U_w^{\ostar}(x_w,(h^A,h^B))
      =U_w^A(x_w,h^A)U_w^B(x_w,h^B),
\]
and at the residue leaf set
\[
    U_*^{\ostar}(a,(h^A,h^B))
      =\sum_{\substack{b,c\in\Z_r\\b+c=a}}
          U_*^A(b,h^A)U_*^B(c,h^B).
\]
Expanding the contraction of these arrays and distributing the finite sums
over all paired bond indices gives exactly
$\sum_{b+c=a}A(\bx|_J,b)B(\bx|_L,c)$, which is
$(A\ostar B)(\bx,a)$ by~\eqref{eq:coefficient-tensor-product}.  Each new
bond dimension is the product of the two old dimensions.

For marginalization, replace the array at leaf $v$ of $\mathcal N_A$ by
$s(h)=\sum_{x_v\in\{0,1\}}U_v(x_v,h)$ and contract $s$ into the array at
its parent.  Delete $v$ and its incident edge.  If an unlabelled vertex now
has one child, contract its two-index array into the adjacent array and
suppress that vertex; repeat until the tree is $\mathcal T_{J\setminus\{v\}}$.
When $J=\{v\}$, this leaves the one-vertex array at $*$.  The contraction
of the resulting TTN is
$\sum_{x_v\in\{0,1\}}A(\bx\cup\{v\mapsto x_v\},a)$,
and no surviving bond dimension increases.  All arrays have at most three
bond indices, each of dimension at most
$\mathsf{bd}(\mathcal N_A)\mathsf{bd}(\mathcal N_B)$, and at most one
physical index of size $r$.  There are $O(|\mathcal R|)$ arrays, so both
constructions have the claimed running time.
\end{proof}

\subsection{Local factors on a tree decomposition}

Let $C$ be a quadratically $r$-compatible circuit with consistent fixed
computational-basis boundaries.  Write $G=P_C=(V,E)$ for its phase graph
and define the collected unary coefficients by
\begin{equation}
    b_v\coloneqq\sum_{F\in\mathcal F}b_{F,v}\quad\text{in }\Z_r,
    \qquad v\in V,
    \label{eq:phase-unary-coefficients}
\end{equation}
using the atom coefficients of~\eqref{eq:phase-atom-representation}.
We evaluate the reduced scalar circuit query $\widehat Q_C$ over $K_r$.
In this section we construct local bag factors by assigning the terms
of the input representatives to bags; their product agrees pointwise
with the product of the input atoms.

Fix a rooted tree decomposition of $P_C$ with tree $T$, bags $B_t$, and at
most two children per node.
Also fix a coefficient layout $\mathcal R$ on $V$.
The tree $T$ determines the query updates.  The layout $\mathcal R$
determines how their factors and messages are represented.
For a non-root node $t$ with parent $p(t)$, set
$S_t=B_t\cap B_{p(t)}$ and $U_t=B_t\setminus S_t$.  At the root, set
$S_t=\emptyset$ and $U_t=B_t$.
The \emph{bag subcircuit} $C_t$ is the local query induced by $B_t$,
with separator-only terms assigned above $t$.  It has variables $B_t$, unary coefficients
$b(v)$ for $v\in U_t$, and quadratic interactions of coefficient $r/2$
\[
    E_t:=E(P_C[B_t])\setminus E(P_C[S_t]).
\]
Thus $E_t$ contains the edges within $B_t$ that have at least one
endpoint in $U_t$.  Edges wholly inside the separator are assigned
above $t$.
Assign each gate occurrence's normalization and constant phase to one
bag, counting each contribution exactly once.
Write $\lambda_t\phases{c_t}$ for their product within bag $t$.
Each input normalization and constant phase thus contributes to exactly
one local factor.
The phase graph of $C_t$ is
$P_{C_t}=(B_t,E_t)$.  Let $W_t:\{0,1\}^{B_t}\to K_r$ be the value of
$C_t$ with all bag variables free, and write $\mathsf W_t$ for its
coefficient tensor.
Running intersection assigns each vertex to exactly one $U_t$ and each
edge of $G$ to exactly one $E_t$.

Take $J=B_t$ in the coefficient-layout notation above, so $J_e$ is
defined for every $e\in E(\mathcal R[J])$.  Set
\begin{equation}
    \rho_t(\mathcal R)
    :=\max_{e\in E(\mathcal R[J])}
      \gamma_{P_{C_t}}(J_e),
    \label{eq:local-coefficient-tree-width}
\end{equation}
with maximum zero when the restricted tree has no edges.

We will then use the following bound to establish the initial ``small'' representations on which the algorithm then operates.

\begin{lemma}[Bag tensor bound]
\label{lem:bag-tensor-bound}
The tensor $\mathsf W_t$ has an exact TTN on
$\mathcal T_{B_t}$ with maximum bond dimension at most
$r2^{\rho_t(\mathcal R)}$.  Such a TTN can be obtained in
\[
    r^{O(1)}2^{O(\rho_t(\mathcal R))}
       \poly(|\mathcal R|+|V|+|E|+|C|)
\]
time.
\end{lemma}
\begin{proof}
Each bond in the variable subtree will record two pieces of information
about the variables below it: their internal phase residue and their interaction signature
with the remaining variables.  These suffice to combine two child
subtrees without retaining their Boolean assignments.
Write $J=B_t$, $H=P_{C_t}$, and put $\beta_v=b(v)$ for $v\in U_t$ and
$\beta_v=0$ for $v\in S_t$.  For
$X\subseteq J$ and $\bx\in\{0,1\}^X$, define
\[
    p_X(\bx)
      :=\sum_{v\in X}\beta_vx_v
        +\frac r2\sum_{\{v,w\}\in E(H[X])}x_vx_w
          \quad\text{in }\Z_r,
    \qquad
    f_X(\bx)
      :=\bx^{\mathsf T}A_H[X,J\setminus X]
          \quad\text{in }\F_2^{J\setminus X}.
\]
The coefficient tensor to be represented is therefore
\begin{equation}
    \mathsf W_t(\bx,a)
      =\lambda_t\mathbf 1[a=c_t+p_J(\bx)].
    \label{eq:bag-coefficient-entry}
\end{equation}
If $J=\emptyset$, take the one-vertex TTN with
$U_*(a)=\lambda_t\mathbf 1[a=c_t]$.  Assume henceforth that $J\ne\emptyset$.

Let $K$ be the variable subtree of $\mathcal T_J$, with root $q$; thus
$\mathcal T_J$ is obtained by attaching $K$ and the leaf $*$ to a new root.
For each vertex $u$ of $K$, let $X_u\subseteq J$ be the labels below $u$.
Identify the bond index on the edge from $u$ to its parent with the finite
set
\[
    H_u:=\Z_r\times\operatorname{im}f_{X_u}.
\]
Since $X_q=J$, the second component of $H_q$ is the zero vector.
Identify the bond index at $*$ with $\Z_r$ and set
$U_*(a,h)=\mathbf 1[a=h]$.

We choose the arrays so that, for every assignment
$\bx\in\{0,1\}^{X_u}$, the subtree rooted at $u$ contracts to
\begin{equation}
    \mathbf 1\!\left[
      h_u=\left(p_{X_u}(\bx),f_{X_u}(\bx)\right)\right],
    \label{eq:bag-ttn-invariant}
\end{equation}
with only its parent bond $h_u$ left open.
For a variable leaf $u$ labelled $v$, this determines its array:
\[
    U_v(x,h_u)
      :=\mathbf 1\bigl[
          h_u=(p_{\{v\}}(x),f_{\{v\}}(x))\bigr].
\]
For an internal vertex $u$ of $K$ with children $v,w$, choose, for each
$s\in\operatorname{im}f_{X_v}$ and
$s'\in\operatorname{im}f_{X_w}$, representatives
$z_v(s)\in\F_2^{X_v}$ and $z_w(s')\in\F_2^{X_w}$ with the indicated
signatures.  For $h_v=(a_v,s)$ and $h_w=(a_w,s')$, define
\[
\begin{aligned}
    a_u&:=a_v+a_w+
       \frac r2\,z_v(s)^{\mathsf T}A_H[X_v,X_w]z_w(s')
            &&\text{in }\Z_r,\\
    s_u&:=s|_{J\setminus X_u}+s'|_{J\setminus X_u}
            &&\text{in }\F_2^{J\setminus X_u},\\
    C_u(h_u,h_v,h_w)&:=\mathbf 1[h_u=(a_u,s_u)].
\end{aligned}
\]
The first expression is independent of the representatives: a change in
$z_v(s)$ lies in the kernel of $f_{X_v}$ and hence has zero interaction
with every assignment on $X_w$; the same holds with $v,w$ exchanged.
At the new root, with incident indices $h_q=(a_q,0)$ and $h_*$, set
\[
    C_{\mathrm{new}}(h_q,h_*)
       :=\lambda_t\mathbf 1[h_*=a_q+c_t].
\]

To verify~\eqref{eq:bag-ttn-invariant}, the leaf case holds by construction.
At an internal vertex the child indicators fix their two bond states.
The merge adds their internal phases and the interactions between the
children, and adds their signatures restricted to the remaining variables.
This gives exactly $(p_{X_u}(\bx),f_{X_u}(\bx))$, proving the invariant
by induction.  Contracting the new root and
$U_*$ now gives~\eqref{eq:bag-coefficient-entry}, so the TTN is exact.

For $u\ne q$, the edge above $u$ corresponds to an edge $e$ of
$\mathcal R[J]$ with $X_u=J_e$.  Consequently
$|H_u|=r2^{\gamma_H(X_u)}\leq r2^{\rho_t(\mathcal R)}$.
The two edges incident to the new root have dimension $r$, giving the
claimed maximum.  The index sets and representatives follow from binary
Gaussian elimination.  The tree $\mathcal T_J$ has $O(|J|)$ vertices,
and each array has at most three bond indices of dimension at most
$r2^{\rho_t(\mathcal R)}$.  Constructing $\mathcal T_J$ and enumerating
the entries of these arrays therefore takes the stated time.
\end{proof}

\begin{example}
\label{ex:compr-two-bag-ttn}
Use a coefficient layout $\mathcal R$ in which $a_i$ and $b_i$ are
sibling leaves for each $i\in[k]$, and join these pairs by a binary
tree.  Figure~\ref{fig:compr-example-ttn} shows the resulting TTN for
$k=4$.  Its bonds retain a phase residue $z\in\Z_8$ and one bit
$\pi\in\F_2$ describing the interaction with variables outside the
subtree.  For the child bag $J=A\cup B$, a proper subtree containing the
pairs indexed by $I\subsetneq[k]$ has cross-interaction signature
\[
    \pi_I=\sum_{i\in I}a_i\quad\text{in }\F_2.
\]
Indeed, its only edges to variables outside those pairs are the
clique edges from $a_i$, $i\in I$, to $a_j$, $j\notin I$.
Thus the bond set $H_u$ of Lemma~\ref{lem:bag-tensor-bound} can be
identified with $\Z_8\times\F_2$ on each proper pair subtree.  The
same two-component bond set applies at the individual leaves: the
signature of $a_i$ records its interactions with $b_i$ and the other
$a_j$, while that of $b_i$ records its interaction with $a_i$.

In the figure, the normalizations of $H(0,a_i)$ and $H(a_i,b_i)$ are
placed at $U_{a_i}$ and $C_i$, respectively.  Contraction of each pair
subtree gives
\[
    2^{-1}\mathbf1[z=a_i+4a_ib_i,\ \pi=a_i].
\]
When disjoint groups of pairs $I$ and $L$ meet, the merge array is
\[
    C_{I,L}\bigl((z,\pi),(z_I,\pi_I),(z_L,\pi_L)\bigr)
       :=\mathbf1[\pi=\pi_I+\pi_L]\,
         \mathbf1[z=z_I+z_L+4\pi_I\pi_L],
\]
with residue arithmetic in $\Z_8$ and parity arithmetic in $\F_2$.
The term $4\pi_I\pi_L$ collects the $\CZ$ interactions between the two
groups.  At the top, omit the outgoing parity and use equality arrays
at the new root and residue leaf.  This specializes the construction of
Lemma~\ref{lem:bag-tensor-bound} to the child bag subcircuit, with its
Hadamard normalizations placed at the indicated arrays.

The TTN has $O(k)$ arrays and maximum bond dimension $16$; its two top
bonds have dimension $8$.  By Lemma~\ref{lem:coefficient-tree-operations},
summing the $a_i$ gives an exact TTN for the coefficient tensor of
$M(\mathbf b)$ with maximum bond dimension at most $16$ as well.
Thus the outgoing message has a constant-bond representation despite
its $2^k$ phase monomials.
\end{example}

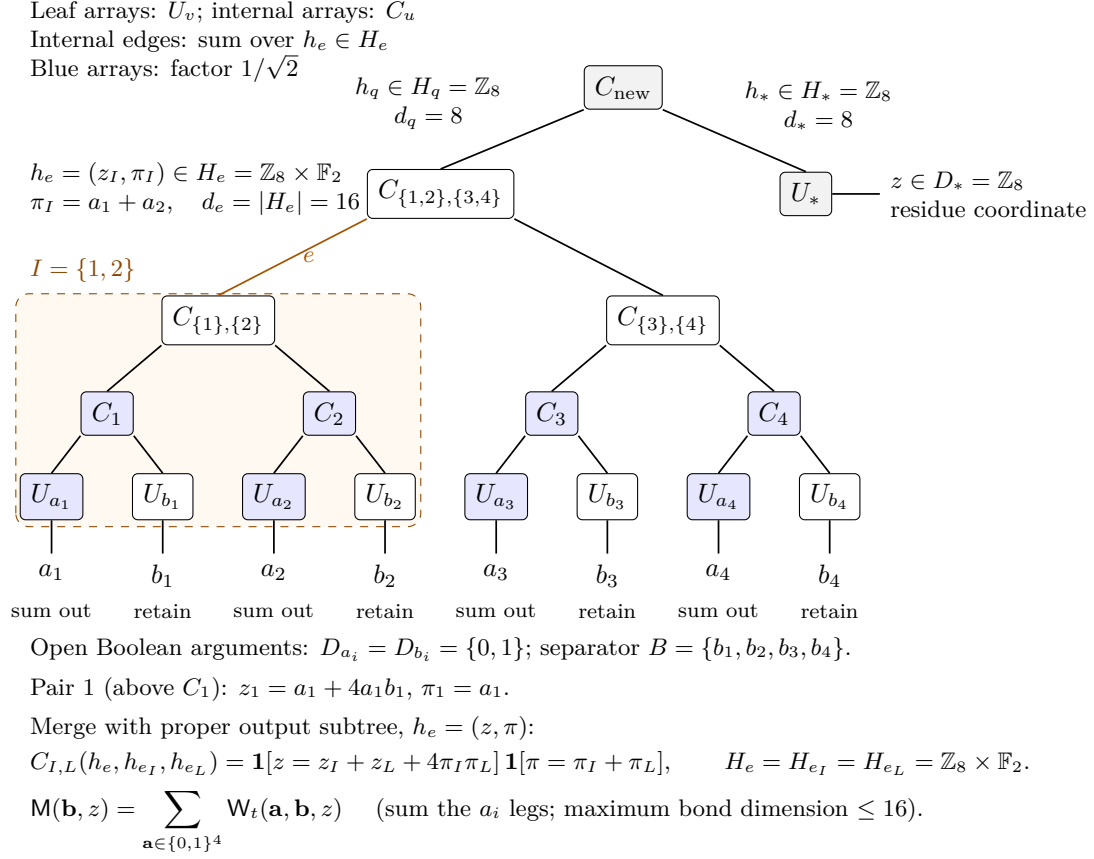
\begin{figure}[t]
\centering
\begin{tikzpicture}[
    x=1.05cm,y=0.88cm,
    every node/.style={font=\small},
    tensor/.style={draw,rounded corners=2pt,fill=white,
        minimum height=0.48cm,inner sep=4pt},
    weighted/.style={tensor,fill=blue!10},
    readout/.style={tensor,fill=black!5},
    bond/.style={line width=0.65pt},
    physical/.style={line width=0.65pt},
    note/.style={align=left,font=\footnotesize}
]
    \filldraw[fill=orange!5,draw=orange!55!black,dashed,rounded corners=5pt]
        (-0.45,-0.45) rectangle (4.65,3.05);
    \node[note,anchor=south west,text=orange!55!black] at (-0.4,3.08)
        {$I=\{1,2\}$};
    \foreach \i/\x in {1/0,2/2.8,3/5.6,4/8.4} {
        \node[weighted] (a\i) at (\x,0) {$U_{a_\i}$};
        \node[tensor] (b\i) at (\x+1.4,0) {$U_{b_\i}$};
        \draw[physical] (a\i.south) -- ++(0,-0.5)
            node[below] {$a_\i$}
            node[below=15pt,font=\scriptsize] {sum out};
        \draw[physical] (b\i.south) -- ++(0,-0.5)
            node[below] {$b_\i$}
            node[below=15pt,font=\scriptsize] {retain};
        \node[weighted] (p\i) at (\x+0.7,1.25) {$C_\i$};
        \draw[bond] (a\i.north) -- (p\i.south west);
        \draw[bond] (b\i.north) -- (p\i.south east);
    }
    \node[tensor] (g12) at (2.1,2.65) {$C_{\{1\},\{2\}}$};
    \node[tensor] (g34) at (7.7,2.65) {$C_{\{3\},\{4\}}$};
    \draw[bond] (p1.north) -- (g12.south west);
    \draw[bond] (p2.north) -- (g12.south east);
    \draw[bond] (p3.north) -- (g34.south west);
    \draw[bond] (p4.north) -- (g34.south east);

    \node[tensor] (q) at (4.9,4.55) {$C_{\{1,2\},\{3,4\}}$};
    \draw[bond,orange!65!black] (g12.north) -- (q.south west)
        node[midway,right,font=\footnotesize] {$e$};
    \draw[bond] (g34.north) -- (q.south east);
    \node[note,anchor=south west] at (-0.4,4.08) {
        $h_e=(z_I,\pi_I)\in H_e=\Z_8\times\F_2$\\
        $\pi_I=a_1+a_2,\quad d_e=|H_e|=16$
    };
    \node[readout] (root) at (7.2,6.15) {$C_{\mathrm{new}}$};
    \node[readout] (residue) at (9.5,4.55) {$U_*$};
    \draw[bond] (q.north) -- (root.south west)
        node[midway,above left,font=\footnotesize,align=center]
            {$h_q\in H_q=\Z_8$\\$d_q=8$};
    \draw[bond] (residue.north) -- (root.south east)
        node[midway,above right,font=\footnotesize,align=center]
            {$h_*\in H_*=\Z_8$\\$d_*=8$};
    \draw[physical] (residue.east) -- ++(0.6,0)
        node[right,font=\footnotesize,align=left]
            {$z\in D_*=\Z_8$\\residue coordinate};
    \node[note,anchor=north west] at (-0.4,7.6) {
        Leaf arrays: $U_v$; internal arrays: $C_u$\\
        Internal edges: sum over $h_e\in H_e$\\
        Blue arrays: factor $1/\sqrt{2}$
    };

    \node[note,anchor=north west] at (-0.4,-1.95) {
        Open Boolean arguments: $D_{a_i}=D_{b_i}=\{0,1\}$;
        separator $B=\{b_1,b_2,b_3,b_4\}$.
    };
    \node[note,anchor=north west] at (-0.4,-2.55) {
        Pair 1 (above $C_1$): $z_1=a_1+4a_1b_1$, $\pi_1=a_1$.
        \\[4pt]
        Merge with proper output subtree, $h_e=(z,\pi)$:
        \\[2pt]
        $\displaystyle C_{I,L}(h_e,h_{e_I},h_{e_L})
        =\mathbf1[z=z_I+z_L+4\pi_I\pi_L]\,
         \mathbf1[\pi=\pi_I+\pi_L],
         \qquad H_e=H_{e_I}=H_{e_L}=\Z_8\times\F_2.$
        \\[4pt]
        $\displaystyle \mathsf M(\mathbf b,z)
          =\sum_{\mathbf a\in\{0,1\}^4}
             \mathsf W_t(\mathbf a,\mathbf b,z)$
        \quad (sum the $a_i$ legs; maximum bond dimension $\leq16$).
    };
\end{tikzpicture}
\caption{Exact coefficient TTN for the child bag $A\cup B$, with $k=4$
(Example~\ref{ex:compr-two-bag-ttn}).
Each proper pair subtree carries a phase residue and one parity bit.
Summing the $a_i$ legs yields the separator message without increasing
the maximum bond dimension $16$.}
\label{fig:compr-example-ttn}
\end{figure}

\subsection{Updating across separators}
We move on to the final step. We investigate how we can combine individual representations of phase coefficient tensors along a tree decomposition while controlling their representation size (as TTNs).
After processing the subtree rooted at $t$, its contribution is
summarized by a message on $S_t$.  This message includes all factors
consumed below $t$, so the bound for the local factor $W_t$ alone does
not bound its representation size.  Example~\ref{ex:compr-two-bags}
illustrates how a message can have a large explicit expansion, even
though its coefficient tensor has a compact TTN.
To control the bond dimensions of separator messages, we use
permutations of separator variables that preserve the message.
The following lemma turns these symmetries into rank bounds for every
cut of the coefficient layout.

\begin{lemma}[Separator symmetry bound]
\label{lem:separator-exchangeability-rank}
Let $S$ be a finite set, let $A$ be a coefficient tensor on scope $S$, and let
$\mathcal C$ be a partition of $S$.
If $A$ is invariant under permutations of the Boolean coordinates
within each $C\in\mathcal C$, then every coefficient layout
$\mathcal R$ on a variable set containing $S$ satisfies
\[
    D_{\mathcal R}(A)\leq\prod_{C\in\mathcal C}(|C|+1).
\]
\end{lemma}
\begin{proof}
Fix a cut $X\subseteq S$ induced by an edge of $\mathcal R[S]$.
The corresponding coefficient flattening has rows indexed by
$\{0,1\}^X$ and columns indexed by
$\{0,1\}^{S\setminus X}\times\Z_r$.
If two row assignments have the same Hamming
weight on $X\cap C$ for every $C\in\mathcal C$, a permutation within
each $X\cap C$ carries one row assignment to the other while fixing
$S\setminus X$.  The preceding invariance makes their full matrix rows
identical, for every column assignment and residue.  There are at most
$\prod_{C\in\mathcal C}(|X\cap C|+1)$ possible tuples of these
Hamming weights.  The rank is at most the number of distinct rows,
and hence at most $\prod_{C\in\mathcal C}(|C|+1)$.
Taking the maximum over layout edges proves the claim.
The same argument includes $S=\emptyset$, when the empty product is one.
\end{proof}

To describe the factors summarized by a separator message, let $T_t$ be
the subtree of $T$ rooted at $t$ and set
\[
    V_t:=\bigcup_{s\in V(T_t)}B_s,
    \qquad
    I_t:=V_t\setminus S_t,
    \qquad
    E_t^\downarrow:=E(P_C[V_t])\setminus E(P_C[S_t]).
\]
By running intersection, each variable in $I_t$ and each edge in
$E_t^\downarrow$ is assigned to exactly one bag of $T_t$.  Unary terms
on $S_t$ and edges wholly inside $S_t$ are assigned above $t$.
Thus the product of the $W_s$ over $s\in V(T_t)$ contains exactly the
unary terms indexed by $I_t$ and the quadratic interactions indexed by
$E_t^\downarrow$, together with the gate normalizations and constant
phases assigned to bags of $T_t$.
The separator message is
\begin{equation}
    M_t(\bz)
    :=\bigoplus_{\bx\in\{0,1\}^{I_t}}
       \bigotimes_{s\in V(T_t)}
          W_s((\bx,\bz)|_{B_s}),
    \qquad \bz\in\{0,1\}^{S_t}.
    \label{eq:separator-factor}
\end{equation}
Write $\mathsf M_t$ for its coefficient tensor on scope $S_t$.
To compute this message, multiply $\mathsf W_t$ by the child messages,
which have already summed out variables strictly below their separators,
then sum the variables $U_t$.  The \InsideOut update in coefficient form is
\begin{equation}
    \mathsf M_t
    =\Sigma_{U_t}\left(
       \mathsf W_t\ostar
       \bigostar_{s\text{ child of }t}\mathsf M_s
      \right).
    \label{eq:tree-tensor-bucket}
\end{equation}

Put $H_t:=(V_t,E_t^\downarrow)$.  Two separator variables are
interchangeable if swapping them can be absorbed by relabelling the
summed-out variables while preserving the query's phase terms.

For distinct $u,v\in S_t$, let $\tau_{u,v}$ swap $u$ and $v$ and fix
$S_t\setminus\{u,v\}$.  Write $u\sim_t v$ if $u=v$, or if
$\tau_{u,v}$ extends to an automorphism $\pi$ of $H_t$ with
$\pi(S_t)=S_t$ and $b(\pi(i))=b(i)$ for every $i\in I_t$.
This is an equivalence relation: inverses give symmetry, while
conjugating extensions of $\tau_{u,v}$ and $\tau_{v,w}$ gives an
extension of $\tau_{u,w}$.  Let $\mathcal C_t$ be its classes.
Preserving $b$ on $I_t$ is needed for the reindexing below: an
automorphism of $H_t$ alone can exchange internal variables carrying
different unary phases and change the separator message.
Transpositions within each class generate every permutation of $S_t$
preserving those classes, and their extensions compose to automorphisms
with the same properties.

Let $\gamma$ be such a class-preserving permutation and choose an
extension $\pi$.  The unary coefficients and edges in the product
defining $M_t$ are preserved by $\pi$.  Reindexing the sum over $I_t$
therefore gives $M_t(\gamma\cdot\bz)=M_t(\bz)$, where
$(\gamma\cdot\bz)_v=z_{\gamma^{-1}(v)}$.  The contributions with empty
scope are also unchanged.  Equality in $K_r$ is coefficientwise, so $\mathsf M_t$
satisfies the invariance hypothesis of
Lemma~\ref{lem:separator-exchangeability-rank}.  Hence
\[
    D_{\mathcal R}(\mathsf M_t)
    \leq\prod_{C\in\mathcal C_t}(|C|+1)
    \leq(|V|+1)^{|\mathcal C_t|}.
\]
The classes $\mathcal C_t$ depend on $b$ only through its equality
classes.

Within each symmetry class, the message depends only on the number
of variables assigned $1$.  The number of classes therefore controls
the bound on separator-message ranks.  We combine this quantity with
the local cut-rank bound in the following definition.

\subsection{Tensor layout symmetry width and evaluation}

For an $n$-qubit circuit $C$ over a quadratically $r$-compatible gate set,
with chosen gate phase representations, we evaluate its reduced scalar
query $\widehat Q_C$ over $K_r$, using the local bag factors defined
above.  Write
$G=P_C=(V,E)$ for the phase graph and $b:V\to\Z_r$ for the unary
coefficients of~\eqref{eq:phase-unary-coefficients}.
\begin{definition}[Tensor layout symmetry width]
\label{def:compressed-sop-plan}
A plan $\mathcal P=((T,(B_t)_{t\in V(T)}),\mathcal R)$ for $\widehat Q_C$
consists of a rooted binary tree decomposition of $G$ and a rooted binary
coefficient layout on $V$. Let
\[
    \rho(\mathcal P)
       :=\max_{t\in V(T)}\rho_t(\mathcal R),
    \qquad
    q(\mathcal P):=\max_{t\in V(T)}|\mathcal C_t|.
\]
The \emph{TLSW profile} of $\mathcal P$ is the pair
$\tlsw(\mathcal P):=(\rho(\mathcal P),q(\mathcal P))$.
Choose a plan $\mathcal P_*$ that lexicographically minimizes
$\bigl(\max\{\rho(\mathcal P),q(\mathcal P)\},
\rho(\mathcal P),q(\mathcal P)\bigr)$, and define the \emph{TLSW}
of $\widehat Q_C$ to be the pair
\begin{equation}
    \tlsw(\widehat Q_C)
       :=\tlsw(\mathcal P_*).
    \label{eq:sop-compression-width}
\end{equation}
The tie-break makes this pair independent of the choice of $\mathcal P_*$.
For scalar comparisons, set
\[
    \widehat{\tlsw}(\widehat Q_C)
       :=\max\{\rho(\mathcal P_*),q(\mathcal P_*)\}
       =\min_{\mathcal P}\max\{\rho(\mathcal P),q(\mathcal P)\}.
\]
For the empty graph, both trees consist only of their roots, so
$\tlsw(\widehat Q_C)=(0,0)$ and $\widehat{\tlsw}(\widehat Q_C)=0$.
\end{definition}

For Example~\ref{ex:compr-two-bags}, take the two-bag decomposition
there and the paired coefficient layout of
Example~\ref{ex:compr-two-bag-ttn}.  A cut between whole pairs in the
child bag crosses only the clique on $A$, whose cross-adjacency matrix
over $\F_2$ is all ones and has rank at most one.  A cut within a pair
has a singleton side, so it too has rank at most one.  The root bag
contains the clique on $B$ and has the same local cut-rank bound.
Thus $\rho(\mathcal P)=1$.  For the child separator $B$, every
permutation $\sigma$ of $[k]$ extends to the automorphism
$a_i\mapsto a_{\sigma(i)}$, $b_i\mapsto b_{\sigma(i)}$ of $H_t$;
it preserves the unary coefficient $1$ on $A$.  Hence
$\mathcal C_t=\{B\}$, while the root has empty separator, and
$q(\mathcal P)=1$.  This plan therefore has TLSW profile $(1,1)$.

\begin{lemma}[Plan normalization]
\label{lem:tlsw-plan-normalization}
Every plan $\mathcal P$ can be reduced in time polynomial in its encoding
length to a plan $\mathcal P'$ with $O(|V|+1)$ nodes in each tree and
$\rho(\mathcal P')\leq\rho(\mathcal P)$ and
$q(\mathcal P')\leq q(\mathcal P)$.
\end{lemma}
\begin{proof}
Suppress vertices with one child in the coefficient layout.  Every edge
cut of the resulting layout already occurred in the original layout, so
no local cut-rank increases.

In the tree decomposition, delete a non-root node $t$ with
$U_t=\emptyset$ and at most one child, attaching its child, if present,
to the parent of $t$.  Since $B_t\subseteq B_{p(t)}$, the node $t$ has no
assigned unary or quadratic term.  Transfer its gate normalizations and
constant phases to the parent.  Running intersection gives
$B_s\cap B_{p(t)}=B_s\cap B_t$ for its child $s$, so the child's
separator is unchanged.  At every surviving node $u$, its bag,
separator, and labelled subproblem $(V_u,E_u^\downarrow,b|_{I_u})$
are unchanged: $B_t$ adds no vertex to an ancestor subtree.  Thus neither profile
coordinate increases.  Repeat until no such node remains.  An empty
root with one child can likewise be removed, transferring its gate
normalizations and constant phases to the new root.

Each remaining node that introduces a vertex accounts for a distinct
vertex of $V$.  Every other node has two children, so there are at most
as many of these branching nodes as introducing nodes, apart from the
one-node empty case.  The reduced tree decomposition has $O(|V|+1)$
nodes.  The reduced binary coefficient layout has the same bound.
\end{proof}

These widths refer to the reduced query with its chosen phase
representations and the local factor assignment above.

With these two additional bounds we are now ready to provide a full algorithm.  At a high level, local cut-rank controls the construction of each bag factor, and separator symmetry bounds the messages after exact compression.

\begin{theorem}
\label{thm:compressed-sop-evaluation}
Fix a finite quadratically $r$-compatible gate set
with chosen witnessing expressions.  Let $C$ be an $n$-qubit circuit over this gate set, and let $\bx, \by \in \{0,1\}^n$. Let
$P_C=(V,E)$ be its phase graph and $b$ its unary coefficients
from~\eqref{eq:phase-unary-coefficients}.
Given a reduced plan $\mathcal P$ as in
Lemma~\ref{lem:tlsw-plan-normalization}, the amplitude
$\me{\by}{U_C}{\bx}$ can be computed in time
\begin{equation}
    r^{O(1)}2^{O(\rho(\mathcal P))}
       (|V|+1)^{O(q(\mathcal P))}
       \poly(|C|+\log r).
    \label{eq:compressed-sop-runtime}
\end{equation}
\end{theorem}
\begin{proof}
Put $N=|V|$ and $D=\max\{r,(N+1)^{q(\mathcal P)}\}$.
Construct the local bag factors $W_t$ as specified above and process
$T$ in postorder on the common layout $\mathcal R$.
We maintain an exact TTN of bond dimension at most $D$ for each message
passed to a parent.

At bag $t$, Lemma~\ref{lem:bag-tensor-bound} constructs $\mathsf W_t$
with bond dimension at most $r2^{\rho(\mathcal P)}$.
Multiply it by the at most two child messages and sum out $U_t$, using
Lemma~\ref{lem:coefficient-tree-operations}.  This implements
\eqref{eq:tree-tensor-bucket} exactly and gives a TTN for $\mathsf M_t$
with bond dimension at most $r2^{\rho(\mathcal P)}D^2$.
Before passing this message upward, compress it using
Lemma~\ref{lem:exact-tree-reduction}.  Separator symmetry and
Lemma~\ref{lem:separator-exchangeability-rank} give
$D_{\mathcal R}(\mathsf M_t)\leq(N+1)^{q(\mathcal P)}$, so compression
restores the bound $D$ and proves the invariant.

Each update takes time polynomial in the layout size, $r$, and
$r2^{\rho(\mathcal P)}D^2$.  The reduced plan has $O(N+1)$ bags and
layout nodes, and $N=O(|C|)$, giving~\eqref{eq:compressed-sop-runtime}.
At the root no Boolean arguments remain: the resulting residue array
is the coefficient vector of the resulting $K_r$-value.  Applying $\nu_r$
gives the amplitude by~\eqref{eq:circuit-faq-phase-readout}.
\end{proof}

The time bound counts exact arithmetic and zero tests over $\R$ at unit
cost and excludes finding a plan.  An arbitrary supplied plan can
first be reduced in time polynomial in its encoding length, after which
the bound applies with no larger width coordinates.  The algorithm does
not need to construct the classes $\mathcal C_t$.

\begin{proposition}
\label{prop:tlsw-comparisons}
Let $C$ be an $n$-qubit circuit over a
quadratically $r$-compatible gate set with chosen witnessing expressions.
Let $P_C$ be its phase graph
from~\eqref{eq:phase-coefficient-data}, and let $G_C$ be its circuit
graph as in Section~\ref{sec:explicit-tables}, with maximum degree
$\Delta_C$ counting parallel edges with multiplicity. Then
\begin{equation}
\begin{aligned}
    \widehat{\tlsw}(\widehat Q_C)&\leq\rw(P_C),\\
    \widehat{\tlsw}(\widehat Q_C)&\leq\Delta_C\bigl(\tw(G_C)+1\bigr).
\end{aligned}
    \label{eq:tlsw-rank-comparison}
\end{equation}
In particular, for gates acting on at most $a$ qubits, $\Delta_C\leq2a$.
\end{proposition}
\begin{proof}
For the first bound, take the one-bag tree decomposition with bag
$V(P_C)$ and root an optimal rank decomposition of $P_C$ to obtain the
coefficient layout. The separator is empty, so this plan has
$q(\mathcal P)=0$ and $\rho(\mathcal P)\leq\rw(P_C)$.

For the second bound, it is easy to observe that $P_C$ is a  graph minor of $\lline(G_C)$ (see e.g.,~\citep[Lemma~6]{deColnet+2026}).
By minor-monotonicity of treewidth and
$\rw(P_C)\leq\tw(P_C)+1$~\citep{oum2008rank},
\[
    \widehat{\tlsw}(\widehat Q_C)
    \leq\rw(P_C)\leq\tw(P_C)+1
    \leq\tw(\lline(G_C))+1.
\]
To relate the last term to $\tw(G_C)$, take a width-$k$ tree
decomposition of $G_C$ and replace each bag by all edges incident to
its vertices. These bags form a tree decomposition of $\lline(G_C)$:
incident edges occur together in any original bag containing their
common endpoint, and the bags containing an edge $uv$ form the union
of the intersecting subtrees for $u$ and $v$. Each new bag has size at
most $\Delta_C(k+1)$. Taking $k=\tw(G_C)$ gives
\[
    \tw(\lline(G_C))+1\leq\Delta_C\bigl(\tw(G_C)+1\bigr),
\]
which proves the second bound.
\end{proof}

Bounded $\tlsw$ does not force bounded phase-graph rank-width or circuit-graph treewidth.

\begin{theorem}[Clique-chain separation]
\label{thm:clique-chain-separation}
For $n\geq3$, let
$C(n):=\Hblock{n};(\CZblock{n};\Tblock{n};\Hblock{n})^{[n]}$.
With the $K_8$ gate lifts in Figure~\ref{fig:phase-gate-lifts}, for any
fixed computational-basis boundaries,
\begin{equation}
    \widehat{\tlsw}(\widehat Q_{C(n)})\leq2,
    \qquad
    \rw(P_{C(n)})\geq n/3,
    \qquad
    \tw(G_{C(n)})\geq n/4-1.
    \label{eq:clique-chain-separation-summary}
\end{equation}
\end{theorem}
The construction and proof are given in
Appendix~\ref{app:clique-chain-separation}.

\section{Affine--quadratic evaluation and Clifford simulation}
\label{sec:gk-faq}
We show that, for a quantum circuit $C$ over a fixed gate set admitting
the affine--quadratic representations defined below, the standard circuit
FAQ $Q_C$ of Section~\ref{sec:circuit-faq} can be evaluated in polynomial
time by \InsideOut{}.  The query and the elimination scheme are unchanged:
we encode its gate factors and intermediate messages as
\emph{affine--quadratic factors} and implement \BucketEval{} on these
encodings.  Their descriptions remain polynomial in size under every
elimination order, so multiplication and aggregation do not require
enumerating assignments to a message's scope.

Affine supports with quadratic phases are standard descriptions of
stabilizer amplitudes~\citep{DehaeneDeMoor2003}.  We use the quadratic phase
functions (QPFs) of~\citet[Section~3.3]{deColnet+2026}, with an exact scalar
encoding over the ring specified below.  In that work, QPFs represent
Fourier-mode tables on cut-signature spaces and support an evaluation
combining rank-width with stabilizer decompositions.  Here we import their
closure properties to implement the \BucketEval{} interface of
Section~\ref{subsec:insideout} directly on FAQ factor scopes.  We establish
the resulting evaluation guarantee for every elimination order, with
exact scalar arithmetic and symbolic answers on free variables.  As a
special case, we recover the exact-amplitude consequence of the
Gottesman--Knill theorem for Clifford circuits with computational-basis
boundaries~\citep{Gottesman1998,AaronsonGottesman2004}.

\subsection{The affine--quadratic representation}
\label{subsec:aq-representation}

The representation is modelled on Clifford gates.  A unitary $U$ is
\emph{Clifford} if conjugation by $U$ maps Pauli operators, that is,
tensor products of $I$, $X$, $Y$ and $Z$ up to a factor in
$\{\pm1,\pm i\}$, to Pauli operators.  Up to a global phase, every
Clifford unitary is a product of the generators $H$,
$S\coloneqq\operatorname{diag}(1,i)$ and
$\CNOT$~\citep{Gottesman1998,AaronsonGottesman2004}.  Write
$G(v;u)\coloneqq\me{v}{U_G}{u}$ for the gate factor of
Section~\ref{sec:circuit-faq} with input tuple $u$ and output tuple $v$.
The generators have the gate factors
\begin{align*}
    H(v;u)
      &=2^{-1/2}i^{2uv},\\
    S(v;u)
      &=\iv{v=u}\,i^u,\\
    \CNOT(v_1,v_2;u_1,u_2)
      &=\iv{v_1=u_1}\,\iv{v_2=u_1\oplus u_2},
\end{align*}
where, to recall, $\oplus$ denotes addition in $\mathbb F_2$.  Each factor
consists of a scalar, an affine support specified by parity equations, and
a phase whose exponent is quadratic modulo~$4$, with even quadratic
coefficients.  Their scalars and phases lie in the commutative ring
\[
    (\cR,+,\cdot,0,1),
    \qquad
    \cR\coloneqq\mathbb Z[i,2^{-1/2}]\subseteq\C,
\]
which also contains the multipliers $2$ and $1\pm i$ that arise in
aggregation.  Hence, for a circuit over these gates, the standard circuit
FAQ can be evaluated over $\cR\subseteq\C$ without changing its factors or
its value, since identity and boundary indicators also take values
in~$\cR$.

Fix an ordering of each finite scope $W$ of Boolean variables.  An
\emph{affine--quadratic factor}, or \emph{AQ factor}, on $W$ is either the
zero function or a function
\begin{equation}
\label{eq:aq-factor}
    F(\mathbf u)
      =\lambda\,\iv{A\mathbf u=\mathbf b}\,i^{q(\mathbf u)},
      \qquad \mathbf u\in\{0,1\}^{W},
\end{equation}
where $\lambda\in\cR\setminus\{0\}$, the matrix
$A\in\mathbb F_2^{r\times W}$ and vector $\mathbf b\in\mathbb F_2^r$
define a consistent affine system, and
\begin{equation}
\label{eq:aq-phase}
    q(\mathbf u)
      =c+\sum_j\alpha_j u_j
         +2\sum_{j<k}\beta_{jk}u_ju_k\pmod4,
    \qquad
    c,\alpha_j\in\mathbb Z_4,
    \quad \beta_{jk}\in\mathbb F_2.
\end{equation}
The bits are read as integers in the phase polynomial; $r=0$ gives full
support.  We write $\mathsf{AQ}(W)$ for this class.  The even quadratic
coefficients are what make the class closed under elimination: since
$u_1\oplus u_2=u_1+u_2-2u_1u_2$ over the integers, substituting
$x=u_1\oplus u_2$ into a term $2xu_3$ creates a cubic term with coefficient
divisible by~$4$, which vanishes, whereas substituting it into $xu_3$
would leave a cubic phase.  Absorbing $i^c$ into $\lambda$ and separating
the linear phase from the quadratic sign gives exactly the QPF form
of~\citet{deColnet+2026}, restricted to scalars in $\cR$.

An AQ encoding stores a row-reduced affine system with zero rows removed,
the phase coefficients, and the scalar.  On $k$ variables, the nonscalar
data occupy $O((k+1)^2)$ bits, excluding scope identifiers.  We encode the
scalar as
\begin{equation}
\label{eq:aq-scalar-encoding}
    \lambda
      =2^{-h/2}\bigl(a_0+a_1\sqrt2+i(a_2+a_3\sqrt2)\bigr),
    \qquad
    h\in\mathbb N,\quad a_0,a_1,a_2,a_3\in\mathbb Z,
\end{equation}
with all integers in binary.  Zero testing and multiplication take
polynomial time in this encoding.  The coefficient ring $\cR$ and the
represented class $\mathsf{AQ}(W)$ are distinct: the latter is not closed
under arbitrary pointwise addition.  What matters for FAQ evaluation is
closure under the particular product-and-aggregate operation that forms a
message.

\subsection{AQ bag evaluation}
\label{subsec:aq-bucketeval}

At step $i$, let $z_i$ be the eliminated variable and $\cB_i$ its bucket,
as in Section~\ref{subsec:insideout}.  Put
\[
    U_i=\{z_i\}\cup\bigcup_{F\in\cB_i}S_F,
    \qquad R_i=U_i\setminus\{z_i\}.
\]
We take $\cP_i=\varnothing$, so \BucketEval{} computes
\begin{equation}
\label{eq:aq-bucket-message}
    M_i(\mathbf x_{R_i})
      =\sum_{x_{z_i}\in\{0,1\}}
         \prod_{F\in\cB_i}F(\mathbf x_{S_F}).
\end{equation}
This is exactly the message required by \InsideOut{}; the AQ
implementation computes its encoding instead of materializing its table.

\begin{proposition}[AQ bucket evaluation]
\label{prop:aq-bucketeval}
Let $F_1,\ldots,F_b$ be AQ factors with scopes contained in $W$, where
$|W|=w$.  Their product, extended to $W$, and its aggregate over any subset
of $W$ admit AQ encodings computable in $O\bigl((b+1)(w+1)^3\bigr)$ time, apart from exact scalar arithmetic.  In particular, this gives an
AQ implementation of~\eqref{eq:aq-bucket-message} with the same bound.
\end{proposition}

\begin{proof}
Extend the factors to $W$ and apply the QPF closure result
of~\citet[Lemma~3 and Appendix~E]{deColnet+2026}.  Pointwise multiplication
of two QPFs on $W$ takes $O((w+1)^3)$ time, so their iterated product
$F=\prod_{j=1}^b F_j$ is computable in
$O((b+1)(w+1)^3)$ time.  The empty product is the constant-one AQ factor. For aggregation, let $P\subseteq W$ be the retained scope, put $p=|P|$,
and let $\Pi:\mathbb F_2^W\to\mathbb F_2^P$ be the coordinate projection.
The required aggregate is exactly the pushforward $\mathbf y\longmapsto
    \sum_{\mathbf x:\Pi\mathbf x=\mathbf y}F(\mathbf x)$. By~\citet[Lemma~3(iv)]{deColnet+2026}, this produces a QPF in
$O((w+p+1)^3)$ time.  Since $p\le w$, the total nonscalar cost is
$O((b+1)(w+1)^3)$.

Products multiply the incoming scalars, while each single-variable
elimination introduces only a multiplier from $\{1,2,1+i,1-i\}$, up to a
power of $i$, or produces zero.  Hence the result remains in $\cR$ and is
computed exactly using encoding~\eqref{eq:aq-scalar-encoding}.
\end{proof}

Operationally, AQ elimination captures both functional dependencies and
interference without materializing intermediate tables.  If the affine
support determines the eliminated variable, say
$z=\ell(\mathbf u)$, aggregation simply substitutes
$z\mapsto\ell(\mathbf u)$, as in the functional-dependency rewrites of
Section~\ref{sec:circuit-faq}.  If $z$ is unconstrained, the two terms
with $z=0$ and $z=1$ may instead cancel.  For example, two successive
Hadamards give
\[
    \sum_{z\in\{0,1\}}
       \tfrac12\,i^{2xz+2zy}
    =
    \tfrac12\bigl(1+(-1)^{x\oplus y}\bigr)
    =
    \iv{x=y}.
\]
Thus destructive interference is represented by the new affine support
constraint $x=y$, rather than by enumerating the surviving assignments.

\subsection{Evaluation and output guarantees}
\label{subsec:aq-evaluation}

\begin{theorem}[AQ-FAQ evaluation]
\label{thm:aq-faq}
Let $Q$ be an FAQ over $\cR$ on $N$ Boolean variables, with free-variable
set $Y$ and AQ input encodings of total length
$L=\lVert Q\rVert_{\mathrm{in}}$.  For every ordering of the bound
variables, \InsideOut{} with AQ \BucketEval{} and $\cP_i=\varnothing$
computes a single AQ encoding of $Q$ on $Y$ in $O\bigl((L+N)(N+1)^3+(L+N)^3\bigr)$ 
time, including exact scalar arithmetic.  In particular, scalar AQ FAQs
are exactly evaluable in polynomial time.
\end{theorem}

\begin{proof}
Proposition~\ref{prop:aq-bucketeval} preserves the AQ form of every
message.  Every intermediate factor has at most $N$ variables and hence
$O((N+1)^2)$ nonscalar data.  After elimination, the output operation
multiplies the remaining factors, extends their scopes to $Y$, and returns
the resulting AQ encoding.

If bucket $i$ contains $b_i$ factors and $w_i=|U_i|\le N$, then
Proposition~\ref{prop:aq-bucketeval} bounds its nonscalar cost by
$O((b_i+1)(w_i+1)^3)$.  Each input factor and each generated message is
consumed at most once, either in a later bucket or in the final output
product.  There are at most $N$ messages, so, writing $M$ for the number
of input factors, the total number of factors consumed is $O(M+N)$.
Since $M\le L$, the total nonscalar cost is
$O((L+N)(N+1)^3)$.

For scalars, the same consumption property ensures that every nonzero
intermediate scalar is a product of input scalars, each used at most
once, and at most $N$ multipliers from $\{2,1+i,1-i\}$, up to a power
of~$i$.  Its encoding therefore has $O(L+N)$ bits.  There are
$O(L+N)$ scalar multiplications and zero tests, which together cost at
most $O((L+N)^3)$ time using elementary integer arithmetic.  Combining
the two contributions gives the stated bound.
\end{proof}

The output convention is essential: a query on $|Y|$ free variables may
have $2^{|Y|}$ nonzero answers.  The theorem produces an AQ encoding of
size polynomial in $N+L$, rather than an explicit table of these values.

\subsection{Clifford circuits}
\label{subsec:aq-clifford}

We now apply the preceding theorem to the standard circuit FAQ of
Section~\ref{sec:circuit-faq}, with its original segment variables
and gate factors.  Let $Q$ denote the query with the input and output
boundary variables free, so that $Q(x,y)=\me{y}{U_C}{x}$.
Fixing boundary bits amounts to substituting their values, or equivalently
to adding the corresponding indicator factors and aggregating the pinned
variables; fixing all boundary bits gives $Q_C$.  No rewriting of
$Q$ is needed: we use AQ encodings for its factors and for the messages
computed by \InsideOut{}.

\begin{corollary}[Polynomial-time evaluation of Clifford circuit FAQs]
\label{thm:gk-faq}
Let $C$ be a Clifford circuit with $n$ qubits and $m$ gates over
$\{H,S,\CNOT\}$, and let $Q$ be as above.  For every ordering of its
bound variables, \InsideOut{} evaluates $Q$ in time polynomial in
$n+m$, using the AQ implementation of \BucketEval{} and
$\cP_i=\varnothing$.  The same guarantee holds after fixing any subset of
the input or output boundary variables.  The result is an AQ encoding on
the remaining free variables, or the exact scalar $\me{y}{U_C}{x}$ when
all boundary variables are fixed.
\end{corollary}

\begin{proof}
Initialize \InsideOut{} with the original gate factors of $Q$, using
the AQ encodings displayed in Section~\ref{subsec:aq-representation}.
Identity factors and any fixed computational-basis boundary factors are
also AQ.  At step $i$, \InsideOut{} replaces its bucket $\cB_i$ by the
ordinary message~\eqref{eq:aq-bucket-message}.
Proposition~\ref{prop:aq-bucketeval} computes exactly this message in AQ
form, so all active factors remain AQ throughout the run.

The query has $N=O(n+m)$ Boolean wire-segment variables and AQ input
encoding length polynomial in $n+m$.  Theorem~\ref{thm:aq-faq} therefore
gives the claimed running time and symbolic output for every elimination
order.  Fixing boundary variables preserves the AQ form, either by affine
substitution or by adding AQ indicators and treating the pinned variables
as bound, so the same theorem applies to each such query.
\end{proof}

This corollary recovers the exact-amplitude consequence of the
Gottesman--Knill theorem~\citep{Gottesman1998,AaronsonGottesman2004} for circuits
over $\{H,S,\CNOT\}$ with computational-basis boundaries; the same
consequence was obtained through QPF evaluation
by~\citet[Corollary~3]{deColnet+2026}.  The scalar encoding preserves the
global phase of the specified gate circuit.  Here the result follows by
applying \InsideOut{} directly to the standard FAQ $Q_C$.  Queries with
free boundary variables use the same evaluation scheme.

\paragraph{Comparison with specialised Clifford simulators.}
Corollary~\ref{thm:gk-faq} gives an elimination-order-independent
tractability guarantee, rather than a sharp dependence on the number of
qubits and gates.  For comparison, phase-sensitive quadratic-form
simulation gives an $O((m+1)n^2)$ bound for a computational-basis
amplitude~\citep{Beaudrap2022faststabiliser}, while fast-linear-algebra
methods give an $O(mn^{\omega-1})$ bound when $m=\Omega(n)$, where
$\omega$ is a matrix-multiplication exponent~\citep[Corollary~37]{pang2025}.
Our guarantee applies directly to the standard FAQ and holds for every
elimination order, with a symbolic AQ answer on the remaining free
variables.  Sharper quantitative bounds would require a more refined
implementation of AQ evaluation, for example one exploiting
circuit-respecting orders and incremental or blocked linear algebra.

\section{Related work}
\label{sec:related-work}
Classical quantum circuit simulation encompasses several computational tasks, including
amplitude evaluation, output-probability computation, and sampling.
Their exact and approximate variants admit different guarantees and techniques. We focus on exact computational-basis amplitudes.  We organize related work around the representations used to evaluate them and the structure that keeps intermediate computations small.

\paragraph{Sum-product evaluation and annotated relations.}
The functional aggregate query (FAQ) framework
of~\citet{AboKhamisNgoRudra2016}  unifies problems from databases,
constraint satisfaction, graphical models, and matrix computation.
Its \InsideOut{} algorithm builds on variable elimination and bucket
elimination~\citep{Dechter1999}, combining them with fractional-cover
bounds and worst-case optimal join
algorithms~\citep{AtseriasGroheMarx2013,NgoPoratReRudra2018}.
The interpretation of factors as semiring-annotated relations is also
standard in database provenance~\citep{GreenKarvounarakisTannen2007}.
The original FAQ framework already accommodates compactly represented
factors.  We specialize this flexibility to quantum circuits by making
support processing and the algebraic representation of intermediate
factors explicit in the evaluation analysis.  Sparse joins compute
buckets from factor supports, functional dependencies simplify circuit
factors, and algebraic encodings retain messages whose explicit tables
would be large.  These choices preserve the generic \InsideOut{}
semantics while changing how \BucketEval{} computes and represents
its result.

\paragraph{Tensor networks and graphical models.}
\citet{MarkovShi2008} established the central connection between
quantum-circuit simulation and tensor-network contraction, with
contraction complexity characterized by the treewidth of the
network's line graph.  \citet{BoixoEtAl2017} formulated circuit
amplitudes as complex-valued undirected graphical models and evaluated
them by variable elimination, hence close in spirit to the circuit FAQ.  Subsequent work has improved contraction planning
and execution, including hyper-optimized contraction
paths~\citep{GrayKourtis2021} and parallel execution with reuse of shared
contractions~\citep{VincentEtAl2022}.  These methods complement our focus
on the representation and computation of individual intermediate
factors.  Moreover, tensor methods are not restricted to dense
contraction of the original circuit network: low-entanglement
simulation~\citep{Vidal2003} and tree-tensor-network descriptions based
on Schmidt-rank width~\citep{VanDenNestEtAl2007} already exploit algebraic
compression.  Consequently, comparisons between width bounds must
specify both the graph being measured and the representation used for
its intermediate tensors.

\paragraph{Path sums, exponential sums, and model counting.}
The path-sum view expresses amplitudes as sums of phase contributions
over classical assignments.  \citet{DawsonEtAl2005} connected quantum
computation to counting solutions of polynomial equations over
$\mathbb Z_2$.  Symbolic path sums and their rewrite rules were developed
for circuit verification by~\citet{Amy2019};
\citet{AmyStinchcombe2025} subsequently proved polynomial-time
simulation of a family of hidden-shift circuits through confluent
rewriting of symbolic sums.  \citet{MeiEtAl2024} instead reduced strong
simulation to weighted model counting using a Pauli-based encoding of
circuit evolution.  These approaches share the use of classical
assignments and weighted aggregation, but expose different intermediate
objects.  In our quadratic sum-of-powers (SOP) setting, the circuit FAQ
is first reduced using the relevant functional dependencies and
boundary constraints.  Its phase contributions are then retained over
$K_r=\mathbb R[\mathbb Z_r]$, with the complex answer recovered by the
homomorphism $\nu_r$.  This connects phase-residue counting to the same
factor-evaluation interface used for the original circuit FAQ.

\paragraph{Decision diagrams and cut-rank compression.}
\citet{WangEtAl2025} introduced \textsc{FeynmanDD}, which compresses
Feynman-path counting functions using classical decision diagrams.
\citet{Cheng+2025} analyzed this representation in terms of the linear
rank-width of the associated path-variable graph.
\citet{deColnet+2026} obtained a rank-decomposition dynamic program for
quadratic SOPs with $4^k\operatorname{poly}(N)$ arithmetic complexity,
for fixed modulus and a supplied decomposition of width $k$.
\citet{KuyanovKissinger2026} also obtained rank-width-based evaluation
with $4^k$ dependence through contraction of graph-like ZX-diagrams.
Our rank-width result recovers this guarantee within \InsideOut{},
by evaluating the lift of the FD-reduced circuit FAQ using
phase-monomial messages.  The messages group eliminated
assignments according to their effect on the remaining variables;
\BucketEval{} exploits this grouping without materializing the full
remaining-variable function.  The relevant phase graph is extracted
after the reductions required by our compatibility assumptions, rather
than obtained by ignoring arbitrary support constraints.

\paragraph{Clifford simulation and affine--quadratic factors.}
Efficient stabilizer simulation is classically realized through
compact tableau representations~\citep{AaronsonGottesman2004}.
Affine supports and quadratic phases provide another description of
stabilizer amplitudes and Clifford
operations~\citep{DehaeneDeMoor2003}, developed algorithmically through
quadratic-form expansions by~\citet{Beaudrap2022faststabiliser}.
\citet{BuKoh2022} gave a related treatment through periodic quadratic
half Gauss sums and connected their tractability to affine signatures
in the Holant framework.  The phase-sensitive closure of quadratic
phase functions under products, affine substitutions, and summation
is also stated and used by~\citet{deColnet+2026}.  Our affine--quadratic
(AQ) evaluator uses this established algebra within the FAQ framework:
AQ factors can have arbitrary scopes, and multiplication and
aggregation preserve polynomial-size encodings under every ordering
of the bound variables.  Applied to the standard circuit FAQ $Q_C$,
this gives polynomial-time evaluation of Clifford amplitudes without
requiring a width bound or a circuit-specific elimination order.
This derives the Clifford case from a factor-level closure theorem
and realizes it through the same \BucketEval{} interface.  Free
variables can remain symbolic, with the answer returned as an AQ
encoding rather than an exponentially large table.

\paragraph{Stabilizer decompositions and hybrid representations.}
Near-Clifford simulation expands non-Clifford resources into sums of
tractable stabilizer objects, with costs governed by decomposition
size or related magic measures~\citep{BravyiGosset2016,BravyiEtAl2019}.
Graphical and partial stabilizer decompositions combine this approach
with ZX-calculus simplification~\citep{KissingerEtAl2022}, while
\textsc{LIMDD} incorporates stabilizer structure into a decision-diagram
representation~\citep{VinkhuijzenEtAl2023}.
There are also explicit unifications of algebraic and structural
simulation: \citet{CodsiLaakkonen2026} combine graph measures with
stabilizer decompositions, and \citet{deColnet+2026} combine cut-rank
and stabilizer-based evaluation within their SOP dynamic program.
We share the goal of combining structural and algebraic tractability,
but organize the comparison around the semantics of intermediate FAQ
factors: the same elimination step can be implemented with different
representations, whose construction, combination, and aggregation
costs are analyzed separately from their scopes.  Our results give
concrete instances of this representation-independent interface,
rather than a general conversion between all simulation representations.

\paragraph{Database techniques for quantum computing.}
The current paper can be viewed as part of a broader ongoing research programme investigating how database techniques—both theoretical and systems-oriented—can support quantum computing, complementing the growing body of work on quantum algorithms for database problems. \citet{HaiEtAl2025QuantumDataManagement} articulate this
research direction and identify opportunities for data management in
the quantum-computing stack.   In particular, it proposed the use of database techniques for building a classical quantum circuit simulation engine. The current paper can be seen as foundation for such engine. On the system side,
Qymera~\citep{LittauHai2025Qymera}
represents quantum states and gates as relations and translates circuit
execution into SQL, allowing relational database engines to perform the
simulation.  InferQ~\citep{IlinescuEtAl2026InferQ} develops a
database-oriented benchmark that generates compositional circuits and
corresponding SQL workloads, together with circuit and query features
for studying execution performance and backend selection.  At the
workload level, QuWARP~\citep{LittauHai2026QuWARP} plans the
materialization and reuse of intermediate simulation states across
related tasks, while checking that the retained states support the
requested continuations.  
Geerts and Hai~\citep{Geerts+2025} outline a translation of
quantum-circuit simulation into conjunctive queries over
complex-valued semiring annotations, highlighting how hypertree width
and gate sparsity can jointly inform simulation complexity. This work generalizes this idea  by developing a representation-aware account of FAQ evaluation
for quantum-circuit simulation.  We separate the sum--product semantics
from the coefficient semiring, factor representation, and elimination
plan, and study how support and phase structure can be exploited to
compute compact intermediate factors.  The resulting perspective is
not tied to a particular SQL encoding or database engine, but makes the
representation of intermediate functions an explicit part of the
evaluation algorithm and its complexity analysis.

\section{Conclusion}
\label{sec:conclusion}
We considered simulating quantum circuits through classical computation with the core goal of a \textit{unified} approach: the same FAQs (functional aggregate queries) were used to derive existing and new results. Let us review the pay-off of the approach:

\begin{itemize}
    \item \textbf{FAQs with explicit factors} (yielding Markov-Shi):\\[.5em] In Section~\ref{sec:explicit-tables}, we considered FAQs whose factors are represented by indexed listings of their nonzero entries. This representation allowed us to recover the landmark treewidth bound of \cite{MarkovShi2008}. In addition, we were able to show stronger bounds on the exponents of simulation algorithms. Our technical toolkit for that was considering functional dependencies and fractional edge covers.

    \item \textbf{FAQs with phase factors} (yielding Cheng et al./de Colnet et al./Wang et al.):\\[.5em] In Section~\ref{sec:compressed-boundaries} we used a semiring that retains phases formally. To that end, we introduced the phase semiring $K_r\coloneqq\mathbb R[\Z_r]$ for natural numbers $r$. This allowed us to show results for any class of circuits with the same restriction on gates as in \cite{WangEtAl2025,Cheng+2025,deColnet+2026}: an upper bound on the complexity of quantum circuit simulation 
    in terms of the $\fhtw$ (fractional hypertree-width) of FAQs that is never worse than the upper bound of 
    \cite{WangEtAl2025,Cheng+2025,deColnet+2026} in terms of the (linear) rank-width. In addition, the reconstruction is technically much simpler and more compact than the original construction. 

    \item \textbf{FAQs with coefficient tensors} (generalizing both above with a more complex semiring): \\[.5em]
    In Section~\ref{sec:sopw} we obtained a novel approach by \textit{combining and integrating} ideas of the above two families of approaches. To this end, we introduced the parameter tlsw (tensor layout symmetry width) that is strictly more general than both the treewidth approach of \cite{MarkovShi2008}, and the  rank-width methods of  \cite{WangEtAl2025,Cheng+2025,deColnet+2026} for identifying tractable fragments for strong simulation. As with the above approaches, key to the unified approach is that the algorithm is still simply \InsideOut, but over the (more complex) semiring we named phase coefficient tenors. In addition, we identified a class of quantum circuits for which the treewidth of the graph representation and the rank-width are unbounded but the tlsw is at most 2.

    \item \textbf{FAQs with affine-quadratic factors} (yielding the exact-amplitude consequence of the Gottessman-Knill theorem):\\[.5em]
    In Section~\ref{sec:gk-faq} we consider the scenario of the Gottesman-Knill theorem~\citep{Gottesman1998,AaronsonGottesman2004}, namely for circuits over $\{H,S,\CNOT\}$ with computational-basis boundaries. In particular, we recovered the exact-amplitude consequence of the Gottesman--Knill theorem. What is technically interesting and unifying is that it is obtained by
applying \InsideOut{} directly to the standard FAQ $Q_C$, concluding our journey of using a unifying formalism for a multitude of results.

\end{itemize}

\paragraph{Future work.}
We have shown that the established theory of FAQs provides a unified perspective of many algorithmic developments for classical quantum circuit simulation. Building on this perspective and opens several natural directions for further research.

\emph{Beyond quadratically compatible gate sets.}
Our results for rank width and $\tlsw$ assume quadratic $r$-compatible gate sets. This limitation is also present in prior work that considered rank width and still permits universal gate sets.
However, important examples, such as the $\mathsf{fSim}(\pi/2,\pi/6)$ gate used in Google's quantum supremacy experiment~\citep{Arute2019} motivates investigating relaxations of these conditions. In particular, we believe that algebraic generalizations of rank width can allow for the required relaxation of our assumptions while retaining the algorithmic core of Section~\ref{sec:compressed-boundaries}.

\emph{Faster and more general AQ evaluation.}
An interesting gap remains in our results for the affine-quadratic setting. Specialized algorithms for simulation of Clifford circuits achieve better worst-case complexity than our current presentation. It is then of interest if similarly efficient algorithms can be achieved without deviating substantially algorithmically. 
In particular, it should be investigated whether one can close this gap by exploiting recent advances in combining FAQ evaluation and fast matrix  multiplication~\citep{AboKhamisHuSuciu2025}.
Moreover, combining AQ factors with stabilizer-rank
techniques~\citep{BravyiGosset2016,BravyiEtAl2019}
may further extend the approach beyond Clifford circuits.

\emph{Alternative factor representations.}
Our results highlight the importance of compact representations of intermediate FAQ factors. Indeed, compact representation turns out to the only relevant consideration in our results, as the actual overarching algorithm remains unchanged throughout. 
This motivates similar exploration also for other tensor-network representations, including MPS and PEPS, as well as techniques from knowledge-compilation \citep{Quist2026,Vinkhuijzen2026}.
A particularly intriguing direction is to establish a theory of representation-aware FAQ evaluation,
where elimination strategies and factor representations
are optimized jointly, i.e., allowing for potential changes in representation during evaluation.

\emph{Noisy quantum circuits.}
Finally, while FAQs already accommodate quantum channels through density-matrix encodings, efficiently simulating realistic noisy circuits remains an important challenge. Exploring noise-aware factor representations and approximate FAQ evaluation could extend our framework to more realistic quantum-computing scenarios.

\clearpage

\appendix

\section{Proof of the clique-chain separation}
\label{app:clique-chain-separation}
We prove the three bounds in Theorem~\ref{thm:clique-chain-separation}
and the stronger bound $\lrw(P_{C(n)})\geq n/2$ for linear rank-width.

\subsection{Circuit family and phase graph}
\label{app:clique-chain-construction}

For $n,m\geq1$, use the blocks of~\eqref{eq:standard-circuit-blocks}
to define
\begin{equation}
    \mathsf B_n:=\CZblock{n};\Tblock{n};\Hblock{n},
    \qquad C'(n,m):=\Hblock{n};\mathsf B_n^{[m]},
    \qquad C(n):=C'(n,n).
    \label{eq:clique-chain-circuit}
\end{equation}
Figure~\ref{fig:example-circuit} shows $C'(4,3)$.
Write $\bx$ and $\by$ for the fixed input and output strings.
For $i\in[m]$ and $j\in[n]$, let $v_{i,j}$ be the single variable
obtained by identifying the wire-segment variables on qubit $j$
through the diagonal part $\CZblock{n};\Tblock{n}$ of repetition $i$.
The Hadamard factors on qubit $j$ have arguments
$(x_j,v_{1,j})$, $(v_{i,j},v_{i+1,j})$ for $i\in[m-1]$, and
$(v_{m,j},y_j)$.  For $C(n)$, write $G_n:=P_{C(n)}$.  This graph has vertex set
$V=\bigcup_{i=1}^n K_i$, where
\[
    K_i=\{v_{i,1},\ldots,v_{i,n}\}\qquad(i\in[n]),
\]
and edge set
\[
    E(P_{C(n)})
    =\bigcup_{i=1}^n\binom{K_i}{2}
      \;\cup\;
      \bigl\{\{v_{i,j},v_{i+1,j}\}:
               i\in[n-1],\ j\in[n]\bigr\}.
\]

Thus $G_n$ is a chain of $n$ cliques of size $n$, with a perfect
matching between consecutive cliques.

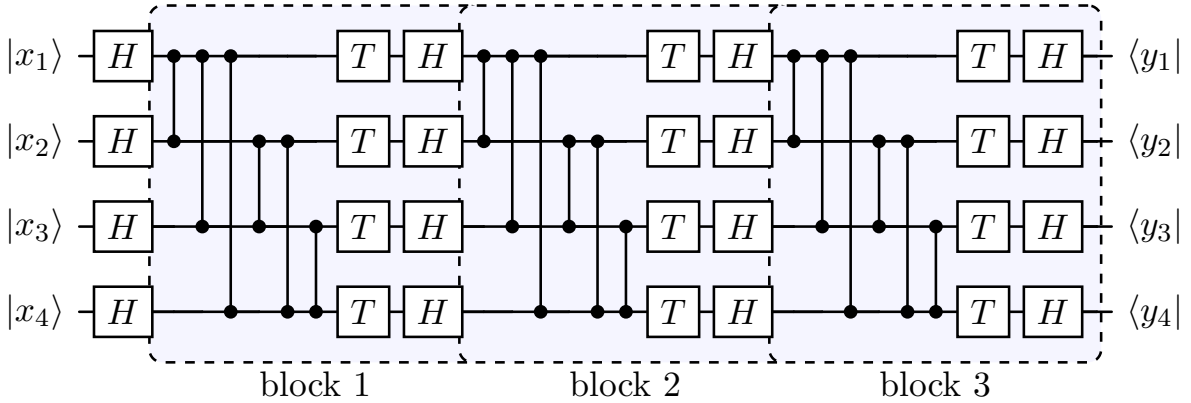
\begin{figure}[H]
\centering
\resizebox{\textwidth}{!}{
\begin{quantikz}[row sep=0.38cm, column sep=0.16cm]
\lstick{$\ket{x_1}$}
 & \gate{H}
 & \ctrl{1}\blockgroup{block 1}
 & \ctrl{2}
 & \ctrl{3}
 & \qw
 & \qw
 & \qw
 & \gate{T}
 & \gate{H}
 & \ctrl{1}\blockgroup{block 2}
 & \ctrl{2}
 & \ctrl{3}
 & \qw
 & \qw
 & \qw
 & \gate{T}
 & \gate{H}
 & \ctrl{1}\blockgroup{block 3}
 & \ctrl{2}
 & \ctrl{3}
 & \qw
 & \qw
 & \qw
 & \gate{T}
 & \gate{H}
 & \qw
 & \rstick{$\bra{y_1}$}
\\
\lstick{$\ket{x_2}$}
 & \gate{H}
 & \control{}
 & \qw
 & \qw
 & \ctrl{1}
 & \ctrl{2}
 & \qw
 & \gate{T}
 & \gate{H}
 & \control{}
 & \qw
 & \qw
 & \ctrl{1}
 & \ctrl{2}
 & \qw
 & \gate{T}
 & \gate{H}
 & \control{}
 & \qw
 & \qw
 & \ctrl{1}
 & \ctrl{2}
 & \qw
 & \gate{T}
 & \gate{H}
 & \qw
 & \rstick{$\bra{y_2}$}
\\
\lstick{$\ket{x_3}$}
 & \gate{H}
 & \qw
 & \control{}
 & \qw
 & \control{}
 & \qw
 & \ctrl{1}
 & \gate{T}
 & \gate{H}
 & \qw
 & \control{}
 & \qw
 & \control{}
 & \qw
 & \ctrl{1}
 & \gate{T}
 & \gate{H}
 & \qw
 & \control{}
 & \qw
 & \control{}
 & \qw
 & \ctrl{1}
 & \gate{T}
 & \gate{H}
 & \qw
 & \rstick{$\bra{y_3}$}
\\
\lstick{$\ket{x_4}$}
 & \gate{H}
 & \qw
 & \qw
 & \control{}
 & \qw
 & \control{}
 & \control{}
 & \gate{T}
 & \gate{H}
 & \qw
 & \qw
 & \control{}
 & \qw
 & \control{}
 & \control{}
 & \gate{T}
 & \gate{H}
 & \qw
 & \qw
 & \control{}
 & \qw
 & \control{}
 & \control{}
 & \gate{T}
 & \gate{H}
 & \qw
 & \rstick{$\bra{y_4}$}
\end{quantikz}
}
\caption{The circuit $C'(4,3)$ of~\eqref{eq:clique-chain-circuit}.
Each dashed box is one occurrence of
$\mathsf B_4=\CZblock{4};\Tblock{4};\Hblock{4}$.}
\label{fig:example-circuit}
\end{figure}

\subsection{Bounded tensor layout symmetry width}
\label{app:clique-chain-tlsw}

Fix $n\geq3$.
We construct a plan $\mathcal P$ with
$\rho(\mathcal P)\leq1$ and $q(\mathcal P)\leq2$.

\begin{proof}
Take any $r$ with $8\mid r$ and the gate expressions
\[
    E_H(u,v)=2^{-1/2}\phases{(r/2)uv},\qquad
    E_T(u)=\phases{(r/8)u},\qquad
    E_{\CZ}(u,v)=\phases{(r/2)uv}.
\]
Each $T$ gate contributes $r/8$ to its variable. The boundary Hadamards
add $r/2$ times the corresponding fixed bit. Thus the unary coefficient
function $b_{\bx,\by}:V\to\Z_r$ is given by
\begin{equation}
    b_{\bx,\by}(v_{i,j})
      =\frac r8+\frac r2\bigl(\mathbf 1[i=1]x_j
                    +\mathbf 1[i=n]y_j\bigr)
      \quad\text{in }\Z_r.
    \label{eq:clique-chain-unary-coefficients}
\end{equation}
Since $n\geq3$, the two boundary contributions occur in different
blocks. Each coefficient is therefore either $r/8$ or $5r/8$. These
values are distinct, and which vertices have equal coefficients is independent
of $r$, so the TLSW profile of any fixed plan is also independent of $r$.

We process the cliques in order. Let $T$ be the path $t_1,\ldots,t_n$,
rooted at $t_n$, with bags
\[
    B_{t_i}=K_i\cup K_{i+1}\quad(i\in[n-1]),
    \qquad B_{t_n}=K_n.
\]
For $i\in[n-1]$, we have $S_{t_i}=K_{i+1}$,
$U_{t_i}=K_i$, and
\[
    E_{t_i}
      =\binom{K_i}{2}
       \cup\bigl\{\{v_{i,j},v_{i+1,j}\}:j\in[n]\bigr\},
    \qquad E_{t_n}=\binom{K_n}{2}.
\]

The coefficient layout groups variables by qubit. For $j\in[n]$, set
\[
    L_j:=\{v_{1,j},\ldots,v_{n,j}\}.
\]
Start with a rooted binary tree with leaf set $[n]$.  For each $j\in[n]$,
replace leaf $j$ by a rooted binary tree with leaf set $L_j$, and denote
the resulting coefficient layout by $\mathcal R$.  Set
$\mathcal P:=((T,(B_t)_{t\in V(T)}),\mathcal R)$.

We first bound the local cut-rank. For $i\in[n-1]$ and
$e\in E(\mathcal R[B_{t_i}])$, the cut
$(J_e,B_{t_i}\setminus J_e)$ either has a
singleton side or partitions $B_{t_i}$ into unions of the pairs
$\{v_{i,j},v_{i+1,j}\}$, $j\in[n]$.  In the first case
$\gamma_{P_{C_{t_i}}}(J_e)\leq1$.  In the second, no matching edge crosses
the cut, and the cross-adjacency matrix of the clique on $K_i$ has rank at
most one over $\F_2$.  Thus $\gamma_{P_{C_{t_i}}}(J_e)\leq1$ in both
cases.  At $t_n$, each cut matrix of the clique on $K_n$ also has rank
at most one.
Hence $\rho(\mathcal P)\leq1$.

For separator symmetry, qubits with the same input bit can be
exchanged throughout the processed part of the circuit. More precisely,
for $i\in[n-1]$, we have
$I_{t_i}=K_1\cup\cdots\cup K_i$ and
$E_{t_i}^\downarrow=E(P_{C(n)}[K_1\cup\cdots\cup K_{i+1}])
\setminus\binom{K_{i+1}}2$.
If $x_j=x_k$, let $\pi_{i;j,k}$ exchange $v_{h,j}$ and $v_{h,k}$
for every $h\in[i+1]$ and fix the remaining vertices of $V_{t_i}$.
Then $\pi_{i;j,k}|_{S_{t_i}}$ exchanges $v_{i+1,j}$ and
$v_{i+1,k}$, and
\[
    \pi_{i;j,k}(E_{t_i}^\downarrow)=E_{t_i}^\downarrow,
    \qquad
    b_{\bx,\by}(\pi_{i;j,k}(v))=b_{\bx,\by}(v)
       \quad(v\in I_{t_i}),
\]
where the second equality follows from
\eqref{eq:clique-chain-unary-coefficients}.  By the definition
of $\sim_{t_i}$,
\[
    x_j=x_k
    \quad\Longrightarrow\quad
    v_{i+1,j}\sim_{t_i}v_{i+1,k}.
\]
Thus $|\mathcal C_{t_i}|\leq2$ for $i\in[n-1]$ and
$\mathcal C_{t_n}=\emptyset$, so $q(\mathcal P)\leq2$.  Therefore
\[
    \widehat{\tlsw}(\widehat Q_{C(n)})
       \leq\max\{\rho(\mathcal P),q(\mathcal P)\}\leq2.
\]
\end{proof}

\subsection{Treewidth of the circuit graph}
\begin{proposition}
\label{prop:clique-chain-treewidth}
For every $n\geq2$, the circuit $C(n)$ satisfies
$\tw(G_{C(n)})\geq n/4-1$.
\end{proposition}
\begin{proof}
Fix one occurrence of $\CZblock{n}$. Each wire segment is a vertex of
$\lline(G_{C(n)})$. For each $\CZ$ gate in the block, contract the edge
between its input and output segments on each of its two wires.
These edges exist because the segments meet at the gate.

The segments of the block on qubit $j$ now form one vertex $v_j$.
For every $j\ne\ell$, the gate $\CZ_{j,\ell}$ gives an edge
$v_jv_\ell$. Thus $K_n$ is a minor of $\lline(G_{C(n)})$, and
\[
    \tw\bigl(\lline(G_{C(n)})\bigr)\geq n-1.
\]
The circuit graph has maximum degree four. Applying the
line-graph comparison from the proof of
Proposition~\ref{prop:tlsw-comparisons} gives
\[
    n\leq\tw\bigl(\lline(G_{C(n)})\bigr)+1
      \leq4\bigl(\tw(G_{C(n)})+1\bigr),
\]
so $\tw(G_{C(n)})\geq n/4-1$.
\end{proof}

\subsection{Rank-width of the phase graph}

Both lower bounds follow from the same cut argument: a balanced cut
either splits many cliques or crosses many matching edges. We first
establish these two estimates, then apply them to linear layouts
and rank decompositions.

Fix $n\geq2$, and write $G:=G_n$, $V:=V(G)$, and $A:=A_G$.
We call the cliques $K_i$ \emph{variable blocks}. For disjoint
$X,Y\subseteq V$, write
$A[X,Y]$ for the submatrix with rows indexed by $X$ and columns indexed
by $Y$, and extend the cut-rank notation of
Section~\ref{sec:cut-rank} by setting
\[
    \gamma_G(X,Y):=\operatorname{rank}_{\F_2}A[X,Y].
\]
Thus $\gamma_G(X,V\setminus X)=\gamma_G(X)$.

The rank $\gamma_G(X,Y)$ is monotone in both arguments: deleting rows
or columns cannot increase rank.  By the description of $G_n$,
\[
    A_{v_{i,j},v_{k,\ell}}=
    \begin{cases}
        1 & \text{if }(i=k\text{ and }j\ne\ell)
                    \text{ or }(|i-k|=1\text{ and }j=\ell),\\
        0 & \text{otherwise}.
    \end{cases}
\]
In particular, $A_{v_{i,j},v_{k,\ell}}=0$ whenever
$(k+1,\ell)<_{\mathrm{lex}}(i,j)$.
Let $(P_1,P_2)$ be a partition of $V$.

\begin{lemma}
\label{lem:some-but-not-all-variables}
Suppose that $m$ distinct variable blocks $K_{i_1},\ldots,K_{i_m}$
are split by $(P_1,P_2)$: each meets both $P_1$ and $P_2$.
Then $\gamma_G(P_1,P_2)\geq m/2$.
\end{lemma}

\begin{proof}
Order the split blocks as $i_1<\cdots<i_m$ and select every other
block, starting with $K_{i_1}$.  In each selected block $K_{i_j}$, choose
$u_j\in P_1\cap K_{i_j}$ and $w_j\in P_2\cap K_{i_j}$.
The block is a clique, so $A_{u_j,w_j}=1$.  Distinct selected block
indices differ by at least two.  Edges between different blocks occur
only between consecutive blocks, so $A_{u_j,w_k}=0$ for $j\ne k$.
These rows and columns therefore form an identity submatrix of
$A[P_1,P_2]$ of size $\lceil m/2\rceil$, proving the claim.
\end{proof}

For $P\subseteq V$ and $i\in[n-1]$, define the sets of vertices whose
matching neighbors across consecutive variable blocks lie outside $P$:
\[
\begin{aligned}
    D_i^+(P)&:=\{v_{i+1,j}\in P:v_{i,j}\notin P\},\\
    D_i^-(P)&:=\{v_{i,j}\in P:v_{i+1,j}\notin P\}.
\end{aligned}
\]
\begin{lemma}
\label{lem:decreasing-variables-increase-rank}
For every partition $(P_1,P_2)$ of $V$,
\[
    \gamma_G(P_1,P_2)\geq\sum_{i=1}^{n-1}|D_i^+(P_1)|,
    \qquad
    \gamma_G(P_1,P_2)\geq\sum_{i=1}^{n-1}|D_i^-(P_1)|.
\]
\end{lemma}
\begin{proof}
We prove the first inequality.  Set
\[
    X:=\bigcup_{i=1}^{n-1}D_i^+(P_1),
    \qquad
    Y:=\{v_{i,j}:v_{i+1,j}\in D_i^+(P_1),\ i\in[n-1]\}.
\]
Then $X\subseteq P_1$ and $Y\subseteq P_2$, so
$\gamma_G(P_1,P_2)\geq\gamma_G(X,Y)$.
List $X=\{v_{i_1,j_1},\ldots,v_{i_s,j_s}\}$ in lexicographic order,
and order $Y$ as $v_{i_1-1,j_1},\ldots,v_{i_s-1,j_s}$.
The diagonal entries of $A[X,Y]$ are all one.  If $\ell<k$, then
$(i_\ell,j_\ell)<_{\mathrm{lex}}(i_k,j_k)$, so
$A_{v_{i_k,j_k},v_{i_\ell-1,j_\ell}}=0$.
Thus $A[X,Y]$ is upper triangular with rank
$|X|=\sum_i|D_i^+(P_1)|$.
Reversing the order of the variable blocks exchanges the two
inequalities and proves the second.
\end{proof}

In particular, if $P_1\cap K_i=\emptyset$, then for every $j$,
\begin{equation}
    \gamma_G(P_1,P_2)\geq |P_1\cap K_j|.
    \label{eq:clique-chain-empty-block-cut}
\end{equation}
Indeed, each coordinate represented in $P_1\cap K_j$ must cross from
$P_2$ into $P_1$ somewhere between blocks $i$ and $j$.
These give distinct vertices in the sets $D_h^+(P_1)$ when $i<j$,
or in $D_h^-(P_1)$ when $i>j$; apply
Lemma~\ref{lem:decreasing-variables-increase-rank}.  For $i=j$, the
right-hand side is zero.

\begin{proposition}
\label{prop:unbounded_lrw}
For every $n\geq2$, the phase graph $G_n$ satisfies
$\lrw(G_n)\geq n/2$.
\end{proposition}
\begin{proof}
    Fix a linear layout and let $(P_1,P_2)$ be its middle cut, so
    $|P_1|=\lfloor n^2/2\rfloor$. If every block is split by the cut,
    Lemma~\ref{lem:some-but-not-all-variables} gives
    $\gamma_G(P_1,P_2)\geq n/2$.

    Otherwise, after exchanging $P_1$ and $P_2$ if necessary, there is a
    block $K_i$ with $P_1\cap K_i=\emptyset$.  Since
    $|P_1|\geq\lfloor n^2/2\rfloor$, some block $K_j$ satisfies
    $|P_1\cap K_j|\geq\lceil\lfloor n^2/2\rfloor/n\rceil\geq n/2$.
    Equation~\eqref{eq:clique-chain-empty-block-cut} gives
    \begin{equation*}
        \gamma_G(P_1,P_2)\geq n/2.
    \end{equation*}
    Thus every linear layout has a cut of rank at least $n/2$, proving the
    claim.
\end{proof}

For rank decompositions, we need a balanced cut. Every subcubic tree
with $N\geq3$ leaves has an edge whose two
sides contain between $N/3$ and $2N/3$ leaves.  To see this, root the tree at
a leaf and assign each node the number of descendant leaves.  Starting at
the root, move to a child of weight greater than $2N/3$ whenever one exists.
At the first node where neither child is that heavy, one child has weight at
least $N/3$; its incident edge gives the desired partition.

\begin{proposition}
\label{prop:unbounded_rw}
For every $n\geq3$, the phase graph $G_n$ satisfies
$\rw(G_n)\geq n/3$.
\end{proposition}

\begin{proof}
Fix an arbitrary rank decomposition.  By the preceding balanced-edge
argument, some edge induces a partition $(P_1,P_2)$ with
$|P_1|\geq n^2/3$ and $|P_2|\geq n^2/3$.

If every block is split by the cut,
Lemma~\ref{lem:some-but-not-all-variables} gives
$\gamma_G(P_1,P_2)\geq n/2$.
Otherwise, after exchanging $P_1$ and $P_2$ if necessary, there is a block
$K_i$ with $P_1\cap K_i=\emptyset$.  Since $|P_1|\geq n^2/3$, some block
$K_j$ satisfies $|P_1\cap K_j|\geq n/3$.
Equation~\eqref{eq:clique-chain-empty-block-cut} gives
$\gamma_G(P_1,P_2)\geq n/3$.  Thus every
rank decomposition has width at least $n/3$, proving the claim.
\end{proof}


\begin{thebibliography}{67}
\providecommand{\natexlab}[1]{#1}
\providecommand{\url}[1]{\texttt{#1}}
\expandafter\ifx\csname urlstyle\endcsname\relax
  \providecommand{\doi}[1]{doi: #1}\else
  \providecommand{\doi}{doi: \begingroup \urlstyle{rm}\Url}\fi

\bibitem[Aaronson and Gottesman(2004)]{AaronsonGottesman2004}
Scott Aaronson and Daniel Gottesman.
\newblock Improved simulation of stabilizer circuits.
\newblock \emph{Physical Review A}, 70:\penalty0 052328, 2004.
\newblock URL \url{https://doi.org/10.1103/PhysRevA.70.052328}.

\bibitem[Abdulla et~al.(2025)Abdulla, Chen, Chen, Hol{\'{\i}}k, Leng{\'{a}}l, Lin, Lo, and Tsai]{DBLP:journals/pacmpl/AbdullaCCHLLLT25}
Parosh~Aziz Abdulla, Yo{-}Ga Chen, Yu{-}Fang Chen, Luk{\'{a}}s Hol{\'{\i}}k, Ondrej Leng{\'{a}}l, Jyun{-}Ao Lin, Fang{-}Yi Lo, and Wei{-}Lun Tsai.
\newblock Verifying quantum circuits with level-synchronized tree automata.
\newblock \emph{Proc. {ACM} Program. Lang.}, 9\penalty0 ({POPL}):\penalty0 923--953, 2025.
\newblock URL \url{https://doi.org/10.1145/3704868}.

\bibitem[Abdulla et~al.(2026)Abdulla, Chen, Hecko, Hol{\'{\i}}k, Leng{\'{a}}l, Lin, and Thinniyam]{DBLP:journals/pacmpl/AbdullaCHHLLT26}
Parosh~Aziz Abdulla, Yu{-}Fang Chen, Michal Hecko, Luk{\'{a}}s Hol{\'{\i}}k, Ondrej Leng{\'{a}}l, Jyun{-}Ao Lin, and Ramanathan~S. Thinniyam.
\newblock Parameterized verification of quantum circuits.
\newblock \emph{Proc. {ACM} Program. Lang.}, 10\penalty0 ({POPL}):\penalty0 2021--2050, 2026.
\newblock URL \url{https://doi.org/10.1145/3776712}.

\bibitem[Abo~Khamis et~al.(2016)Abo~Khamis, Ngo, and Rudra]{AboKhamisNgoRudra2016}
Mahmoud Abo~Khamis, Hung~Q. Ngo, and Atri Rudra.
\newblock {FAQ}: {Q}uestions {A}sked {F}requently.
\newblock In \emph{Proceedings of the 35th {ACM SIGMOD-SIGACT-SIGAI} Symposium on Principles of Database Systems ({PODS})}, pages 13--28. Association for Computing Machinery, 2016.
\newblock URL \url{https://doi.org/10.1145/2902251.2902280}.

\bibitem[Abo~Khamis et~al.(2025)Abo~Khamis, Hu, and Suciu]{AboKhamisHuSuciu2025}
Mahmoud Abo~Khamis, Xiao Hu, and Dan Suciu.
\newblock Fast matrix multiplication meets the submodular width.
\newblock \emph{Proc. ACM Manag. Data}, 3\penalty0 (2), June 2025.
\newblock URL \url{https://doi.org/10.1145/3725235}.

\bibitem[Amy(2019)]{Amy2019}
Matthew Amy.
\newblock Towards large-scale functional verification of universal quantum circuits.
\newblock In \emph{Proceedings of the 15th International Conference on Quantum Physics and Logic ({QPL 2018})}, volume 287 of \emph{Electronic Proceedings in Theoretical Computer Science}, pages 1--21, 2019.
\newblock URL \url{https://doi.org/10.4204/EPTCS.287.1}.

\bibitem[Amy and Stinchcombe(2025)]{AmyStinchcombe2025}
Matthew Amy and Lucas~Shigeru Stinchcombe.
\newblock Polynomial-time classical simulation of hidden shift circuits via confluent rewriting of symbolic sums.
\newblock \emph{Quantum}, 9:\penalty0 1926, 2025.
\newblock URL \url{https://doi.org/10.22331/q-2025-12-02-1926}.

\bibitem[Arute et~al.(2019)Arute, Arya, Babbush, Bacon, Bardin, Barends, Biswas, Boixo, Brandao, Buell, Burkett, Chen, Chen, Chiaro, Collins, Courtney, Dunsworth, Farhi, Foxen, Fowler, Gidney, Giustina, Graff, Guerin, Habegger, Harrigan, Hartmann, Ho, Hoffmann, Huang, Humble, Isakov, Jeffrey, Jiang, Kafri, Kechedzhi, Kelly, Klimov, Knysh, Korotkov, Kostritsa, Landhuis, Lindmark, Lucero, Lyakh, Mandr{\`a}, McClean, McEwen, Megrant, Mi, Michielsen, Mohseni, Mutus, Naaman, Neeley, Neill, Niu, Ostby, Petukhov, Platt, Quintana, Rieffel, Roushan, Rubin, Sank, Satzinger, Smelyanskiy, Sung, Trevithick, Vainsencher, Villalonga, White, Yao, Yeh, Zalcman, Neven, and Martinis]{Arute2019}
Frank Arute, Kunal Arya, Ryan Babbush, Dave Bacon, Joseph~C. Bardin, Rami Barends, Rupak Biswas, Sergio Boixo, Fernando G. S.~L. Brandao, David~A. Buell, Brian Burkett, Yu~Chen, Zijun Chen, Ben Chiaro, Roberto Collins, William Courtney, Andrew Dunsworth, Edward Farhi, Brooks Foxen, Austin Fowler, Craig Gidney, Marissa Giustina, Rob Graff, Keith Guerin, Steve Habegger, Matthew~P. Harrigan, Michael~J. Hartmann, Alan Ho, Markus Hoffmann, Trent Huang, Travis~S. Humble, Sergei~V. Isakov, Evan Jeffrey, Zhang Jiang, Dvir Kafri, Kostyantyn Kechedzhi, Julian Kelly, Paul~V. Klimov, Sergey Knysh, Alexander Korotkov, Fedor Kostritsa, David Landhuis, Mike Lindmark, Erik Lucero, Dmitry Lyakh, Salvatore Mandr{\`a}, Jarrod~R. McClean, Matthew McEwen, Anthony Megrant, Xiao Mi, Kristel Michielsen, Masoud Mohseni, Josh Mutus, Ofer Naaman, Matthew Neeley, Charles Neill, Murphy~Yuezhen Niu, Eric Ostby, Andre Petukhov, John~C. Platt, Chris Quintana, Eleanor~G. Rieffel, Pedram Roushan, Nicholas~C. Rubin, Daniel Sank, Kevin~J.
  Satzinger, Vadim Smelyanskiy, Kevin~J. Sung, Matthew~D. Trevithick, Amit Vainsencher, Benjamin Villalonga, Theodore White, Z.~Jamie Yao, Ping Yeh, Adam Zalcman, Hartmut Neven, and John~M. Martinis.
\newblock Quantum supremacy using a programmable superconducting processor.
\newblock \emph{Nature}, 574\penalty0 (7779):\penalty0 505--510, 2019.
\newblock URL \url{https://doi.org/10.1038/s41586-019-1666-5}.

\bibitem[Atserias et~al.(2013)Atserias, Grohe, and Marx]{AtseriasGroheMarx2013}
Albert Atserias, Martin Grohe, and D{\'a}niel Marx.
\newblock Size bounds and query plans for relational joins.
\newblock \emph{SIAM Journal on Computing}, 42\penalty0 (4):\penalty0 1737--1767, 2013.
\newblock URL \url{https://doi.org/10.1137/110859440}.

\bibitem[Bagan et~al.(2007)Bagan, Durand, and Grandjean]{Bagan+2007}
Guillaume Bagan, Arnaud Durand, and Etienne Grandjean.
\newblock On acyclic conjunctive queries and constant delay enumeration.
\newblock In Jacques Duparc and Thomas~A. Henzinger, editors, \emph{Computer Science Logic}, pages 208--222, Berlin, Heidelberg, 2007. Springer Berlin Heidelberg.
\newblock URL \url{https://doi.org/10.1007/978-3-540-74915-8_18}.

\bibitem[Beaudrap and Herbert(2022)]{Beaudrap2022faststabiliser}
Niel~de Beaudrap and Steven Herbert.
\newblock Fast {S}tabiliser {S}imulation with {Q}uadratic {F}orm {E}xpansions.
\newblock \emph{Quantum}, 6:\penalty0 803, September 2022.
\newblock URL \url{https://doi.org/10.22331/q-2022-09-15-803}.

\bibitem[Boixo et~al.(2017)Boixo, Isakov, Smelyanskiy, and Neven]{BoixoEtAl2017}
Sergio Boixo, Sergei~V. Isakov, Vadim~N. Smelyanskiy, and Hartmut Neven.
\newblock Simulation of low-depth quantum circuits as complex undirected graphical models.
\newblock arXiv preprint arXiv:1712.05384, 2017.
\newblock URL \url{https://arxiv.org/abs/1712.05384}.

\bibitem[{Boixo} et~al.(2018){Boixo}, {Isakov}, {Smelyanskiy}, {Babbush}, {Ding}, {Jiang}, {Bremner}, {Martinis}, and {Neven}]{Boixto2018}
Sergio {Boixo}, Sergei~V. {Isakov}, Vadim~N. {Smelyanskiy}, Ryan {Babbush}, Nan {Ding}, Zhang {Jiang}, Michael~J. {Bremner}, John~M. {Martinis}, and Hartmut {Neven}.
\newblock {Characterizing quantum supremacy in near-term devices}.
\newblock \emph{Nature Physics}, 14\penalty0 (6):\penalty0 595--600, June 2018.
\newblock URL \url{https://doi.org/10.1038/s41567-018-0124-x}.

\bibitem[Bravyi and Gosset(2016)]{BravyiGosset2016}
Sergey Bravyi and David Gosset.
\newblock Improved classical simulation of quantum circuits dominated by {Clifford} gates.
\newblock \emph{Physical Review Letters}, 116\penalty0 (25):\penalty0 250501, 2016.
\newblock URL \url{https://doi.org/10.1103/PhysRevLett.116.250501}.

\bibitem[Bravyi et~al.(2019)Bravyi, Browne, Calpin, Campbell, Gosset, and Howard]{BravyiEtAl2019}
Sergey Bravyi, Dan Browne, Padraic Calpin, Earl Campbell, David Gosset, and Mark Howard.
\newblock Simulation of quantum circuits by low-rank stabilizer decompositions.
\newblock \emph{Quantum}, 3:\penalty0 181, 2019.
\newblock URL \url{https://doi.org/10.22331/q-2019-09-02-181}.

\bibitem[Bu and Koh(2022)]{BuKoh2022}
Kaifeng Bu and Dax~Enshan Koh.
\newblock Classical simulation of quantum circuits by half {Gauss} sums.
\newblock \emph{Communications in Mathematical Physics}, 390:\penalty0 471--500, 2022.
\newblock URL \url{https://doi.org/10.1007/s00220-022-04320-1}.

\bibitem[Ceruti et~al.(2025)Ceruti, Kressner, and Sulz]{CerutiKressnerSulz25}
Gianluca Ceruti, Daniel Kressner, and Dominik Sulz.
\newblock Low-rank tree tensor network operators for long-range pairwise interactions.
\newblock \emph{SIAM J. Sci. Comput.}, 47\penalty0 (4):\penalty0 A2248--A2271, 2025.
\newblock URL \url{https://doi.org/10.1137/24M1661996}.

\bibitem[Chen et~al.(2025)Chen, Chung, Leng{\'{a}}l, Lin, Tsai, and Yen]{DBLP:journals/cacm/ChenCLLTY25}
Yu{-}Fang Chen, Kai{-}Min Chung, Ondrej Leng{\'{a}}l, Jyun{-}Ao Lin, Wei{-}Lun Tsai, and Di{-}De Yen.
\newblock An automata-based framework for verification and bug hunting in quantum circuits.
\newblock \emph{Commun. {ACM}}, 68\penalty0 (6):\penalty0 85--93, 2025.
\newblock URL \url{https://doi.org/10.1145/3725728}.

\bibitem[Cheng et~al.(2025)Cheng, Wang, Deng, Chen, and Ji]{Cheng+2025}
Bin Cheng, Ziyuan Wang, Ruixuan Deng, Jianxin Chen, and Zhengfeng Ji.
\newblock Breaking the treewidth barrier in quantum circuit simulation with decision diagrams.
\newblock arXiv preprint arXiv:2510.06775, 2025.
\newblock URL \url{https://arxiv.org/abs/2510.06775}.

\bibitem[Codsi and Laakkonen(2026)]{CodsiLaakkonen2026}
Julien Codsi and Tuomas Laakkonen.
\newblock Unifying graph measures and stabilizer decompositions for the classical simulation of quantum circuits.
\newblock arXiv preprint arXiv:2603.06377, 2026.
\newblock URL \url{https://arxiv.org/abs/2603.06377}.

\bibitem[Cygan et~al.(2015)Cygan, Fomin, Kowalik, Lokshtanov, Marx, Pilipczuk, Pilipczuk, and Saurabh]{DBLP:books/sp/CyganFKLMPPS15}
Marek Cygan, Fedor~V. Fomin, Lukasz Kowalik, Daniel Lokshtanov, D{\'{a}}niel Marx, Marcin Pilipczuk, Michal Pilipczuk, and Saket Saurabh.
\newblock \emph{Parameterized Algorithms}.
\newblock Springer, 2015.
\newblock URL \url{https://doi.org/10.1007/978-3-319-21275-3}.

\bibitem[Dawson et~al.(2005)Dawson, Hines, Mortimer, Haselgrove, Nielsen, and Osborne]{DawsonEtAl2005}
Christopher~M. Dawson, Andrew~P. Hines, Duncan Mortimer, Henry~L. Haselgrove, Michael~A. Nielsen, and Tobias~J. Osborne.
\newblock Quantum computing and polynomial equations over the finite field {$\mathbb{Z}_2$}.
\newblock \emph{Quantum Information and Computation}, 5\penalty0 (2):\penalty0 102--112, 2005.
\newblock URL \url{https://dl.acm.org/doi/abs/10.5555/2011626.2011628}.

\bibitem[de~Colnet et~al.(2026)de~Colnet, Geerts, Hai, Laarman, Lee, and P{\'e}rez]{deColnet+2026}
Alexis de~Colnet, Floris Geerts, Rihan Hai, Alfons Laarman, Joon~Hyung Lee, and Guillermo~A. P{\'e}rez.
\newblock Quadratic sums-of-powers for fixed-parameter tractable quantum-circuit simulation.
\newblock arXiv preprint arXiv:2605.29944, 2026.
\newblock URL \url{https://arxiv.org/abs/2605.29944v2}.

\bibitem[Dechter(1999)]{Dechter1999}
Rina Dechter.
\newblock Bucket elimination: A unifying framework for reasoning.
\newblock \emph{Artificial Intelligence}, 113\penalty0 (1--2):\penalty0 41--85, 1999.
\newblock URL \url{https://doi.org/10.1016/S0004-3702(99)00059-4}.

\bibitem[Dehaene and De~Moor(2003)]{DehaeneDeMoor2003}
Jeroen Dehaene and Bart De~Moor.
\newblock {Clifford} group, stabilizer states, and linear and quadratic operations over {GF(2)}.
\newblock \emph{Physical Review A}, 68:\penalty0 042318, 2003.
\newblock URL \url{https://doi.org/10.1103/PhysRevA.68.042318}.

\bibitem[Dudek et~al.(2020)Dudek, Dueñas-Osorio, and Vardi]{Dudek+2020}
Jeffrey~M. Dudek, Leonardo Dueñas-Osorio, and Moshe~Y. Vardi.
\newblock Efficient contraction of large tensor networks for weighted model counting through graph decompositions.
\newblock arXiv preprint arXiv:1908.04381, 2020.
\newblock URL \url{https://arxiv.org/abs/1908.04381}.

\bibitem[Geerts and Hai(2025)]{Geerts+2025}
Floris Geerts and Rihan Hai.
\newblock {QC} meet {CQ}: Quantum conjunctive queries.
\newblock In \emph{Proceedings of the 2nd Workshop on Quantum Computing and Quantum-Inspired Technology for Data-Intensive Systems and Applications}, Q-Data '25, page~25. Association for Computing Machinery, 2025.
\newblock URL \url{https://doi.org/10.1145/3736393.3736696}.

\bibitem[Gottesman(1998)]{Gottesman1998}
Daniel Gottesman.
\newblock Theory of fault-tolerant quantum computation.
\newblock \emph{Physical Review A}, 57\penalty0 (1):\penalty0 127--137, January 1998.
\newblock URL \url{https://link.aps.org/doi/10.1103/PhysRevA.57.127}.

\bibitem[Gray and Kourtis(2021)]{GrayKourtis2021}
Johnnie Gray and Stefanos Kourtis.
\newblock Hyper-optimized tensor network contraction.
\newblock \emph{Quantum}, 5:\penalty0 410, 2021.
\newblock URL \url{https://doi.org/10.22331/q-2021-03-15-410}.

\bibitem[Green et~al.(2007)Green, Karvounarakis, and Tannen]{GreenKarvounarakisTannen2007}
Todd~J. Green, Grigoris Karvounarakis, and Val Tannen.
\newblock Provenance semirings.
\newblock In \emph{Proceedings of the 26th {ACM SIGMOD-SIGACT-SIGART} Symposium on Principles of Database Systems ({PODS})}, pages 31--40. Association for Computing Machinery, 2007.
\newblock URL \url{https://doi.org/10.1145/1265530.1265535}.

\bibitem[Grohe and Marx(2014)]{DBLP:journals/talg/GroheM14}
Martin Grohe and D{\'{a}}niel Marx.
\newblock Constraint solving via fractional edge covers.
\newblock \emph{{ACM} Trans. Algorithms}, 11\penalty0 (1):\penalty0 4:1--4:20, 2014.
\newblock URL \url{https://doi.org/10.1145/2636918}.

\bibitem[Guo et~al.(2019)Guo, Liu, Xiong, Xue, Fu, Huang, Qiang, Xu, Liu, Zheng, et~al.]{guo2019general}
Chu Guo, Yong Liu, Min Xiong, Shichuan Xue, Xiang Fu, Anqi Huang, Xiaogang Qiang, Ping Xu, Junhua Liu, Shenggen Zheng, et~al.
\newblock General-purpose quantum circuit simulator with projected entangled-pair states and the quantum supremacy frontier.
\newblock \emph{Physical review letters}, 123\penalty0 (19):\penalty0 190501, 2019.
\newblock URL \url{https://doi.org/10.1103/PhysRevLett.123.190501}.

\bibitem[Hai et~al.(2025)Hai, Hung, Coopmans, Littau, and Geerts]{HaiEtAl2025QuantumDataManagement}
Rihan Hai, Shih-Han Hung, Tim Coopmans, Tim Littau, and Floris Geerts.
\newblock Quantum data management in the {NISQ} era.
\newblock \emph{Proceedings of the VLDB Endowment}, 18\penalty0 (6):\penalty0 1720--1729, 2025.
\newblock URL \url{https://doi.org/10.14778/3725688.3725701}.

\bibitem[Harrow and Montanaro(2017)]{HarrowMontalee2017}
Aram~W. Harrow and Ashley Montanaro.
\newblock Quantum computational supremacy.
\newblock \emph{Nature}, 549\penalty0 (7671):\penalty0 203--209, 2017.
\newblock URL \url{https://doi.org/10.1038/nature23458}.

\bibitem[Ilinescu et~al.(2026)Ilinescu, Patwardhan, and Hai]{IlinescuEtAl2026InferQ}
Andrei Ilinescu, Aadi Patwardhan, and Rihan Hai.
\newblock {InferQ}: A database-oriented benchmark for quantum circuits simulation.
\newblock \emph{Proceedings of the ACM on Management of Data}, 4\penalty0 (4):\penalty0 1--26, 2026.
\newblock URL \url{https://doi.org/10.1145/3837116}.

\bibitem[Kissinger et~al.(2022)Kissinger, van~de Wetering, and Vilmart]{KissingerEtAl2022}
Aleks Kissinger, John van~de Wetering, and Renaud Vilmart.
\newblock Classical simulation of quantum circuits with partial and graphical stabiliser decompositions.
\newblock In \emph{Proceedings of the 17th Conference on the Theory of Quantum Computation, Communication and Cryptography ({TQC})}, volume 232 of \emph{Leibniz International Proceedings in Informatics (LIPIcs)}, pages 5:1--5:13. Schloss Dagstuhl -- Leibniz-Zentrum f{\"u}r Informatik, 2022.
\newblock URL \url{https://doi.org/10.4230/LIPIcs.TQC.2022.5}.

\bibitem[Kschischang et~al.(2001)Kschischang, Frey, and Loeliger]{Kschischang+2001}
Frank~R. Kschischang, Brendan~J. Frey, and Hans-Andrea Loeliger.
\newblock Factor graphs and the sum-product algorithm.
\newblock \emph{IEEE Transactions on Information Theory}, 47\penalty0 (2):\penalty0 498--519, 2001.
\newblock URL \url{https://doi.org/10.1109/18.910572}.

\bibitem[Kuyanov and Kissinger(2026)]{KuyanovKissinger2026}
Fedor Kuyanov and Aleks Kissinger.
\newblock Efficient classical simulation of low-rank-width quantum circuits using {ZX}-calculus.
\newblock arXiv preprint arXiv:2603.06764, 2026.
\newblock URL \url{https://arxiv.org/abs/2603.06764}.

\bibitem[Lanzinger(2022)]{DBLP:conf/pods/Lanzinger22}
Matthias Lanzinger.
\newblock The complexity of conjunctive queries with degree 2.
\newblock In \emph{Proceedings of the 41st ACM SIGMOD-SIGACT-SIGAI Symposium on Principles of Database Systems ({PODS})}, pages 91--102. Association for Computing Machinery, 2022.
\newblock URL \url{https://doi.org/10.1145/3517804.3524152}.

\bibitem[Lee et~al.(2025)Lee, Choi, Ohm, and Lee]{lee2025scalable}
Sung-Bin~B Lee, Hee~Ryang Choi, Daniel~Donghyon Ohm, and Seung-Sup~B Lee.
\newblock Scalable projected entangled-pair state representation of random quantum circuit states.
\newblock \emph{Physical Review Research}, 7\penalty0 (3):\penalty0 033252, 2025.
\newblock URL \url{https://doi.org/10.1103/rzgm-cywf}.

\bibitem[Littau and Hai(2025)]{LittauHai2025Qymera}
Tim Littau and Rihan Hai.
\newblock {Qymera}: Simulating quantum circuits using {RDBMS}.
\newblock In \emph{Companion of the 2025 International Conference on Management of Data}, SIGMOD-Companion '25, pages 179--182. Association for Computing Machinery, 2025.
\newblock URL \url{https://doi.org/10.1145/3722212.3725126}.

\bibitem[Littau and Hai(2026)]{LittauHai2026QuWARP}
Tim Littau and Rihan Hai.
\newblock {QuWARP}: A workload-aware reuse planner for simulating quantum circuits.
\newblock arXiv preprint arXiv:2609.23664, 2026.
\newblock URL \url{https://arxiv.org/abs/2609.23664}.
\newblock Accepted at IEEE QCE 2026.

\bibitem[Markov and Shi(2008)]{MarkovShi2008}
Igor~L. Markov and Yaoyun Shi.
\newblock Simulating quantum computation by contracting tensor networks.
\newblock \emph{SIAM Journal on Computing}, 38\penalty0 (3):\penalty0 963--981, 2008.
\newblock URL \url{https://doi.org/10.1137/050644756}.

\bibitem[Mei et~al.(2024{\natexlab{a}})Mei, Bonsangue, and Laarman]{MeiEtAl2024}
Jingyi Mei, Marcello Bonsangue, and Alfons Laarman.
\newblock Simulating quantum circuits by model counting.
\newblock In \emph{Computer Aided Verification ({CAV})}, volume 14683 of \emph{Lecture Notes in Computer Science}, pages 555--578. Springer, 2024{\natexlab{a}}.
\newblock URL \url{https://doi.org/10.1007/978-3-031-65633-0_25}.

\bibitem[Mei et~al.(2024{\natexlab{b}})Mei, Coopmans, Bonsangue, and Laarman]{DBLP:conf/ijcar/MeiCBL24}
Jingyi Mei, Tim Coopmans, Marcello~M. Bonsangue, and Alfons Laarman.
\newblock Equivalence checking of quantum circuits by model counting.
\newblock In Christoph Benzm{\"{u}}ller, Marijn J.~H. Heule, and Renate~A. Schmidt, editors, \emph{Automated Reasoning - 12th International Joint Conference, {IJCAR} 2024, Nancy, France, July 3-6, 2024, Proceedings, Part {II}}, volume 14740 of \emph{Lecture Notes in Computer Science}, pages 401--421. Springer, 2024{\natexlab{b}}.
\newblock URL \url{https://doi.org/10.1007/978-3-031-63501-4_21}.

\bibitem[Ngo et~al.(2018)Ngo, Porat, R{\'e}, and Rudra]{NgoPoratReRudra2018}
Hung~Q. Ngo, Ely Porat, Christopher R{\'e}, and Atri Rudra.
\newblock Worst-case optimal join algorithms.
\newblock \emph{Journal of the ACM}, 65\penalty0 (3):\penalty0 16:1--16:40, 2018.
\newblock URL \url{https://doi.org/10.1145/3180143}.

\bibitem[Oum(2008)]{oum2008rank}
Sang-il Oum.
\newblock Rank-width is less than or equal to branch-width.
\newblock \emph{Journal of Graph Theory}, 57\penalty0 (3):\penalty0 239--244, 2008.
\newblock URL \url{https://doi.org/10.1002/jgt.20280}.

\bibitem[Oum and Seymour(2006)]{OumSeymour2006}
Sang{-}il Oum and Paul~D. Seymour.
\newblock Approximating clique-width and branch-width.
\newblock \emph{Journal of Combinatorial Theory, Series B}, 96\penalty0 (4):\penalty0 514--528, 2006.
\newblock URL \url{https://doi.org/10.1016/j.jctb.2005.10.006}.

\bibitem[Pan et~al.(2022)Pan, Chen, and Zhang]{PanChenZhang22}
Feng Pan, Keyang Chen, and Pan Zhang.
\newblock Solving the sampling problem of the {Sycamore} quantum circuits.
\newblock \emph{Physical Review Letters}, 129\penalty0 (9):\penalty0 090502, 2022.
\newblock URL \url{https://doi.org/10.1103/PhysRevLett.129.090502}.

\bibitem[Pang and Solomonik(2025)]{pang2025}
Yuchen Pang and Edgar Solomonik.
\newblock Simulating {C}lifford circuits with {Gaussian} elimination.
\newblock arXiv preprint arXiv:2511.06127, 2025.
\newblock URL \url{https://arxiv.org/abs/2511.06127}.

\bibitem[Quist et~al.(2026)Quist, Farreras, de~Colnet, van~de Wetering, and Laarman]{Quist2026}
Arend-Jan Quist, Marc Farreras, Alexis de~Colnet, John van~de Wetering, and Alfons Laarman.
\newblock From tensor networks to tractable circuits, and back.
\newblock In \emph{Proceedings of the 23rd International Conference on Principles of Knowledge Representation and Reasoning}, KR '26, pages 546--557, 2026.
\newblock URL \url{https://doi.org/10.24963/kr.2026/52}.

\bibitem[Robertson and Seymour(1986)]{DBLP:journals/jal/RobertsonS86}
Neil Robertson and Paul~D. Seymour.
\newblock Graph minors. {II.} algorithmic aspects of tree-width.
\newblock \emph{J. Algorithms}, 7\penalty0 (3):\penalty0 309--322, 1986.
\newblock URL \url{https://doi.org/10.1016/0196-6774(86)90023-4}.

\bibitem[Robeva and Seigal(2019)]{Robeva+2019}
Elina Robeva and Anna Seigal.
\newblock Duality of graphical models and tensor networks.
\newblock \emph{Information and Inference: A Journal of the IMA}, 8\penalty0 (2):\penalty0 273--288, June 2019.
\newblock URL \url{https://doi.org/10.1093/imaiai/iay009}.

\bibitem[Seitz et~al.(2023)Seitz, Medina, Cruz, Huang, and Mendl]{seitz2023simulating}
Philipp Seitz, Ismael Medina, Esther Cruz, Qunsheng Huang, and Christian~B Mendl.
\newblock Simulating quantum circuits using tree tensor networks.
\newblock \emph{Quantum}, 7:\penalty0 964, 2023.
\newblock URL \url{https://doi.org/10.22331/q-2023-03-30-964}.

\bibitem[Shchesnovich(2021)]{Shchesnovic2021}
Valery Shchesnovich.
\newblock Distinguishing noisy boson sampling from classical simulations.
\newblock \emph{Quantum}, 5:\penalty0 423, 2021.
\newblock URL \url{https://doi.org/10.22331/q-2021-03-29-423}.

\bibitem[Shi et~al.(2006)Shi, Duan, and Vidal]{shi2006classical}
Yaoyun Shi, Luming Duan, and Guifre Vidal.
\newblock Classical simulation of quantum many-body systems with a tree tensor network.
\newblock \emph{Physical Review A—Atomic, Molecular, and Optical Physics}, 74\penalty0 (2):\penalty0 022320, 2006.
\newblock URL \url{https://doi.org/10.1103/PhysRevA.74.022320}.

\bibitem[Singh(2025)]{Singh2025}
Harbaksh Singh.
\newblock Future directions and challenges, quantum supremacy, and beyond.
\newblock In Mohammad Hammoudeh, Clinton~M. Firth, Harbaksh Singh, Christoph Capellaro, and Mohamed Al~Kuwaiti, editors, \emph{Quantum Technology Applications, Impact, and Future Challenges}, pages 141--162. CRC Press, 2025.
\newblock URL \url{https://doi.org/10.1201/9781003537243-9}.

\bibitem[Van~den Nest(2010)]{VandenNest10}
Maarten Van~den Nest.
\newblock Classical simulation of quantum computation, the {G}ottesman--{K}nill theorem, and slightly beyond.
\newblock \emph{Quantum Information and Computation}, 10\penalty0 (3\&4):\penalty0 258--271, 2010.
\newblock URL \url{https://doi.org/10.26421/QIC10.3-4-6}.

\bibitem[Van~den Nest et~al.(2007)Van~den Nest, D{\"u}r, Vidal, and Briegel]{VanDenNestEtAl2007}
Maarten Van~den Nest, Wolfgang D{\"u}r, Guifr{\'e} Vidal, and Hans~J. Briegel.
\newblock Classical simulation versus universality in measurement-based quantum computation.
\newblock \emph{Physical Review A}, 75\penalty0 (1):\penalty0 012337, 2007.
\newblock URL \url{https://doi.org/10.1103/PhysRevA.75.012337}.

\bibitem[Vidal(2003)]{Vidal2003}
Guifr{\'e} Vidal.
\newblock Efficient classical simulation of slightly entangled quantum computations.
\newblock \emph{Physical Review Letters}, 91:\penalty0 147902, 2003.
\newblock URL \url{https://doi.org/10.1103/PhysRevLett.91.147902}.

\bibitem[Villalonga et~al.(2019)Villalonga, Boixo, Nelson, Henze, Rieffel, Biswas, and Mandr{\`a}]{villalonga2019flexible}
Benjamin Villalonga, Sergio Boixo, Bron Nelson, Christopher Henze, Eleanor Rieffel, Rupak Biswas, and Salvatore Mandr{\`a}.
\newblock A flexible high-performance simulator for verifying and benchmarking quantum circuits implemented on real hardware.
\newblock \emph{npj Quantum Information}, 5\penalty0 (1):\penalty0 86, 2019.
\newblock URL \url{https://doi.org/10.1038/s41534-019-0196-1}.

\bibitem[Vincent et~al.(2022)Vincent, O'Riordan, Andrenkov, Brown, Killoran, Qi, and Dhand]{VincentEtAl2022}
Trevor Vincent, Lee~J. O'Riordan, Mikhail Andrenkov, Jack Brown, Nathan Killoran, Haoyu Qi, and Ish Dhand.
\newblock {Jet}: Fast quantum circuit simulations with parallel task-based tensor-network contraction.
\newblock \emph{Quantum}, 6:\penalty0 709, 2022.
\newblock URL \url{https://doi.org/10.22331/q-2022-05-09-709}.

\bibitem[Vinkhuijzen et~al.(2023)Vinkhuijzen, Coopmans, Elkouss, Dunjko, and Laarman]{VinkhuijzenEtAl2023}
Lieuwe Vinkhuijzen, Tim Coopmans, David Elkouss, Vedran Dunjko, and Alfons Laarman.
\newblock {LIMDD}: A decision diagram for simulation of quantum computing including stabilizer states.
\newblock \emph{Quantum}, 7:\penalty0 1108, 2023.
\newblock URL \url{https://doi.org/10.22331/q-2023-09-11-1108}.

\bibitem[Vinkhuijzen et~al.(2026)Vinkhuijzen, Coopmans, and Laarman]{Vinkhuijzen2026}
Lieuwe Vinkhuijzen, Tim Coopmans, and Alfons Laarman.
\newblock A knowledge compilation map for quantum information.
\newblock \emph{Proceedings of the AAAI Conference on Artificial Intelligence}, 40\penalty0 (23):\penalty0 19406--19414, 2026.
\newblock URL \url{https://doi.org/10.1609/aaai.v40i23.39018}.

\bibitem[Wang et~al.(2025)Wang, Cheng, Yuan, and Ji]{WangEtAl2025}
Ziyuan Wang, Bin Cheng, Longxiang Yuan, and Zhengfeng Ji.
\newblock {FeynmanDD}: Quantum circuit analysis with classical decision diagrams.
\newblock In \emph{Computer Aided Verification ({CAV})}, volume 15934 of \emph{Lecture Notes in Computer Science}, pages 28--52. Springer, 2025.
\newblock URL \url{https://doi.org/10.1007/978-3-031-98685-7_2}.

\bibitem[Ye and Lim(2018)]{YeLim18}
Ke~Ye and Lek-Heng Lim.
\newblock Tensor network ranks.
\newblock arXiv preprint arXiv:1801.02662, 2018.
\newblock URL \url{https://arxiv.org/abs/1801.02662}.

\bibitem[Zulehner and Wille(2019)]{ZulehnerWille2018}
Alwin Zulehner and Robert Wille.
\newblock Advanced simulation of quantum computations.
\newblock \emph{IEEE Transactions on Computer-Aided Design of Integrated Circuits and Systems}, 38\penalty0 (5):\penalty0 848--859, 2019.
\newblock URL \url{https://doi.org/10.1109/TCAD.2018.2834427}.

\end{thebibliography}
\end{document}